\documentclass[10pt]{article}
\usepackage{graphicx}
\usepackage[T1]{fontenc}
\usepackage[utf8]{inputenc}
\usepackage{amssymb}
\usepackage{amsfonts}
\usepackage{dsfont}
\usepackage{mathtools}
\usepackage{amsthm}
\usepackage{amsmath}
\usepackage{textcomp}
\usepackage{eurosym}
\usepackage{chngcntr}
\counterwithout{table}{section}
\counterwithout{table}{subsection}

\usepackage{geometry}
\newtheorem{lemme}{Lemma}
\newtheorem{thm}{Theorem}
\newtheorem{prop}{Proposition}

\newtheorem{rem}{Remark}

\begin{document}

\title{Size matters for OTC market makers: general results and dimensionality reduction techniques\thanks{The authors would like to thank Yves Achdou (Université Paris Diderot), Bastien Baldacci (Ecole Polytechnique), Guillaume Bioche (BNP Paribas), Iuliia Manziuk (Université Paris 1 Panthéon-Sorbonne and Ecole Polytechnique), Domingo Puertas Trillo (BNP Paribas), Marie-Claire Quenez (Université Paris Diderot), Andrei Serjantov (BNP Paribas), and Goncalo Simoes (BNP Paribas) for the numerous and insightful discussions they had with them on the subject. Also, two anonymous referees deserve warm thanks as the paper was substantially improved following their remarks.}}

\author{Philippe \textsc{Bergault}\footnote{Université Paris 1 Panthéon-Sorbonne, Centre d'Economie de la Sorbonne, 106 Boulevard de l'Hôpital, 75642 Paris Cedex 13, France, philippe.bergault@etu.univ-paris1.fr} \and Olivier \textsc{Guéant}\footnote{Université Paris 1 Panthéon-Sorbonne, Centre d'Economie de la Sorbonne, 106 Boulevard de l'Hôpital, 75642 Paris Cedex 13, France, olivier.gueant@univ-paris1.fr \textit{Corresponding author.}}}
\date{}
\maketitle

\abstract{\noindent In most OTC markets, a small number of market makers provide liquidity to other market participants. More precisely, for a list of assets, they set prices at which they agree to buy and sell. Market makers face therefore an interesting optimization problem: they need to choose bid and ask prices for making money while mitigating the risk associated with holding inventory in a volatile market. Many market making models have been proposed in the academic literature, most of them dealing with single-asset market making whereas market makers are usually in charge of a long list of assets. The rare models tackling multi-asset market making suffer however from the curse of dimensionality when it comes to the numerical approximation of the optimal quotes. The goal of this paper is to propose a dimensionality reduction technique to address multi-asset market making by using a factor model. Moreover, we generalize existing market making models by the addition of an important feature: the existence of different transaction sizes and the possibility for the market makers in OTC markets to answer different prices to requests with different sizes.}

\vspace{8mm}

\setlength\parindent{0pt}


\vspace{5mm}

\section{Introduction}

The electronification of financial markets has changed the traditional role played by market makers. This is evident in the case of most order-driven markets, such as many stock markets,  where the traditional market makers in charge of maintaining fair and orderly markets now often compete with high-frequency market making companies. Surprisingly maybe, many OTC markets organized around dealers have also undergone upheaval linked to electronification over the last ten years. This is the case of the corporate bond markets on both sides of the Atlantic ocean where the electronification process is dominated by Multi-dealer-to-client (MD2C) platforms enabling clients to send the same request for quote (RFQ) to several dealers simultaneously and therefore instantly put them into competition with one another. Electronification is also in progress inside investment banks as most of them replace their traders by algorithms to be able to provide clients with quotes for a large set of assets and automate their market making business, at least for small tickets.\\

Building market making algorithms is a difficult task as the optimization problem faced by a market maker involves both static and dynamic components. A market maker faces indeed a first (static) trade-off: high margin and low volume versus low margin and high volume. A market maker quoting a large bid-ask spread (with no skew) trades indeed rarely, but each trade is associated with large Mark-to-Market (MtM) gain. Conversely, a market maker who quotes a narrow bid-ask spread (with no skew) trades often, but each trade is associated with a small MtM gain. In addition to this simple static trade-off, market makers face a dynamic problem: in a volatile market, they must quote in a dynamic way to mitigate their market risk exposure and, in particular, skew their quotes as a function of their inventory. For example, a single-asset market maker with a long inventory should price in a conservative manner on the bid side and rather aggressively on the ask side, if she wants -- a reasonable behaviour -- to decrease her probability to buy and increase her probability to sell.\\

The optimization problem faced by market makers has been addressed in a long list of academic papers. The first two references commonly cited in the market making literature are two economic papers: Grossman and Miller \cite{grossman1988liquidity} and Ho and Stoll \cite{ho1981optimal}. If the former is a classic from a theoretical point of view, the latter was revived in 2008 by Avellaneda and Stoikov \cite{avellaneda2008high} to build the first practical model of single-asset market making. Since then, many models have been proposed, most of them to tackle the same problem of single-asset market making. For instance, \cite{gueant2013dealing} provides a rigourous analysis of the stochastic optimal control problem introduced by Avellaneda and Stoikov and proves that the problem boils down to a system of linear ordinary differential equations (ODE) in the case of exponential intensity functions. Cartea et al. (\cite{cartea2017algorithmic,cartea2015algorithmic,cartea2014buy}) contributed a lot to the literature and added many features to the initial models:  alpha signals, ambiguity aversion, etc. They also considered a different objective function: a risk-adjusted expectation instead of the Von Neumann-Morgenstern expected utility of \cite{avellaneda2008high} and \cite{gueant2013dealing}. Multi-asset market making has been considered in \cite{gueant2016financial,gueant2017optimal} for both kinds of objective functions and the author shows that the problem boils down, for general intensity functions, to solving a system of (\emph{a priori} nonlinear) ODEs. Most of the above models are well suited to tackle market making in OTC markets or in order-driven markets when the tick/spread ratio is large. For major stock markets or for some foreign-exchange platforms, other models are better suited such as those of Guilbaud and Pham who really took the microstructure into account (see \cite{guilbaud2013optimal,guilbaud2015optimal}).\footnote{Option market making has also been addressed, see for instance \cite{bergault2019algorithmic, el2015stochastic, stoikov2009option}.}\\

In spite of a large and growing literature on market making, several problems are rarely addressed. A first example is that of trade sizes: in markets organized around requests for quotes, quotes can and should depend on the size of the requests. A second and more general problem is that of the numerical approximation of the optimal quotes. If optimal quotes can theoretically be computed through the solution of a system of ODEs, the size of that system (which grows exponentially with the number of assets) prevents any concrete computation with grid methods when it comes to portfolios with more than 4 or 5 assets. To our knowledge, the only attempt to approximate the solution of the Hamilton-Jacobi equations associated with market making models in high dimension is \cite{gueant2019deep} in which the authors propose a method -- inspired by reinforcement learning techniques -- that uses neural networks instead of grids.\\

In this paper, our goal is twofold. Our first goal is to generalize existing models to introduce a distribution of trade size. This extension is not straightforward as the optimal controls cannot be modeled anymore with real-valued stochastic processes, but must instead be modeled with predictable maps. A consequence, in terms of mathematics, is that the problem does not anymore boil down to a finite system of ODEs but instead to an integro-differential equation of the Hamilton-Jacobi type that can be regarded as an ordinary differential equation in an infinite-dimensional space. Our second goal is to propose a numerical method for approximating the optimal bid and ask quotes of a market maker over a large universe of assets. For that purpose, we show that the real dimension of the problem is not that of the number of assets, but rather that of the rank of the correlation matrix of asset prices. Then, by using a factor model, we show how to approximate the optimal quotes of a market maker. Indeed, if market risk is projected on a low-dimensional space of factors, solving the market making problem boils down to solving a low-dimensional Hamilton-Jacobi equation. In particular, if the number of factors is lower than 3, grid methods can be applied independently of the number of assets. Furthermore, we suggest a Monte-Carlo method to approximate the influence of the residual risk not taken into account when the risk is projected on the space of factors.\\

In Section~2, we present our market making model with distributed request sizes. We characterize the value function associated with the stochastic optimal control problem as the solution of an integro-differential equation of the Hamilton-Jacobi type by using ODE techniques in a well-chosen infinite-dimensional space and a verification argument. We subsequently provide expressions for the optimal quotes as a function of time, inventory, and request size. In Section~3, we show how the equations can be simplified when the dependence structure between the prices of the different assets can be modeled by risk factors. We then show how this simplification leads to an approximation that helps to tackle the curse of dimensionality by solving a low-dimensional equation on a grid. We also explain how Monte-Carlo simulations could be used to account for the part of the risk not accounted by the factors. We apply these techniques in Section~4 to portfolios of 2 and 30 bonds and discuss the results.\\

\section{Market making with marked point processes}\label{genpb}

In all this paper, we consider a filtered probability space $\left( \Omega, \mathcal{F},\mathbb{P}; \mathbb{F}= (\mathcal{F}_{t})_{t\geq 0} \right)$ satisfying the usual conditions. We assume this probability space to be large enough to support all the processes we introduce.\\

\subsection{Modeling framework and notations}

We consider a market maker in charge of $d$ assets. For $i \in \{1,\ldots,d\}$, the reference price of asset $i$ is modeled by a process $(S^{i}_{t})_{t\ge 0}$ with the following dynamics:
\begin{equation*}
dS^{i}_{t} = \sigma^{i}dW^{i}_{t},
\end{equation*}
with $S^{i}_{0}$ given, $\sigma^{i}>0$, and $\left((W^{1}_{t},\ldots,W^{d}_{t})'\right)_{t\geq 0}$ a $d$-dimensional\footnote{The sign $'$ designates the transpose operator. It transforms here a line vector into a column vector.} Brownian motion with correlation matrix $(\rho^{i,j})_{1\leq i,j\leq d}$, adapted to the filtration $(\mathcal{F}_{t})_{t\geq 0}$. We denote by $\Sigma = \left( \rho^{i,j}\sigma^{i} \sigma^{j} \right)_{1\leq i,j \leq d}$ the variance-covariance matrix associated with the process $(S_{t})_{t\geq 0} = \left((S^{1}_{t},\ldots,S^{d}_{t})'\right)_{t\geq 0}$.\\

Those assets are traded with requests for quote (RFQ): the market maker first receives a request for quote from a client wishing to buy or sell a given asset and she then proposes a price to the client who finally decides whether she accepts to trade at that price or not.\\

At any time, the market maker must be ready to propose bid and ask quotes to buy and sell any of the $d$ assets. These bid and ask quotes depend on the size $z\in \mathbb{R}_{+}^{*}$ of the RFQ (in all this paper, we use the notation $\mathbb R_+^* := (0,+\infty)$). For a given asset $i$, they are modeled by maps $S^{i,b},S^{i,a} : \Omega \times [0,T] \times \mathbb{R}_{+}^{*}\rightarrow \mathbb{R}$ which are $\mathcal{P} \otimes \mathcal{B}(\mathbb{R}_{+}^{*})$-measurable, where $\mathcal{P}$ denotes the $\sigma$-algebra of $\mathbb{F}$-predictable subsets of $\Omega \times[0,T]$ and $\mathcal{B}(\mathbb{R}_{+}^{*})$ denotes the Borelian sets of $\mathbb{R}_{+}^{*}$.\\

For each $i\ \in\ \{1,\ldots,d\}$, we introduce $J^{i,b}(dt,dz)$ and $J^{i,
a}(dt,dz)$ two \textit{càdlàg} $\mathbb{R}_{+}$-marked point processes.\footnote{These processes are explicitly constructed in the Appendix. Note that in our model, as in most real OTC markets, there are no simultaneous RFQs in multiple assets.}\\

For  $i \in \{1,\ldots,d\}$, we denote by $\left(\nu^{i,b}_{t}(dz)\right)_{t\geq 0}$ and $\left(\nu^{i,a}_{t}(dz)\right)_{t\geq 0}$ the intensity kernels of $J^{i,b}(dt,dz)$ and $J^{i,
a}(dt,dz)$, respectively. In addition, we assume that $\left(\nu^{i,b}_{t}(dz)\right)_{t\geq 0}$ and $\left(\nu^{i,a}_{t}(dz)\right)_{t\geq 0}$ verify:
\begin{equation}
\begin{split}
    \nu^{i,b}_{t}(dz)\ =\ \Lambda^{i,b}(\delta^{i,b}(t,z))\mu^{i,b}(dz), \\
    \nu^{i,a}_{t}(dz)\ =\ \Lambda^{i,a}(\delta^{i,a}(t,z))\mu^{i,a}(dz), \nonumber
\end{split}
\end{equation}
where for all $i \in \{1,\ldots,d\}$, $\left(\mu^{i,b}, \mu^{i,a}\right)$ is a couple of probability measures on $\mathbb{R}_{+}^{*}$, $\delta^{i,b}(t,z) = S^{i}_{t} - S^{i,b}(t,z)$, $\delta^{i,a}(t,z) = S^{i,a}(t,z) -  S^{i}_{t}$, and $\left(\Lambda^{i,b}, \Lambda^{i,a}\right)$ is a couple of functions satisfying the following hypotheses (H):
\begin{itemize}
    \item $\Lambda^{i,b}$ and $\Lambda^{i,a}$ are twice continuously differentiable,
    \item $\Lambda^{i,b}$ and $\Lambda^{i,a}$ are decreasing, with $\forall \delta \in \mathbb{R}$, ${\Lambda^{i,b}}'(\delta)<0$ and ${\Lambda^{i,a}}'(\delta)<0$,
    \item $\underset{\delta \rightarrow +\infty}{\lim}\Lambda^{i,b}(\delta)=0$ and $\underset{\delta \rightarrow +\infty}{\lim} \Lambda^{i,a}(\delta)=0$,
    \item $\underset{\delta \in \mathbb{R}}{\sup}  \frac{\Lambda^{i,b}(\delta){\Lambda^{i,b}}''(\delta)}{\left( {\Lambda^{i,b}}'(\delta) \right)^{2}}  < 2$ and $\underset{\delta \in \mathbb{R}}{\sup}  \frac{\Lambda^{i,a}(\delta){\Lambda^{i,a}}''(\delta)}{\left( {\Lambda^{i,a}}'(\delta) \right)^{2}}  < 2$. \\
\end{itemize}

For all $i \in \{1, \ldots, d\}$, $J^{i,b}(dt,dz)$ and $J^{i,a}(dt,dz)$ model respectively the volumes of transactions at the bid and at the ask for asset $i$. The inventory of the market maker, modeled by the $d$-dimensional process $(q_{t})_{t\geq 0} = \left((q^{1}_{t},\ldots,q^{d}_{t})'\right)_{t\geq 0}$, has therefore the following dynamics:
\begin{equation*}
\forall i \in \{1,\ldots,d\},\ dq^{i}_{t}\ =\ \int_{\mathbb{R}_{+}^{*}} z J^{i,b}(dt,dz) - \int_{\mathbb{R}_{+}^{*}} z J^{i,a}(dt,dz),
\end{equation*}
with $q_{0}$ given.\\

\begin{rem}
For a given asset $i$, $\Lambda^{i,.}$ typically has the form $\Lambda^{i,.}(\delta) = \lambda^{i,.}_{RFQ} f^{i,.}(\delta)$, where $\lambda^{i,.}_{RFQ}$ is the (constant) intensity of arrival of requests for quote and $f^{i,.}(\delta)$ gives the probability that a request will result in a transaction given the quote $\delta$ proposed by the market maker. Furthermore, $\mu^{i, .}$ should be seen as the distribution of sizes.\\
\end{rem}

Finally, the process $(X_{t})_{t\geq 0}$ modeling the market maker's cash account has the dynamics
\begin{eqnarray*}
dX_{t} & = & \sum_{i=1}^{d} \int_{\mathbb{R}_{+}^{*}}S^{i,a}(t,z)zJ^{i,a}(dt,dz) - \sum_{i=1}^{d} \int_{\mathbb{R}_{+}^{*}}S^{i,b}(t,z)zJ^{i,b}(dt,dz)\\
& =& \sum_{i=1}^{d} \int_{\mathbb{R}_{+}^{*}}(S^{i}_{t}+\delta^{i,a}(t,z))zJ^{i,a}(dt,dz) - \sum_{i=1}^{d} \int_{\mathbb{R}_{+}^{*}}(S^{i}_{t}-\delta^{i,b}(t,z))zJ^{i,b}(dt,dz)\\
& =& \sum_{i=1}^{d} \int_{\mathbb{R}_{+}^{*}} \left(\delta^{i,b}(t,z)zJ^{i,b}(dt,dz) + \delta^{i,a}(t,z)zJ^{i,a}(dt,dz) \right) - \sum_{i=1}^{d} S^{i}_{t} dq^{i}_{t}.\\
\end{eqnarray*}

We fix $\delta_{\infty} \ge 0$ and define the set $\mathcal{A}$ of admissible controls\footnote{We introduce here a unique lower bound for the quotes, independently of the asset, the side, the size, and the time. Generalizations are straightforward.} by
\begin{eqnarray*}
\mathcal{A} &=& \Bigg\lbrace\left. \delta = \left(\delta^{i,b},\delta^{i,a}\right)_{1 \le i \le d} : \Omega \times [0,T] \times \mathbb{R}_{+}^{*} \mapsto \mathbb{R}^{2d} \right| \delta \text{ is } \mathcal{P} \otimes \mathcal{B}(\mathbb{R}_{+}^{*})-\text{measurable},\\
&&\forall i \in \{1, \ldots, d\}, \delta^{i,b}(t,z) \geq -\delta_{\infty}\ \mathbb P \otimes dt \otimes \mu^{i,b} \text{ a.e.} \text{ and } \delta^{i,a}(t,z) \geq -\delta_{\infty}\ \mathbb P \otimes dt \otimes \mu^{i,a} \text{ a.e.}  \Bigg\rbrace.
\end{eqnarray*}

As proved in \cite{gueant2017optimal}, under assumptions (H), the functions $\delta \in \mathbb{R} \mapsto \delta \Lambda^{i,b}(\delta)$ and $\delta \in \mathbb{R} \mapsto \delta \Lambda^{i,a}(\delta)$ have a unique maximum on $\mathbb{R}$. It is also easy to see that on $[-\delta_{\infty},+\infty)$, they are bounded from below by $-\delta_{\infty}\Lambda^{i,b}(-\delta_{\infty})$ and $-\delta_{\infty}\Lambda^{i,a}(-\delta_{\infty})$, respectively.\\

For two given continuous penalty functions $\psi:\mathbb{R}^{d} \rightarrow \mathbb{R}_{+}$ and $\ell_{d}:\mathbb{R}^{d} \rightarrow \mathbb{R}_{+}$, modeling the risk aversion of the market maker, we aim at maximizing the objective function
\begin{equation}
\label{MMpb}
\mathbb{E} \left[ X_{T} + \sum_{i=1}^{d} q^{i}_{T}S^{i}_{T} - \ell_{d}(q_{T}) - \int_{0}^{T} \psi(q_{t}) dt \right],
\end{equation}
over the set $\mathcal{A}$ of admissible controls.\\

\begin{rem}
\label{psi_forms}
For instance, we can choose $\psi(q)=\frac{\gamma}{2} q'\Sigma q$ or $\psi(q) = \gamma \sqrt{q'\Sigma q}$ (for $\gamma>0$) and $\ell_{d}(q) = 0$, $\ell_{d}(q) = \frac{\zeta}{2} q'\Sigma q$ or $\ell_{d}(q) = \zeta \sqrt{ q'\Sigma q}$ (for $\zeta>0$), as done in \cite{cartea2015algorithmic}, \cite{cartea2014buy}, \cite{gueant2019deep}, and \cite{gueant2017optimal}.\\
\end{rem}

After applying Itô's formula to $\left(X_{t} + \sum_{i=1}^{d} q^{i}_{t}S^{i}_{t}\right)_{t\geq 0}$ between $0$ and $T$, it is easy to see that the problem is equivalent to maximizing
\begin{equation*}
\mathbb{E}\Bigg[ \int\limits_{0}^{T} \Bigg\lbrace \sum_{i=1}^{d} \int_{\mathbb{R}_{+}^{*}}\left(\delta^{i,b}(t,z)z \Lambda^{i,b}(\delta^{i,b}(t,z))\mu^{i,b}(dz) + \delta^{i,a}(t,z)z \Lambda^{i,a}(\delta^{i,a}(t,z))\mu^{i,a}(dz) \right)
- \psi(q_{t}) \Bigg\rbrace dt -\ell_{d}(q_{T}) \Bigg],
\end{equation*}
over the set of admissible controls $\mathcal{A}$.\\

We introduce the function $\mathcal{J}:[0,T] \times \mathbb{R}^d \times \mathcal{A}\rightarrow \mathbb{R}$ such that, $\forall t \in [0,T]$, $\forall q=(q^{1},\ldots,q^{d})'\in \mathbb{R}^d$, $\forall (\delta^{i})_{i \in \{1, \ldots, d\}}\in \mathcal{A}$,
\begin{equation}
\begin{split}
\mathcal{J}(t,q,(\delta^{i})_{i \in \{1, \ldots, d\}}) =  \mathbb{E}\Bigg[ \int\limits_{t}^{T} \Bigg\lbrace \sum_{i=1}^{d} \int_{\mathbb{R}_{+}^{*}} \bigg(&\delta^{i,b}(s,z)z \Lambda^{i,b}(\delta^{i,b}(s,z))\mu^{i,b}(dz) + \delta^{i,a}(s,z)z \Lambda^{i,a}(\delta^{i,a}(s,z))\mu^{i,a}(dz)\bigg) \\
&- \psi \big(q^{t,q,(\delta^{i})_{i \in \{1, \ldots, d\}}}_{s}\big) \Bigg\rbrace ds -\ell_{d}\big(q^{t,q,(\delta^{i})_{i \in \{1, \ldots, d\}}}_{T}\big) \Bigg], \nonumber
\end{split}
\end{equation}
where $\big(q^{t,q,(\delta^{i})_{i \in \{1, \ldots, d\}}}_{s}\big)_{s\geq t}$ is the inventory process starting in state $q$ at time $t$ and controlled by $(\delta^{i})_{i \in \{1, \ldots, d\}}$.\\

The value function $\theta:[0,T]\times \mathbb{R}^d\rightarrow \mathbb{R}$ of the problem is then defined as follows:
\begin{equation*}
\theta(t,q) = \underset{(\delta^{i})_{i \in \{1, \ldots, d\}} \in \mathcal{A}}{\sup} \mathcal{J}(t,q,(\delta^{i})_{i \in \{1, \ldots, d\}}), \forall (t,q) \in [0,T] \times \mathbb{R}^d.
\end{equation*}

We will show that $\theta$ is the unique (in a large class of functions) classical solution to the following integro-differential Hamilton-Jacobi (HJ) equation:
\begin{eqnarray}
0 &=& \frac{\partial w}{\partial t}(t,q) - \psi(q) + \underset{i=1}{\overset{d}{\sum}}\int_{\mathbb{R}_{+}^{*}}zH^{i,b}\left(\frac{w(t,q) -  w(t,q+ze^{i}) }{z}\right) \mu^{i,b}(dz)\nonumber\\
&&\qquad\qquad + \underset{i=1}{\overset{d}{\sum}}\int_{\mathbb{R}_{+}^{*}}zH^{i,a}\left(\frac{w(t,q) - w(t,q-ze^{i})}{z} \right) \mu^{i,a}(dz),\qquad \forall (t,q) \in [0,T]\times \mathbb{R}^d,
\label{eqn:HJB}
\end{eqnarray}
with terminal condition $w(T,q) = -\ell_{d}(q), \forall q \in \mathbb{R}^d$, where
\begin{equation*}
H^{i,b}:p\in\mathbb{R} \mapsto \underset{\delta \geq -\delta_{\infty}}{\sup} \Lambda^{i,b}(\delta)(\delta-p) \text{ and }H^{i,a}:p\in\mathbb{R} \mapsto  \underset{\delta \geq -\delta_{\infty}}{\sup} \Lambda^{i,a}(\delta)(\delta-p),
\end{equation*}
and where $\left(e^{1},\ldots,e^{d}\right)$ denotes the canonical basis of $\mathbb{R}^{d}$.\\

\subsection{Existence and uniqueness of a solution to \eqref{eqn:HJB}}

\begin{lemme}
\label{lemmH}
$\forall i \in \{1,\ldots,d\}$, $H^{i,b}$ and $H^{i,a}$ are two globally Lipschitz continuously differentiable decreasing functions. Moreover, the supremum in the definition of $H^{i,b}(p)$ (respectively $H^{i,a}(p)$) is reached at a unique $\delta^{i,b*}(p)$ (respectively $\delta^{i,a*}(p)$). Furthermore, $\delta^{i,b*}$ and $\delta^{i,a*}$ are continuous and nondecreasing functions.\\
\end{lemme}

\begin{proof}
We prove the result only for the ask side. The proof is similar for the bid side.\\

Let $i \in \{1,\ldots,d\}$. For $p \in \mathbb{R}$, we define
\begin{equation*}
    h^{i}_{p} : \delta \in \mathbb{R} \longmapsto  \Lambda^{i,a}(\delta)(\delta - p).
\end{equation*}

$h^{i}_{p}$ is a continuously differentiable function, positive for $\delta \in (p,+\infty)$ and nonpositive otherwise. It is easy to prove (see \cite{gueant2017optimal}) that there is a unique maximizer $\tilde{\delta}^{i,a*}(p)$ of $h^{i}_{p}$ on $\mathbb{R}$ characterized by
\begin{equation*}
p = \tilde{\delta}^{i,a*}(p) + \frac{\Lambda^{i,a}(\tilde{\delta}^{i,a*}(p))}{\Lambda^{i,a'}(\tilde{\delta}^{i,a*}(p))}.
\end{equation*}

By the implicit function theorem, $p\in \mathbb{R} \mapsto \tilde{\delta}^{i,a*}(p)$ is continuously differentiable and
\begin{equation*}
\tilde{\delta}^{i,a*'}(p) = \frac{1}{2 - \frac{\Lambda^{i,a}(\tilde{\delta}^{i,a*}(p))\Lambda^{i,a''}(\tilde{\delta}^{i,a*}(p))}{(\Lambda^{i,a'}(\tilde{\delta}^{i,a*}(p)))^{2}}} >0, \forall p\in \mathbb{R}.
\end{equation*}

In particular, $\tilde{\delta}^{i,a*}$ is increasing.\\

We introduce $\tilde{H}^{i,a}:p\in \mathbb{R} \rightarrow \underset{\delta \in \mathbb{R}}{\sup}\ h^{i}_{p}(\delta)$. Then $\forall p \in \mathbb{R}$, we have $\tilde{H}^{i,a}(p) = h^{i}_{p}(\tilde{\delta}^{i,a*}(p))$ and
\begin{equation*}
\tilde{H}^{i,a'}(p) = -\Lambda^{i,a}(\tilde{\delta}^{i,a*}(p))<0.
\end{equation*}

So $\tilde{H}^{i,a}$ is decreasing and
\begin{equation*}
 \tilde{\delta}^{i,a*}(p) = \left(\Lambda^{i,a}\right)^{-1} \left( -\tilde{H}^{i,a'}(p) \right).
\end{equation*}

Let us now recall that $\forall p \in \mathbb{R}$, $H^{i,a}(p) = \underset{\delta \geq -\delta_{\infty}}{\sup} \ h^{i}_{p}(\delta)$.\\

For all $p \in \mathbb{R}$ such that  $-\delta_{\infty} \leq \tilde{\delta}^{i,a*}(p)$, we clearly have
\begin{equation*}
    H^{i,a}(p) = h_{p}^{i}(\tilde{\delta}^{i,a*}(p)).
\end{equation*}

Otherwise, if $\tilde{\delta}^{i,a*}(p)< -\delta_{\infty}$, we can easily see that $h_{p}^{i}(.)$ is increasing on $]-\infty, \tilde{\delta}^{i,a*}(p)]$ and decreasing on $[\tilde{\delta}^{i,a*}(p),+\infty[$, which implies
\begin{equation*}
H^{i,a}(p) = h^{i}_{p}(-\delta_{\infty}).
\end{equation*}

This means that the supremum in $H^{i,a}(p)$ is reached at a unique $\delta^{i,a*}(p)$ given by
\begin{equation*}
    \delta^{i,a*}(p) = \max(\tilde{\delta}^{i,a*}(p),- \delta_{\infty}).
\end{equation*}

In particular, $\delta^{i,a*}$ is continuous and nondecreasing, so $H^{i,a}$ is continuous. Moreover, for all $p \in \mathbb{R}$ such that $\tilde{\delta}^{i,a*}(p) > -\delta_{\infty}$, we have $H^{i,a}(p) = \tilde{H}^{i,a}(p)$ so $H^{i,a}$ is decreasing on $]\tilde{\delta}^{i,a*-1}(-\delta_{\infty}), +\infty[$ and its derivative on this interval is
\begin{equation*}
{H^{i,a}}'(p) = -\Lambda^{i,a}(\tilde{\delta}^{i,a*}(p)) = -\Lambda^{i,a}(\delta^{i,a*}(p)).
\end{equation*}

On $]-\infty, \tilde{\delta}^{i,a*-1}(-\delta_{\infty})[$, $H^{i,a}$ is affine and its derivative is
\begin{equation*}
{H^{i,a}}'(p) = -\Lambda^{i,a}(-\delta_{\infty}) = -\Lambda^{i,a}(\delta^{i,a*}(p)).
\end{equation*}

Thus, by continuity of $\delta^{i,a*}$, $H^{i,a}$ is continuously differentiable and decreasing on $\mathbb{R}$. In particular, $|{H^{i,a}}'(p)| \leq \Lambda^{i,a}(-\delta_{\infty})$ for all $p \in \mathbb{R},$ so $H^{i,a}$ is Lipschitz.\\

\end{proof}

In what follows, we denote by $L^{i,a}$ the Lipschitz constant of $H^{i,a}$ for all $i \in \{1,\ldots, d\},$ and we define similarly $L^{i,b}$ the Lipschitz constant of $H^{i,b}$ for all $i \in \{1,\ldots, d\}.$\\

For $\pi \in C^0(\mathbb{R}^d, \mathbb{R}_+)$, let us consider $\mathcal{C}_{\pi}$ the following vector space:
$$\mathcal{C}_{\pi} = \left\{ u\in C^0(\mathbb{R}^d, \mathbb{R}) \left| \underset{q \in \mathbb{R}^d}{\sup} \left| \frac{u(q)}{1+\pi(q)} \right| < + \infty  \right. \right\}.$$
Equipped with the norm $ u \in \mathcal{C}_{\pi} \mapsto \|u\|_{\pi} = \underset{q \in \mathbb{R}^d}{\sup} \left| \frac{u(q)}{1+\pi(q)} \right|$, $\mathcal{C}_{\pi}$ is a Banach space.\\

We now consider for the rest of the paper that there exists $p\in \mathbb{N}^*$ and $C > 0$ such that:
\begin{itemize}
    \item $\forall q \in \mathbb{R}^d, \pi(q) \leq C\left(1+ \|q\|^p \right)$,\\
    \item $\forall q,y \in \mathbb{R}^d, \frac{1+\pi(q+y)}{1+\pi(q)} \leq C\left(1+ \|y\|^p \right)$,\\
    \item $\forall i \in \{1,\ldots, d\}$, $\int_{\mathbb{R}_{+}^{*}}\left(z^p\mu^{i,b}(dz) + z^p\mu^{i,a}(dz) \right) < +\infty$,\\
\end{itemize}
where $\|.\|$ denotes the Euclidean norm on $\mathbb R^d$.\\

Moreover, we assume that $\psi, \ell_d \in \mathcal{C}_\pi$.\footnote{This assumption implies in particular that $\psi$ and $\ell_d$ have, at most, polynomial growth at infinity.}\\

\begin{rem}
  For the examples of Remark \ref{psi_forms}, it is natural to choose a quadratic function $\pi$ such that $\psi, \ell_d \le \pi$. Then, the above assumptions are satisfied for $p=2$ whenever $\mu^{i,b}$ and $\mu^{i,a}$ have a finite second moment.\\
\end{rem}

\begin{prop}
For all $u \in \mathcal{C}_{\pi}$, the function
$$ F(u): q \in \mathbb{R}^d \mapsto \psi(q) - \underset{i=1}{\overset{d}{\sum}}\int_{\mathbb{R}_{+}^{*}} zH^{i,b} \left(\frac{u(q) -  u(q+ze^{i}) }{z}\right) \mu^{i,b}(dz) - \underset{i=1}{\overset{d}{\sum}}\int_{\mathbb{R}_{+}^{*}}zH^{i,a} \left(\frac{u(q) - u(q-ze^{i})}{z} \right) \mu^{i,a}(dz)$$
is in $\mathcal{C}_{\pi}.$\\
\end{prop}

\begin{proof}
Let $u\in \mathcal{C}_{\pi}$.\\

Let us consider $q \in \mathbb{R}^d$ and a sequence $(q_n)_n$ converging towards $q$.\\

From the continuity of $\psi$, we have $\lim_{n \to +\infty} \psi(q_n) = \psi(q)$.\\

Also, $\forall i \in \{1, \ldots, d\}, \forall z \in \mathbb{R}_+^*$, from the continuity of $H^{i,b}$ and $u$, we have
$$\lim_{n \to +\infty} z H^{i,b} \left(\frac{u(q_n) -  u(q_n+ze^{i}) }{z}\right) = z H^{i,b} \left(\frac{u(q) -  u(q+ze^{i}) }{z}\right).$$

Now, we write $ H^{i,b}(p)  \leq  H^{i,b}(0) + L^{i,b} |p|$ so that we get
\begin{eqnarray*}
z H^{i,b} \left(\frac{u(q_n) -  u(q_n+ze^{i}) }{z}\right)&\le& z H^{i,b}(0) + L^{i,b} |u(q_n) -  u(q_n+ze^{i})|\\
&\le& z H^{i,b}(0) + L^{i,b} |u(q_n)| + C L^{i,b} \|u\|_\pi \left(1+ \pi(q_n)\right) \left(1+ z^p\right)\\
&\le& z H^{i,b}(0) + L^{i,b} \sup_n|u(q_n)| + C L^{i,b} \|u\|_\pi \left(1+ \sup_n\pi(q_n)\right) \left(1+ z^p\right),
\end{eqnarray*}
which is integrable by assumption. Using the same technique for the terms associated with the ask side and Lebesgue's dominated convergence theorem, we conclude that $\lim_{n \to +\infty} F(u)(q_n) = F(u)(q)$, hence the continuity of $F(u)$.\\

Moreover, for all $q \in \mathbb{R}^d,$ we have
\begin{align*}
\left| \frac{F(u)(q)}{1+\pi(q)} \right|& = \Bigg| \frac{\psi(q)}{1+\pi(q)} - \underset{i=1}{\overset{d}{\sum}}\int_{\mathbb{R}_{+}^{*}} \frac{z}{1+\pi(q)}H^{i,b} \left(\frac{u(q) -  u(q+ze^{i}) }{z}\right) \mu^{i,b}(dz)\\
& \qquad \qquad \qquad \qquad \quad - \underset{i=1}{\overset{d}{\sum}}\int_{\mathbb{R}_{+}^{*}}\frac{z}{1+\pi(q)}H^{i,a} \left(\frac{u(q) - u(q-ze^{i})}{z} \right) \mu^{i,a}(dz) \Bigg|\\
& \leq \|\psi\|_{\pi} + \underset{i=1}{\overset{d}{\sum}}\int_{\mathbb{R}_{+}^{*}} \frac{z}{1+\pi(q)} H^{i,b} \left(\frac{u(q) -  u(q+ze^{i}) }{z} \right) \mu^{i,b}(dz)\\
& \qquad + \underset{i=1}{\overset{d}{\sum}}\int_{\mathbb{R}_{+}^{*}}\frac{z}{1+\pi(q)}H^{i,a} \left(\frac{u(q) - u(q-ze^{i})}{z}  \right) \mu^{i,a}(dz)\\
&\leq \|\psi\|_{\pi} + \underset{i=1}{\overset{d}{\sum}}\int_{\mathbb{R}_{+}^{*}} \frac{1}{1+\pi(q)} \left(zH^{i,b}(0) + L^{i,b} \left|u(q) -  u(q+ze^{i})\right| \right) \mu^{i,b}(dz)\\
& \qquad + \underset{i=1}{\overset{d}{\sum}}\int_{\mathbb{R}_{+}^{*}}\frac{1}{1+\pi(q)} \left(z H^{i,a}(0) + L^{i,a} \left|u(q) -  u(q-ze^{i})\right| \right) \mu^{i,a}(dz)\\
&\leq \|\psi\|_{\pi} + \underset{i=1}{\overset{d}{\sum}}\int_{\mathbb{R}_{+}^{*}} \left(zH^{i,b}(0) + L^{i,b} \|u\|_{\pi} + C L^{i,b} \|u\|_{\pi}\left(1+ z^p\right) \right) \mu^{i,b}(dz)\\
& \qquad + \underset{i=1}{\overset{d}{\sum}}\int_{\mathbb{R}_{+}^{*}}\left(zH^{i,a}(0) + L^{i,a} \|u\|_{\pi} + C L^{i,a} \|u\|_{\pi}\left(1+ z^p\right) \right) \mu^{i,a}(dz).\\
\end{align*}

We conclude that $\underset{q \in \mathbb{R}^d}{\sup} \left| \frac{F(u)(q)}{1+\pi(q)} \right| < +\infty$ and therefore that $F(u) \in \mathcal{C}_{\pi}.$\\
\end{proof}

We can therefore define a functional $F : \mathcal{C}_{\pi} \rightarrow \mathcal{C}_{\pi}$ such that, for all $u\in \mathcal{C}_{\pi}$ and for all $q \in \mathbb{R}^d,$
\begin{align*}
F(u)(q) = \psi(q) - \underset{i=1}{\overset{d}{\sum}}\int_{\mathbb{R}_{+}^{*}} zH^{i,b} \left(\frac{u(q) -  u(q+ze^{i}) }{z}\right) \mu^{i,b}(dz) - \underset{i=1}{\overset{d}{\sum}}\int_{\mathbb{R}_{+}^{*}}zH^{i,a} \left(\frac{u(q) - u(q-ze^{i})}{z} \right) \mu^{i,a}(dz).\\
\end{align*}

We now come to the main property of the function $F$.\\

\begin{prop}
$F$ is Lipschitz on $\mathcal{C}_{\pi}.$\\
\end{prop}

\begin{proof}
Let $u,v \in \mathcal{C}_{\pi}.$ For all $q\in \mathbb{R}^d$, we have
\begin{eqnarray*}
    \left| F(u)(q) - F(v)(q) \right| &\leq& \underset{i=1}{\overset{d}{\sum}}\int_{\mathbb{R}_{+}^{*}} z\left|H^{i,b} \left(\frac{v(q) -  v(q+ze^{i}) }{z}\right)-H^{i,b} \left(\frac{u(q) -  u(q+ze^{i}) }{z}\right)\right| \mu^{i,b}(dz)\\
    && + \underset{i=1}{\overset{d}{\sum}}\int_{\mathbb{R}_{+}^{*}}z\left|H^{i,a} \left(\frac{v(q) - v(q-ze^{i})}{z} \right) - H^{i,a} \left(\frac{u(q) - u(q-ze^{i})}{z} \right)\right| \mu^{i,a}(dz).\\
\end{eqnarray*}

Therefore
\begin{eqnarray*}
    \left| F(u)(q) - F(v)(q) \right|  &\leq & \underset{i=1}{\overset{d}{\sum}}\int_{\mathbb{R}_{+}^{*}}L^{i,b} \left|v(q) -  v(q+ze^{i})- u(q) +  u(q+ze^{i}) \right| \mu^{i,b}(dz)\\
    && + \underset{i=1}{\overset{d}{\sum}}\int_{\mathbb{R}_{+}^{*}}L^{i,a}\left|v(q) -  v(q-ze^{i})- u(q) +  u(q-ze^{i}) \right| \mu^{i,a}(dz).\\
    &\leq& \underset{i=1}{\overset{d}{\sum}}\int_{\mathbb{R}_{+}^{*}}L^{i,b} \left|v(q) - u(q) \right|\mu^{i,b}(dz) +\underset{i=1}{\overset{d}{\sum}}\int_{\mathbb{R}_{+}^{*}}L^{i,b} \left| v(q+ze^{i}) - u(q+ze^{i}) \right| \mu^{i,b}(dz)\\
    && + \underset{i=1}{\overset{d}{\sum}}\int_{\mathbb{R}_{+}^{*}}L^{i,a} \left|v(q) - u(q) \right|\mu^{i,a}(dz) +\underset{i=1}{\overset{d}{\sum}}\int_{\mathbb{R}_{+}^{*}}L^{i,a} \left| v(q-ze^{i}) - u(q-ze^{i}) \right| \mu^{i,a}(dz).\\
\end{eqnarray*}

We obtain therefore,
\begin{eqnarray*}
    \frac{\left| F(u)(q) - F(v)(q) \right|}{1+\pi(q)}&\leq& \underset{i=1}{\overset{d}{\sum}}\int_{\mathbb{R}_{+}^{*}}L^{i,b} \|u-v\|_\pi \mu^{i,b}(dz) +\underset{i=1}{\overset{d}{\sum}}\int_{\mathbb{R}_{+}^{*}}C L^{i,b} \|u-v\|_\pi \left(1+ z^p\right)  \mu^{i,b}(dz)\\
    && + \underset{i=1}{\overset{d}{\sum}}\int_{\mathbb{R}_{+}^{*}}L^{i,a} \|u-v\|_\pi\mu^{i,a}(dz) +\underset{i=1}{\overset{d}{\sum}}\int_{\mathbb{R}_{+}^{*}} C L^{i,a} \|u-v\|_\pi \left(1+ z^p\right) \mu^{i,a}(dz).\\
\end{eqnarray*}

By taking the supremum over $q$, we get that there exists a constant $K>0$ such  $$\left\|F(u) - F(v) \right\|_{\pi} \leq K  \|u-v\|_\pi.$$
We conclude that $F$ is Lipschitz continuous.\\
\end{proof}

The Lipschitz property of $F$ allows to obtain the following existence and uniqueness theorem:
\begin{thm}
\label{exuni}
There exists a unique function $\mathcal{W} \in C^1([0,T], \mathcal{C}_\pi)$ such that $w : (t,q) \in [0,T] \times \mathbb{R}^d \mapsto \mathcal{W}(t)(q)$ is solution to \eqref{eqn:HJB} with terminal condition $w(T,q) = -\ell_d(q), \forall q \in \mathbb{R}^d$.\\
\end{thm}

\begin{proof}
Let us observe that $\mathcal{W} \in C^1([0,T], \mathcal{C}_\pi)$ is solution of the Cauchy problem
$$
\begin{cases}
\mathcal{W}'(t) = F(\mathcal{W}(t)), \forall t\in[0,T] \\
\mathcal{W}(T) = - \ell_d
\end{cases}
$$
if and only if $w : (t,q) \in [0,T] \times \mathbb{R}^d \mapsto \mathcal{W}(t)(q)$ is solution to \eqref{eqn:HJB} with terminal condition $w(T,q) = -\ell_d(q), \forall q \in \mathbb{R}^d$.\\

As $\left(\mathcal{C}_{\pi}, \|. \|_{\pi} \right)$ is a Banach space and $F:\mathcal{C}_{\pi} \rightarrow \mathcal{C}_{\pi}$ is Lipschitz continuous, we know by Cauchy-Lipschitz theorem that there exists a unique maximal solution $\mathcal{W}$ to the above equation, and that this solution is in fact global, meaning in particular that $\mathcal{W}$ is defined on $[0,T]$.\\
\end{proof}

\subsection{Verification theorem}

We now want to prove that $\theta$ is in fact the function $w$ defined in Theorem \ref{exuni} and deduce the optimal controls associated with the problem \eqref{MMpb} using a verification argument.\\

\begin{thm}
\label{controls}
Let $w$ be the function defined in Theorem \ref{exuni}.\\

Let $(t,q)\in [0,T) \times \mathbb{R}^d$.\\

Let us define $(\bar{\delta}^{i})_{i \in \{1,\ldots,d\} }=(\bar{\delta}^{i,b},\bar{\delta}^{i,a})_{i \in \{1,\ldots,d\}} \in \mathcal{A}$ such that $\forall i \in \{1,\ldots,d\}, \forall s \in [t,T], \forall z>0$:
\begin{equation}
\begin{split}
\bar{\delta}^{i,b}(s,z) = \delta^{i,b*}\left( \frac{w(s,q_{s-}) - w(s,q_{s-}+ze^{i})}{z} \right),\\
\bar{\delta}^{i,a}(s,z) = \delta^{i,a*}\left( \frac{w(s,q_{s-}) - w(s,q_{s-}-ze^{i})}{z} \right), \nonumber
\end{split}
\end{equation}
where $\delta^{i,b*}$ and $\delta^{i,a*}$ are the functions defined in Lemma \ref{lemmH} and $(q_{s})_{t\leq s \leq T} = (q^{t,q,(\bar{\delta}^{1},\ldots,\bar{\delta}^{d})}_{s})_{t\leq s \leq T}$.\\

Then, $\theta(t,q)=w(t,q)$ and $(\bar{\delta}^{1},\ldots,\bar{\delta}^{d})$ is an optimal control for our stochastic control problem starting at time $t$ with $q_t = q$.\\
\end{thm}

\begin{proof}
Let $(\delta^i)_{i \in \{1,\ldots,d\}} = (\delta^{i,b}, \delta^{i,a})_{i \in \{1,\ldots,d\}} \in \mathcal{A}$ be an arbitrary control and let us denote by $(q_s)_{s\in [t,T]}$ the process $\left(q_s^{t,q,(\delta^1, \ldots, \delta^d)}\right)_{s\in [t,T]}$.\\

Let us first prove that for all $i \in \left\{ 1,\ldots,d \right\}$,
$$\mathbb{E} \left[ \int_{t}^T \int_{\mathbb{R}_+^*} \left|w(s,q_{s-} + ze^i)  - w(s,q_{s-}) \right| \Lambda^{i,b}(\delta^{i,b}_s) \mu^{i,b}(dz) ds \right] < +\infty.$$

Denoting by $M^w$ the quantity  $\underset{t \in [0,T]}{\sup} \|w(t,\cdot)\|_\pi$, we have
 \begin{align*}
\mathbb{E} \left[ \int_{t}^T \int_{\mathbb{R}_+^*} \left|w(s,q_{s-} + ze^i)\  - \right. \right. & \left. \left.\vphantom{e^i}w(s,q_{s-}) \right| \vphantom{ \int_{t}^T \int_{\mathbb{R}_+^*}} \Lambda^{i,b}(\delta^{i,b}_s) \mu^{i,b}(dz) ds \right]\\
& \leq \Lambda^{i,b}(-\delta_{\infty}) \mathbb{E} \left[ \int_{t}^T \int_{\mathbb{R}_+^*} \left(\left|w(s,q_{s-}  + ze^i)\right| + \left| w(s,q_{s-}) \right|\right) \mu^{i,b}(dz) ds \right]\\
& \leq \Lambda^{i,b}(-\delta_{\infty}) M^w \mathbb{E} \left[ \int_{t}^T \int_{\mathbb{R}_+^*} \left(1 + \pi(q_{s-}+ze^i) + 1+ \pi(q_{s-}) \right) \mu^{i,b}(dz) ds \right]\\
& \leq \Lambda^{i,b}(-\delta_{\infty}) M^w \mathbb{E} \left[ \int_{t}^T \int_{\mathbb{R}_+^*} \left(C\left(1+ z^p\right)(1+\pi(q_{s-})) + 1 + \pi(q_{s-})) \right) \mu^{i,b}(dz) ds \right].\\
\end{align*}
Therefore
\begin{align*}
\mathbb{E} \left[ \int_{t}^T \int_{\mathbb{R}_+^*} \right. & \left.\left| w(s,q_{s-} + ze^i)\  - \vphantom{e^i}w(s,q_{s-}) \right| \vphantom{ \int_{t}^T \int_{\mathbb{R}_+^*}} \Lambda^{i,b}(\delta^{i,b}_s) \mu^{i,b}(dz) ds \right]\\
& \leq \Lambda^{i,b}(-\delta_{\infty}) M^w \mathbb{E} \left[ \int_{t}^T \int_{\mathbb{R}_+^*}\left(C\left(1+ z^p\right)\left(1+C\left(1 + \|q_{s-}\|^p\right)\right) + 1 + C\left(1+\|q_{s-}\|^p\right)\right) \mu^{i,b}(dz) ds \right].\\
\end{align*}

Subsequently, we just have to prove that $$\mathbb{E} \left[ \int_t^T \|q_{s-}\|^p ds \right] < +\infty.$$

Since $\|q_{s}\| \le \|q\| + \|q_{s} - q\|$, $\|q_{s}\|^p \le 2^{p-1} \left(\|q\|^p + \|q_{s} - q\|^p \right)$, and we need to prove that
$$\mathbb{E} \left[ \int_t^T \|q_{s-}-q\|^p ds \right] < +\infty.$$

As we are working in $\mathbb{R}^d$, it is equivalent to prove that $$\mathbb{E} \left[ \int_t^T \|q_{s-}-q\|^p_p ds \right] < +\infty,$$ where $\|(x_1,\ldots, x_d)'\|_p = \left(\sum_{i=1}^d |x_i|^p\right)^{1/p}$.\\

For that purpose, we introduce for each $j \in \{1,\ldots, d\},$ two independent Poisson processes $N^{j,b}$ and $N^{j,a}$ with respective intensities $\Lambda^{j,b}(-\delta_{\infty})$ and $\Lambda^{j,a}(-\delta_{\infty})$, and $(\xi^{j,b}_k)_{k \geq 1}$ and $(\xi^{j,a}_k)_{k \geq 1}$ two sequences of i.i.d. random variables with respective distributions $\mu^{j,b}$ and $\mu^{j,a}.$ Then, we have
\begin{eqnarray*}
\mathbb{E} \left[ \int_t^T \|q_{s-}-q\|^p_p ds \right] & = & \mathbb{E} \left[ \int_t^T  \sum_{j=1}^d \left|\int_{\mathbb{R}_{+}^{*}} z J^{j,b}(dt,dz) - \int_{\mathbb{R}_{+}^{*}} z J^{j,a}(dt,dz)\right|^p ds\right]\\
& \le & \mathbb{E} \left[ \int_t^T  \sum_{j=1}^d \left(\int_{\mathbb{R}_{+}^{*}} z J^{j,b}(dt,dz) + \int_{\mathbb{R}_{+}^{*}} z J^{j,a}(dt,dz)\right)^p\right]\\
& \leq& \mathbb{E} \left[ \int_t^T  \sum_{j=1}^d \left(\sum_{k=1}^{N^{j,b}_s} \xi^{j,b}_k + \sum_{k=1}^{N^{j,a}_s} \xi^{j,a}_k \right)^{p}   ds\right]\\
& \leq& 2^{p-1}\mathbb{E} \left[ \int_t^T \sum_{j=1}^d \left(\left(\sum_{k=1}^{N^{j,b}_s} \xi^{j,b}_k\right)^p + \left(\sum_{k=1}^{N^{j,a}_s} \xi^{j,a}_k \right)^p\right)   ds\right]\\
& \leq& 2^{p-1} \int_{t}^T \sum_{j=1}^d \left(\mathbb{E} \left[\left(N^{j,b}_s \right)^{p-1}\sum_{k=1}^{N^{j,b}_s} \left(\xi^{j,b}_k\right)^p \right] + \mathbb{E} \left[\left(N^{j,a}_s \right)^{p-1}\sum_{k=1}^{N^{j,a}_s} \left(\xi^{j,a}_k\right)^p \right] \right) ds\\
& \leq& 2^{p-1} \int_{t}^T \sum_{j=1}^d \left(\mathbb{E} \left[\left(N^{j,b}_T \right)^{p}\right]\mathbb{E} \left[ \left(\xi^{j,b}_1\right)^p \right] + \mathbb{E} \left[\left(N^{j,a}_T \right)^{p}\right]\mathbb{E} \left[\left(\xi^{j,a}_1\right)^p \right] \right) ds\\
& \leq& 2^{p-1} T  \sum_{j=1}^d \left(\mathbb{E}\left[\left(N^{j,b}_T \right)^{p}\right] \int_{\mathbb{R}_+^*} z^p \mu^{j,b}(dz) + \mathbb{E}\left[\left(N^{j,a}_T \right)^{p}\right] \int_{\mathbb{R}_+^*} z^p \mu^{j,a}(dz)\right) < +\infty.\\
\end{eqnarray*}

Using the above, we have, for all $i \in \left\{ 1,\ldots,d \right\}$,
\begin{align*}
    \mathbb{E} \Bigg[ \int_t^T \int_{\mathbb{R}_+^*} &  \left(w(s,q_{s-} + ze^i) - w(s,q_{s-}) \right) J^{i,b}(ds,dz) \Bigg] \\
    & = \mathbb{E} \Bigg[ \int_{t}^T \int_{\mathbb{R}_+^*} \left(w(s,q_{s-} + ze^i)\  - w(s,q_{s-}) \right) \Lambda^{i,b}(\delta^{i,b}_s) \mu^{i,b}(dz) ds \Bigg],\\
\end{align*}

Of course, we can similarly prove that, for all $i \in \left\{ 1,\ldots,d \right\}$,
\begin{align*}
    \mathbb{E} \Bigg[ \int_t^T \int_{\mathbb{R}_+^*} &  \left(w(s,q_{s-} - ze^i) - w(s,q_{s-}) \right) J^{i,a}(ds,dz) \Bigg] \\
    & = \mathbb{E} \Bigg[ \int_{t}^T \int_{\mathbb{R}_+^*} \left(w(s,q_{s-} - ze^i)\  - w(s,q_{s-}) \right) \Lambda^{i,a}(\delta^{i,a}_s) \mu^{i,a}(dz) ds \Bigg].\\
\end{align*}

Now, by applying Itô's formula, we get
\begin{align*}
    w(T,q_T) = w(t,q) + \int_t^T \frac{\partial w}{\partial t}(s,q_s) ds &+ \sum_{i=1}^d \int_t^T \int_{\mathbb{R}_+^*} \left(w(s,q_{s-} + ze^i) - w(s,q_{s-}) \right) J^{i,b}(ds,dz)\\
    &+ \sum_{i=1}^d \int_t^T \int_{\mathbb{R}_+^*} \left(w(s,q_{s-} - ze^i) - w(s,q_{s-}) \right) J^{i,a}(ds,dz).\\
\end{align*}

By taking expectation, we get
\begin{align*}
    \mathbb{E} \left[ w(T,q_T) \right] = w(t,q) + \mathbb{E} \Bigg[\int_t^T \Bigg\{ \frac{\partial w}{\partial t}(s,q_s)  &+ \sum_{i=1}^d  \int_{\mathbb{R}_+^*} \Lambda^{i,b}(\delta^{i,b}(s,z))\left(w(s,q_{s-} + ze^i) - w(s,q_{s-}) \right) \mu^{i,b}(dz)\\
    &+ \sum_{i=1}^d \int_{\mathbb{R}_+^*} \Lambda^{i,a}(\delta^{i,a}(s,z))\left(w(s,q_{s-} - ze^i) - w(s,q_{s-}) \right) \mu^{i,a}(dz) \Bigg\} ds \Bigg],
\end{align*}
which, by definition of $w$, gives us the following inequality:
\begin{align*}
    \mathbb{E} \left[-\ell_d(q_T) \right] \le w(t,q) + \mathbb{E} \Bigg[\int_t^T \Bigg\{ \psi(q_s)  &- \sum_{i=1}^d  \int_{\mathbb{R}_+^*} z\Lambda^{i,b}(\delta^{i,b}(s,z))\delta^{i,b}(s,z) \mu^{i,b}(dz)\\
    &- \sum_{i=1}^d \int_{\mathbb{R}_+^*} z\Lambda^{i,a}(\delta^{i,a}(s,z))\delta^{i,a}(s,z) \mu^{i,a}(dz) \Bigg\} ds \Bigg],
\end{align*}
with equality when $(\delta^i)_{i \in \{1,\ldots,d\}} = \left(\bar\delta^i\right)_{i \in \{1,\ldots,d\}}$.\\

In other words,
\begin{align*}
\mathbb{E} \Bigg[\int_t^T \Bigg\{ \sum_{i=1}^d  \int_{\mathbb{R}_+^*} &\left(z\Lambda^{i,b}(\delta^{i,b}(s,z))\delta^{i,b}(s,z) \mu^{i,b}(dz)+ z\Lambda^{i,a}(\delta^{i,a}(s,z))\delta^{i,a}(s,z) \mu^{i,a}(dz)\right)\\
 &- \psi(q_s) \Bigg\} ds -\ell_d(q_T) \Bigg] \le w(t,q),
 \end{align*}
with equality when $(\delta^i)_{i \in \{1,\ldots,d\}} = \left(\bar\delta^i\right)_{i \in \{1,\ldots,d\}}$.\\

By taking the supremum over $(\delta^i)_{i \in \{1,\ldots,d\}} \in \mathcal{A}$, we get $\theta(t,q) = w(t,q)$ and the fact that $\left(\bar\delta^i\right)_{i \in \{1,\ldots,d\}}$ is optimal.\\
\end{proof}

\section{Solving the multi-asset market making problem with factors}\label{factorisation}

Let us now consider the particular case of problem \eqref{MMpb} where $\forall q \in \mathbb{R}^{d}$, $\psi(q) = \bar{\psi}\left(q'\Sigma q\right)$ and $\ell_{d}(q) = \bar{\ell}_{d}(q'\Sigma q)$ for some continuous functions $\bar{\psi}$ and $\bar{\ell}_{d}$ with, at most, polynomial growth at infinity. This particular case covers the examples of the literature (see Remark \ref{psi_forms}).\\

If the prices of the $d$ assets are modeled using a small number $k$ of factors, as it is the case in most econometric models of financial asset prices, then the variance-covariance matrix $\Sigma$ takes the form
$$\Sigma = \beta V \beta' + R,$$
where $\beta$ is a $d$-by-$k$ matrix of real coefficients, $V$ the $k$-by-$k$ variance-covariance matrix of the factors, and $R$ the $d$-by-$d$ variance-covariance matrix of the residuals.\\

If the explanatory power of the factors is high, $R$ should be small compared to $\Sigma$ (in Frobenius norm for instance). Our approach consists in ignoring the residuals, i.e. setting $R$ to $0$. In other words, we project the market risk on a space of factors of dimension $k$. As we shall see in Section \ref{num}, this approach provides very good results as measured by the objective function \eqref{MMpb}.\\

In what follows, we also discuss an approximation method based on Monte-Carlo simulations to account for the influence of $R$ once one has computed the optimal quotes in the case with no residual risk. The advantages and drawbacks of this additional approximation method will be discussed in Section \ref{num}.\\

\subsection{A low-dimensional approximation}
\label{grid}
Let us now assume that $\Sigma = \beta V \beta'$, i.e. $R=0$. Under this assumption, we can write problem \eqref{MMpb} as the maximization of
\begin{equation}
\label{MMfactorpb}
\mathbb{E} \left[ X_{T} + \sum_{i=1}^{d} q^{i}_{T}S^{i}_{T} - \bar{\ell}_d\left((\beta'q_{T})' V (\beta'q_{T})\right) - \int_{0}^{T} \bar{\psi}\left((\beta'q_{t})' V (\beta'q_{t})\right)dt \right].
\end{equation}

Using the same ideas as in Section \ref{genpb}, this expression can be written as
\begin{align*}
\mathbb{E}\Bigg[ \int\limits_{0}^{T} \Bigg\lbrace \int_{\mathbb{R}_{+}^{*}} \sum_{i=1}^{d}\left( \delta^{i,b}(t,z)z \Lambda^{i,b}(\delta^{i,b}(t,z))\mu^{i,b}(dz) + \delta^{i,a}(t,z)z \Lambda^{i,a}(\delta^{i,a}(t,z))\mu^{i,a}(dz) \right)  &- \bar{\psi}\left((\beta'q_{t})' V (\beta'q_{t})\right)  \Bigg\rbrace dt\\
& - \bar{\ell}_d\left((\beta'q_{T})' V (\beta'q_{T})\right) \Bigg].
\end{align*}

Let us introduce $(f_{t})_{t\in[0,T]} = (\beta'q_{t})_{t\in[0,T]}$. Then, the problem of maximizing \eqref{MMfactorpb} is equivalent to that of maximizing
\begin{equation*}
\mathbb{E}\Bigg[ \int\limits_{0}^{T} \Bigg\lbrace \int_{\mathbb{R}_{+}^{*}} \sum_{i=1}^{d}\left(\delta^{i,b}(t,z)z \Lambda^{i,b}(\delta^{i,b}(t,z))\mu^{i,b}(dz) + \delta^{i,a}(t,z)z \Lambda^{i,a}(\delta^{i,a}(t,z))\mu^{i,a}(dz) \right)  - \bar{\psi}\left(f_t'Vf_t\right) \Bigg\rbrace dt - \bar{\ell}_d\left(f_T'Vf_T\right) \Bigg].
\end{equation*}

The state process of our problem is now the Markov process $(f_{t})_{t \in[0,T]}$ instead of $(q_{t})_{t \in[0,T]}$: we have reduced the dimension of the problem from $d$ to $k$.\\

Let us introduce  $\tilde{\mathcal{J}}:[0,T] \times \mathbb{R}^k \times \mathcal{A}\rightarrow \mathbb{R}$ such that, $\forall t \in [0,T]$, $\forall f=(f^{1},\ldots,f^{k})'\in \mathbb{R}^k$, $\forall (\delta^{i})_{i \in \{1, \ldots, d\}}\in \mathcal{A}$
\begin{eqnarray*}
\tilde{\mathcal{J}}(t,f,(\delta^{i})_{i \in \{1, \ldots, d\}}) &=&  \mathbb{E}\Bigg[ \int\limits_{t}^{T} \Bigg\lbrace \int_{\mathbb{R}_{+}^{*}} \sum_{i=1}^{d}\bigg(\delta^{i,b}(s,z)z \Lambda^{i,b}(\delta^{i,b}(s,z))\mu^{i,b}(dz) + \delta^{i,a}(s,z)z \Lambda^{i,a}(\delta^{i,a}(s,z))\mu^{i,a}(dz) \bigg) \\
&& - \bar{\psi}\left(f_{s}'Vf_s\right) \Bigg\rbrace ds - \bar{\ell}_d\left(f_T'Vff_T\right)  \Bigg],
\end{eqnarray*}
where $(f_s)_{s\in [t,T]} = (f^{t,f,(\delta^{i})_{i \in \{1, \ldots, d\}}}_{s})_{s\in [t,T]}$ is here the state process starting in state $f$ at time $t$ and controlled by $(\delta^{i})_{i \in \{1, \ldots, d\}}$.\\

The value function $\tilde{\theta}:[0,T]\times \mathbb{R}^k\rightarrow \mathbb{R}$ of the problem is then defined as follows:
\begin{equation*}
\tilde{\theta}(t,f) = \underset{(\delta^{i})_{i \in \{1, \ldots, d\}} \in \mathcal{A}}{\sup} \tilde{\mathcal{J}}(t,f,(\delta^{i})_{i \in \{1, \ldots, d\}}), \quad \forall (t,f) \in [0,T] \times \mathbb{R}^k.
\end{equation*}

By using the same arguments as in Section \ref{genpb}, we can show that $\tilde{\theta}$ is the unique (in a large class of functions) smooth solution to the following integro-differential Hamilton-Jacobi equation:
\begin{align}
0 = \frac{\partial \tilde{\theta}}{\partial t}(t,f) - \bar{\psi}\left(f' V f\right) &+ \underset{i=1}{\overset{d}{\sum}}\int_{\mathbb{R}_{+}^{*}} zH^{i,b} \left(\frac{\tilde{\theta}(t,f) - \tilde{\theta}(t,f+z\tilde{e}^{i}) }{z}\right) \mu^{i,b}(dz) \nonumber\\
& + \underset{i=1}{\overset{d}{\sum}}\int_{\mathbb{R}_{+}^{*}} zH^{i,a} \left(\frac{\tilde{\theta}(t,f) - \tilde{\theta}(t,f-z\tilde{e}^{i})}{z} \right) \mu^{i,a}(dz), \quad \forall (t,f)\in [0,T) \times \mathbb{R}^k,
\label{eqn:HJBfac}
\end{align}
with terminal condition $\tilde{\theta}(T,f) = -\bar{\ell}_{d}(f'Vf), \forall f \in \mathbb{R}^k$, where $\forall i \in \{1, \ldots, d\}, \tilde{e}^i = \beta' e^i$.\\

Furthermore, the optimal controls are now given by:
\begin{equation*}
\begin{split}
\bar{\delta}^{i,b}(s,z) = \delta^{i,b*}\left( \frac{\tilde{\theta}(s,f_{s-}) - \tilde{\theta}(s,f_{s-}+z\tilde{e}^{i})}{z} \right),\\
\bar{\delta}^{i,a}(s,z) = \delta^{i,a*}\left( \frac{\tilde{\theta}(s,f_{s-}) - \tilde{\theta}(s,f_{s-}-z\tilde{e}^{i})}{z} \right).\\
\end{split}
\end{equation*}

When $R=0$, the problem boils down therefore to finding the solution $\tilde{\theta}$ of \eqref{eqn:HJBfac} with the appropriate terminal condition. In particular, from a numerical point of view, we need to approximate the solution of an equation involving time plus $k$ space dimensions, and this is doable with grid methods if $k$ is small.\\

\subsection{A Monte-Carlo method to take the residual risk into account}
\label{MCr}

As we shall see in Section \ref{num}, the above approximation method provides very good results as measured by the value of the objective function \eqref{MMpb}. Nevertheless, when market risk is projected on a low-dimensional space of factors, there are linear combinations of assets that falsely appear to be risk-free. To prevent trajectories of the inventory visiting too often regions that are falsely associated with low risk, it makes sense to look for methods that account for the residual risk measured by the matrix $R$.\\

In what follows, we propose an approximation method to take the residual risk into account. The idea consists in considering the first-order expansion in $\varepsilon$ where $$\Sigma = \beta V \beta' + \varepsilon R.$$
The rationale behind this idea is that, for a factor model with high explanatory power, $R$ should be small and it makes sense therefore to use a perturbative approach.\\

When $\varepsilon = 0,$ we know how to solve the problem, and the value function is given by $\tilde{\theta}.$ To approximate the value function $\theta$ of the problem for $\varepsilon>0$, we consider a first-order expansion of the form
$$\theta(t,q) = \tilde{\theta}(t,\beta'q) + \varepsilon \eta(t,q) + o(\varepsilon), \quad  \forall (t,q) \in [0,T] \times \mathbb{R}^d.$$

By plugging this expression into equation \eqref{eqn:HJB}, we formally get
\begin{align*}
    0  = &\frac{\partial \tilde{\theta}}{\partial t}(t,\beta'q) + \varepsilon \frac{\partial \eta}{\partial t}(t,q) + o(\varepsilon) - \bar{\psi}\left((\beta'q)'V(\beta'q) + \varepsilon q'Rq \right)\\
    & + \sum_{i=1}^d \int_{\mathbb{R}_{+}^{*}} zH^{i,b} \left(\frac{\tilde{\theta}(t,\beta'q) - \tilde{\theta}(t,\beta'q+ z\tilde{e}^i) }{z} + \varepsilon \frac{\eta(t,q) - \eta(t,q+ze^i)}{z} + o(\varepsilon)\right) \mu^{i,b}(dz)\\
    & + \sum_{i=1}^d \int_{\mathbb{R}_{+}^{*}} zH^{i,a} \left(\frac{\tilde{\theta}(t,\beta'q) - \tilde{\theta}(t,\beta'q - z\tilde{e}^i) }{z} + \varepsilon \frac{\eta(t,q) - \eta(t,q-ze^i)}{z} + o(\varepsilon)\right) \mu^{i,a}(dz),
\end{align*}
and $$\tilde{\theta}(T,\beta'q) + \varepsilon \eta(T,q) + o(\varepsilon) = -\bar{\ell}_d\left((\beta'q)'V(\beta'q) + \varepsilon q'Rq\right).$$

Assuming that $\bar{\psi}$ and $\bar{\ell_d}$ are $C^1$ and performing a Taylor expansion, we obtain
\begin{align*}
    0  = &\frac{\partial \tilde{\theta}}{\partial t}(t,\beta'q) + \varepsilon \frac{\partial \eta}{\partial t}(t,q) - \bar{\psi}\left((\beta'q)'V(\beta'q)\right) - \varepsilon \bar{\psi}'\left((\beta'q)'V(\beta'q)\right)  q'Rq\\
    & + \sum_{i=1}^d \int_{\mathbb{R}_{+}^{*}} zH^{i,b} \left(\frac{\tilde{\theta}(t,\beta'q) - \tilde{\theta}(t,\beta'q+ z\tilde{e}^i) }{z}\right)\mu^{i,b}(dz)\\
    & + \varepsilon \sum_{i=1}^d \int_{\mathbb{R}_{+}^{*}} {H^{i,b}}' \left(\frac{\tilde{\theta}(t,\beta'q) - \tilde{\theta}(t,\beta'q+ z\tilde{e}^i) }{z}\right) \left(\eta(t,q) - \eta(t,q+ze^i)\right) \mu^{i,b}(dz)\\
    & + \sum_{i=1}^d \int_{\mathbb{R}_{+}^{*}} zH^{i,a} \left(\frac{\tilde{\theta}(t,\beta'q) - \tilde{\theta}(t,\beta'q - z\tilde{e}^i) }{z}\right)\\
    & +\varepsilon \sum_{i=1}^d \int_{\mathbb{R}_{+}^{*}} {H^{i,a}}'\left(\frac{\tilde{\theta}(t,\beta'q) - \tilde{\theta}(t,\beta'q - z\tilde{e}^i) }{z}\right)\left(\eta(t,q) - \eta(t,q-ze^i)\right) \mu^{i,a}(dz) + o(\varepsilon),
\end{align*}
and $$\tilde{\theta}(T,\beta'q) + \varepsilon \eta(T,q) + o(\varepsilon) = -\bar{\ell}_d\left((\beta'q)'V(\beta'q)\right) -  \varepsilon \bar{\ell}'_d\left((\beta'q)'V(\beta'q)\right) q'Rq + o(\varepsilon).$$

As $\tilde{\theta}$ verifies \eqref{eqn:HJBfac}, we get
\begin{align*}
    0  = & \frac{\partial \eta}{\partial t}(t,q) - \bar{\psi}'\left((\beta'q)'V(\beta'q)\right)  q'Rq + \sum_{i=1}^d \int_{\mathbb{R}_{+}^{*}} {H^{i,b}}' \left(\frac{\tilde{\theta}(t,\beta'q) - \tilde{\theta}(t,\beta'q+ z\tilde{e}^i) }{z}\right) \left(\eta(t,q) - \eta(t,q+ze^i)\right) \mu^{i,b}(dz)\\
    & + \sum_{i=1}^d \int_{\mathbb{R}_{+}^{*}} {H^{i,a}}'\left(\frac{\tilde{\theta}(t,\beta'q) - \tilde{\theta}(t,\beta'q - z\tilde{e}^i) }{z}\right)\left(\eta(t,q) - \eta(t,q-ze^i)\right) \mu^{i,a}(dz).
\end{align*}
and $\eta(T,q) = -\bar{\ell}'_d\left((\beta'q)'V(\beta'q)\right) q'Rq.$\\

This equation, although in space-dimension $d$, is linear. Therefore, by the Feynman-Kac representation theorem, we have the following formula:
\begin{align*}
    \eta(t,q) = \mathbb{E}^{\tilde{\mathbb P}} \left[-\int_t^T \bar{\psi}'\left((\beta'q_s)'V(\beta'q_s)\right)  q_s'Rq_s ds - \bar{\ell}'_d\left((\beta'q_T)'V(\beta'q_T)\right) q_T'Rq_T \bigg| q_t = q \right],
\end{align*}
where under $\tilde{\mathbb P}$, for all $i \in \{1, \ldots, d\}$, $J^{i,b}$ and $J^{i,a}$ have their respective intensity kernels given by
\begin{align*}
    \tilde{\nu}^{i,b}_t(dz) & = - {H^{i,b}}' \left(\frac{\tilde{\theta}(t,\beta'q_{t-}) - \tilde{\theta}(t,\beta'q_{t-}+z\tilde{e}^i) }{z}\right) \mu^{i,b}(dz),\\
    \tilde{\nu}^{i,a}_t(dz) & = - {H^{i,a}}' \left(\frac{\tilde{\theta}(t,\beta'q_{t-}) - \tilde{\theta}(t,\beta'q_{t-}-z\tilde{e}^i)) }{z}\right) \mu^{i,a}(dz).\\
\end{align*}

\begin{rem}
It is noteworthy that the dynamics of $(q_s)_{s \in [t,T]}$ under $\tilde{\mathbb{P}}$ is that associated with the use of the optimal quotes when $R=0$.\\
\end{rem}

Thanks to this probabilistic representation, we can easily compute $\eta(t,q)$ for a given time $t$ and inventory $q$ using a Monte-Carlo method, and therefore easily compute both an approximation of the value function and an approximation of the optimal quotes that account, to the first order, for the residual risk. Of course, in practice, it would be prohibitively expensive in terms of computation time to carry out a Monte-Carlo simulation for all possible values of the inventory, but this method can alternatively be used (online) upon receiving a request for quote for a specific asset and given the current time and inventory (this will be discussed in Section \ref{MCnum}). \\

\begin{rem}
In the computation of the optimal quotes associated with asset $i$, one relies on the approximation
$$\theta(t,q_{t-}) - \theta(t,q_{t-} \pm z e^i) \simeq \tilde\theta(t,\beta'q_{t-}) - \tilde\theta(t,\beta'q_{t-} \pm z\tilde{e}^i) + \eta(t,q_{t-}) - \eta(t,q_{t-}\pm ze^i).$$ To compute $\eta(t,q_{t-}) - \eta(t,q_{t-}\pm ze^i)$, the same sample paths should be used in the estimations of $\eta(t,q_{t-})$ and $\eta(t,q_{t-}\pm ze^i)$. This is the same remark as for the computation of the Greeks of derivatives contracts with Monte-Carlo techniques.\\
\end{rem}

\section{Numerical results}
\label{num}
\subsection{The case of two assets: one factor vs. two factors}

\subsubsection{Model parameters}

In this section, we apply our multi-asset market making model to the case of two highly-correlated assets (here bonds). Our goal is to show that, in this case, the reduced one-factor model gives very similar results to the complete two-factor model. For this purpose, we consider two assets with the following characteristics:\\
\begin{itemize}
  \item Asset prices: $S^1_0 = S^2_0 = 100\ \textrm{\euro}$.
  \item Volatility of asset 1: $\sigma^1 = 1.2\  \textrm{\euro} \cdot \textrm{day}^{-\frac{1}{2}}$.
  \item Volatility of asset 2: $\sigma^2 = 0.6\  \textrm{\euro} \cdot \textrm{day}^{-\frac{1}{2}}$.
  \item Correlation: $\rho = 0.9$.
  \item Intensity functions: $$\Lambda^{i,b}(\delta) = \Lambda^{i,a}(\delta) = \lambda_{RFQ} \frac{1}{1+e^{\alpha_{\Lambda} + \beta_{\Lambda} \delta}},\quad  i\in \{1,2\},$$ with $\lambda_{RFQ} =30\ \textrm{day}^{-1}$, $\alpha_{\Lambda}=0.7$, and $\beta_{\Lambda}=30\ \textrm{\euro}^{-1}$. This corresponds to 30 RFQs per day for each asset, a probability of $\frac 1 {1+e^{0.7}} \simeq 33\%$ to trade when the answered quote is the reference price and a probability of $\frac 1 {1+e^{-0.2}} \simeq 55\%$ to trade when the answered quote is the reference price improved by 3 cents.
  \item Request sizes are distributed according to a Gamma distribution $\Gamma(\alpha_\mu,\beta_\mu)$ with $\alpha_\mu = 4$ and $\beta_\mu=4\cdot 10^{-4}$. This corresponds to an average request size of $10000$ assets (i.e. approximately $1000000\ \textrm{\euro}$) and a standard deviation equal to half the average.\\
\end{itemize}

The variance-covariance matrix is therefore given by
\begin{equation*}
\Sigma =
\begin{bmatrix}
1.44 & 0.648 \\
0.648 & 0.36
\end{bmatrix} = \Omega D \Omega' \simeq \begin{bmatrix}
0.906 & 0.424 \\
0.424 & -0.906
\end{bmatrix} \begin{bmatrix}
1.744 & 0 \\
0 & 0.056
\end{bmatrix}\begin{bmatrix}
0.906 & 0.424 \\
0.424 & -0.906
\end{bmatrix}.
\end{equation*}

We can see that the second eigenvalue is very small in comparison to the first. This justifies that it is reasonable to approximate the two-dimensional problem with a one-dimensional problem using the result of Section \ref{factorisation}, i.e. by considering $\beta \simeq \begin{bmatrix}
0.906 \\
0.424
\end{bmatrix}$ and $V \simeq 1.744$.\\

Regarding the objective function, we consider the following:
\begin{itemize}
  \item Time horizon given by $T = 12\ \textrm{days}$. This horizon ensures convergence towards stationary quotes at time $t=0$ -- see Figures \ref{conv_deltas} and \ref{conv_deltas_1f} below.
  \item $\psi: q \in \mathbb{R}^2 \mapsto \frac \gamma 2 q'\Sigma q$ with $\gamma = 8\cdot10^{-7}\ \textrm{\euro}^{-1}$.
  \item $\ell_d = 0$.
\end{itemize}

\subsubsection{Results with 2 factors}

Since $\theta$ and $\tilde\theta$ are defined on $[0,T] \times \mathbb{R}^2$, a first step for approximating the value functions consists in restricting the state space to a compact set. A traditional way to proceed consists in setting boundary conditions. In what follows, we equivalently impose risk limits in the sense that no trade that would result in an inventory $q \in \mathbb{R}^2$ such that $q'\Sigma q > B$ is admitted, where $B = 2.4\cdot10^{10}$.\footnote{These risk limits can be expressed in the space of factors instead of being expressed with the inventory vector.} We then approximate the solution $\tilde{\theta}$ to~\eqref{eqn:HJBfac} with two factors using a monotone explicit Euler scheme with linear interpolation on a grid of size $141\times141$ for the factors and a discretization of the RFQ size distribution with 4 sizes: $z^1 = 6250$, $z^2 = 12500$, $z^3 = 18750$, and $z^4 = 25000$ assets -- thereafter respectively designated by very small, small, large, and very large size -- with respective probability $p^1 = 0.53$, $p^2 = 0.35$, $p^3 = 0.10$, and $p^4 = 0.02$.\footnote{When there are as many factors as assets, one could directly consider the problem with inventory variables.}\\

The value function (at time $t=0$) as a function of the factors is plotted in Figure \ref{theta_f_3d}. The corresponding value function as a function of the inventory, obtained through linear interpolation is plotted in Figure \ref{theta_3d}.\\

From the value function, we obtain the optimal bid and ask quotes of the market maker as a function of inventories and request size. The optimal bid quotes (at time $t=0$) for asset 1 and asset 2 (in the case of the smallest RFQ size) are plotted in Figures \ref{delta_b_3d_asset_1_size_0} and \ref{delta_b_3d_asset_2_size_0}. The ask quotes are of course symmetric and are not plotted.\\

We see in Figures \ref{delta_b_3d_asset_1_size_0} and \ref{delta_b_3d_asset_2_size_0} that the optimal bid quotes for both assets are increasing functions of both the inventory in asset~1 and asset 2, as expected given the positive value of the correlation parameter $\rho$.\\

As discussed above, we chose $T=12$ days to ensure convergence of the optimal quotes to their stationary values. This is illustrated in Figure \ref{conv_deltas}.\\

\begin{figure}[!h]\centering
\includegraphics[width=\textwidth]{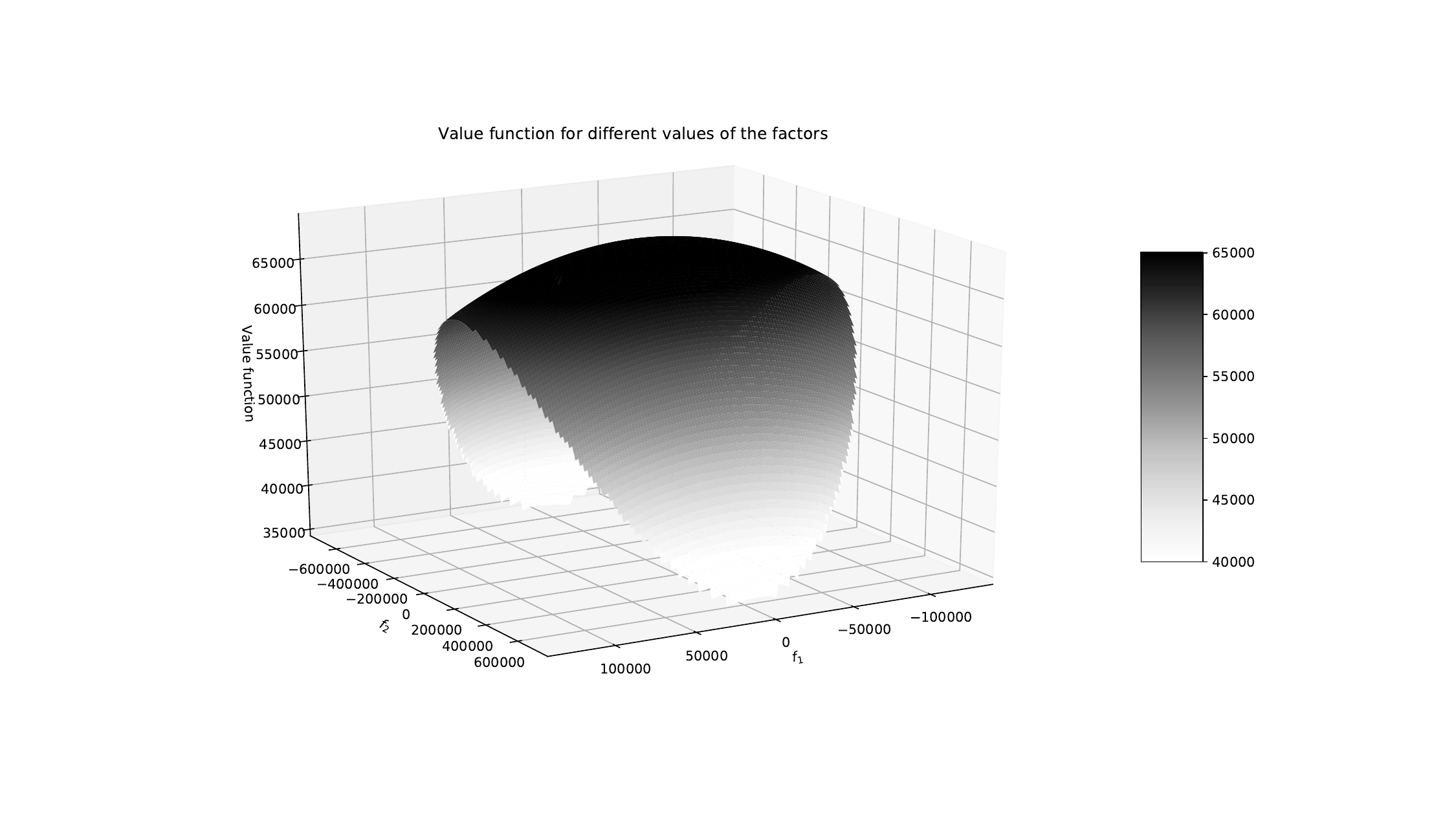}\\
\caption{Value function for different values of the factors.}\label{theta_f_3d}
\end{figure}

\begin{figure}[!h]\centering
\includegraphics[width=0.98\textwidth]{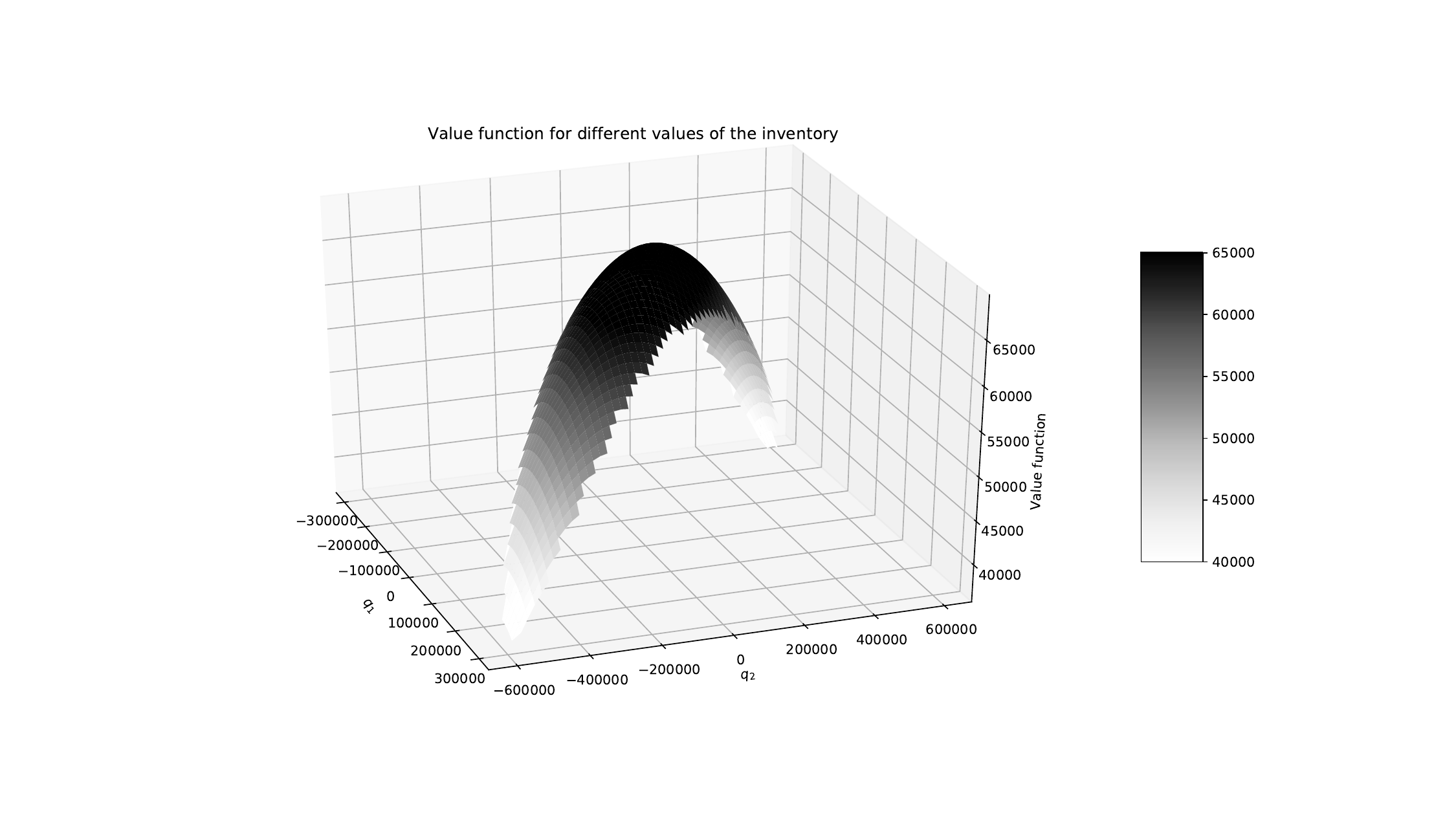}\\
\caption{Value function for different values of the inventory.}\label{theta_3d}
\end{figure}
\vspace{5mm}

\begin{figure}[!h]\centering
\includegraphics[width=0.98\textwidth]{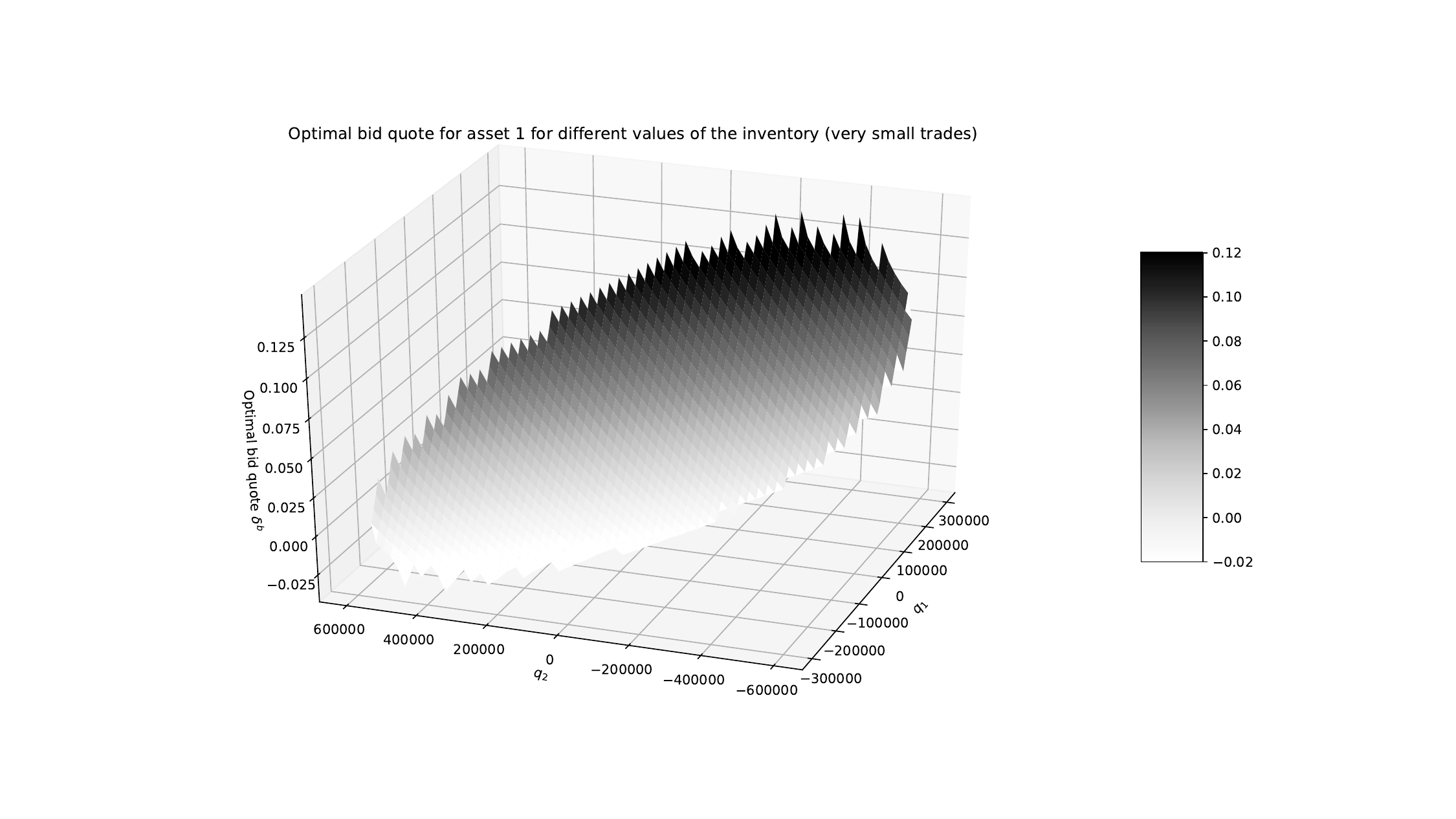}\\
\caption{Optimal bid quote for asset 1 for different values of the inventory (very small trades).}\label{delta_b_3d_asset_1_size_0}
\end{figure}
\vspace{5mm}

\begin{figure}[!h]\centering
\includegraphics[width=0.97\textwidth]{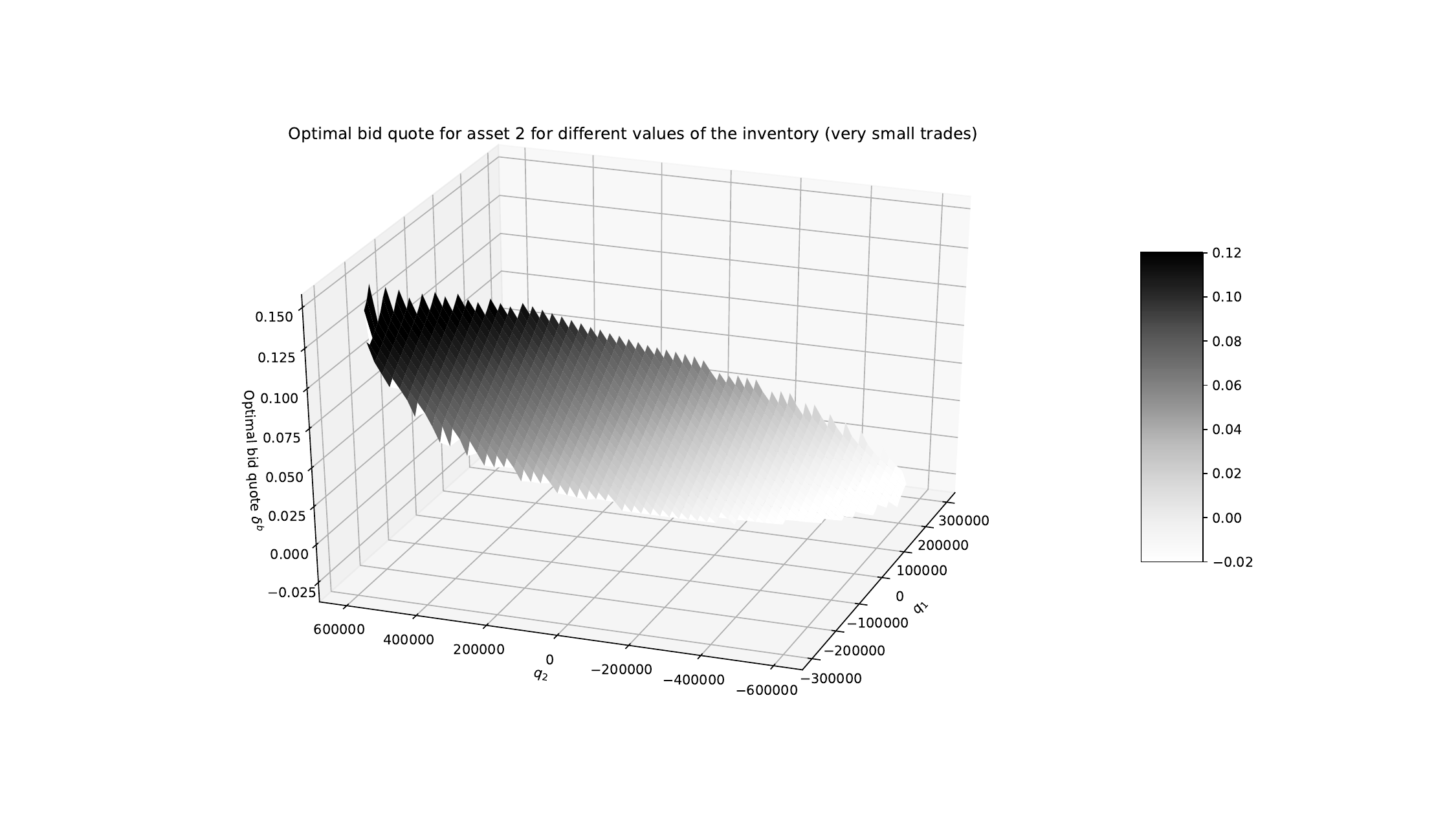}\\
\caption{Optimal bid quote for asset 2 for different values of the inventory (very small trades).}\label{delta_b_3d_asset_2_size_0}
\end{figure}

\begin{figure}[!h]\centering
\includegraphics[width=0.93\textwidth]{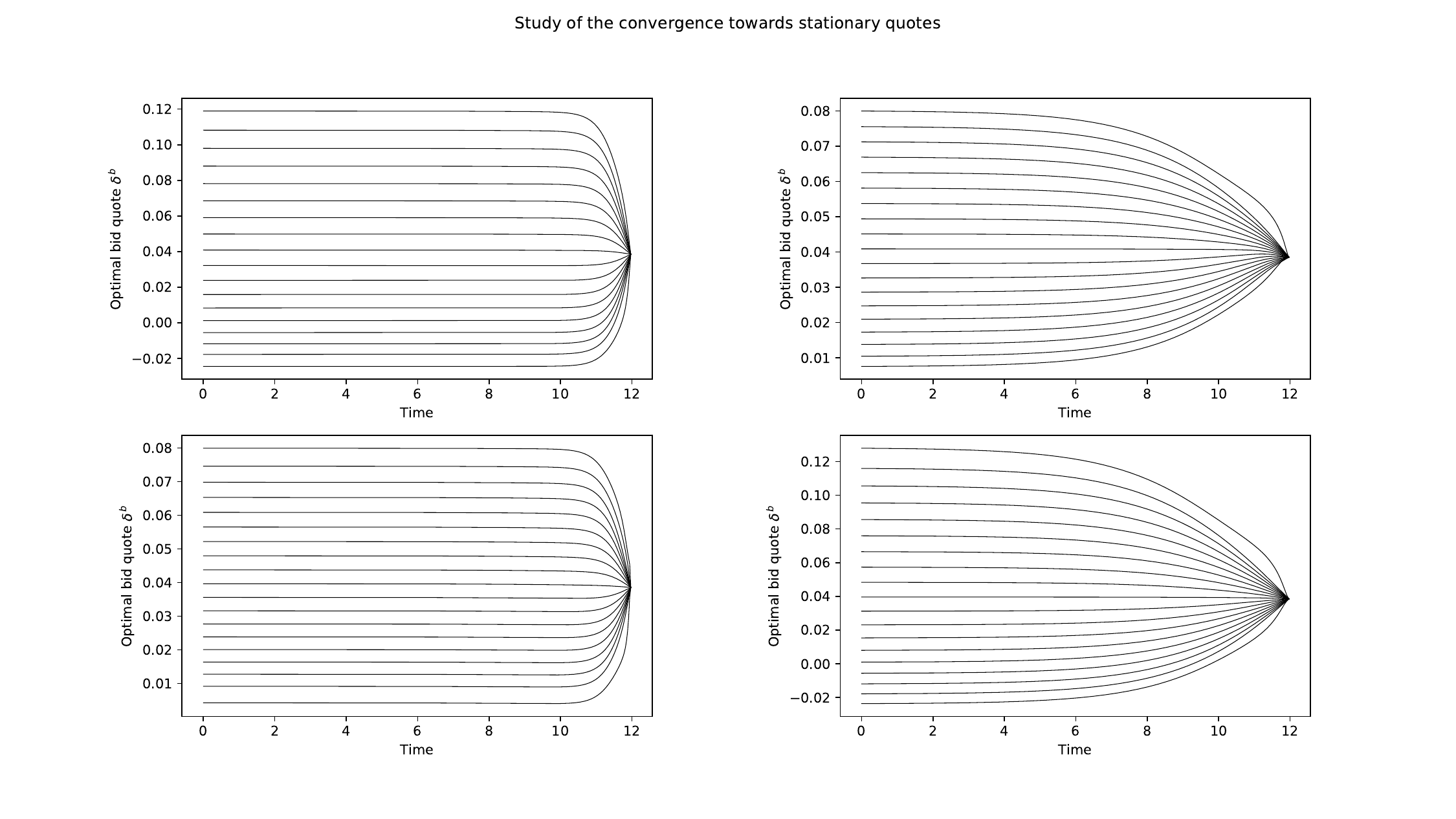}\\
\caption{Optimal bid quotes as a function of time for various values of the factors (very small trades). Top left: Asset~1 when $f^2 =0$. Top right: Asset 1 when $f^1 =0$. Bottom left: Asset 2 when $f^2 =0$. Bottom right: Asset 2 when $f^1 =0$.}\label{conv_deltas}
\end{figure}

To see the impact of the RFQ size on the optimal quotes, we plot in Figure \ref{deltas_asset_1_q_1_different_sizes} the four functions $$q^1 \mapsto \bar{\delta}^{1,b}(0,q^1,0,z^k), k \in \{1,\ldots,4\},$$ in Figure \ref{deltas_asset_1_q_2_different_sizes} the four functions $$q^2 \mapsto \bar{\delta}^{1,b}(0,0,q^2,z^k), k \in \{1,\ldots,4\},$$ and in Figure \ref{deltas_asset_1_major_axis_different_sizes} the four functions $$q^1 \mapsto \bar{\delta}^{1,b}(0,q^1,q_{\text{major axis}}^2(q^1),z^k), k \in \{1,\ldots,4\},$$ where $(q^1,q_{\text{major axis}}^2(q^1))$ spans the major axis of the ellipse of authorized inventories.\\

Likewise for asset 2: we plot in Figure \ref{deltas_asset_2_q_2_different_sizes} the four functions $$q^2 \mapsto \bar{\delta}^{2,b}(0,0,q^2,z^k), k \in \{1,\ldots,4\},$$ in Figure \ref{deltas_asset_2_q_1_different_sizes} the four functions $$q^1 \mapsto \bar\delta^{2,b}(0,q^1,0,z^k), k \in \{1,\ldots,4\},$$ and in Figure \ref{deltas_asset_2_major_axis_different_sizes} the four functions $$q^1 \mapsto \bar\delta^{2,b}(0,q^1,q_{\text{major axis}}^2(q^1),z^k), k \in \{1,\ldots,4\},$$ where $(q^1,q_{\text{major axis}}^2(q^1))$ spans the major axis of the ellipse of authorized inventories.\\

\begin{rem}
The reader should be aware that the $x$-axis does not span the same range of inventory on all graphs even when the variable is the same. This is due to the elliptic geometry of the set of authorized values for the inventory.\\
\end{rem}

We see that accounting for the size of RFQs significantly impacts the optimal quotes of asset 1. This is less the case for asset 2  (this difference is due to the fact that the volatility of asset 1 is twice that of asset 2). In all cases, the monotonicity is unsurprising.\footnote{Boundary effects related to impossible interpolation explain the surprising position of some extreme points.}

\begin{figure}[!h]\centering
\includegraphics[width=0.94\textwidth]{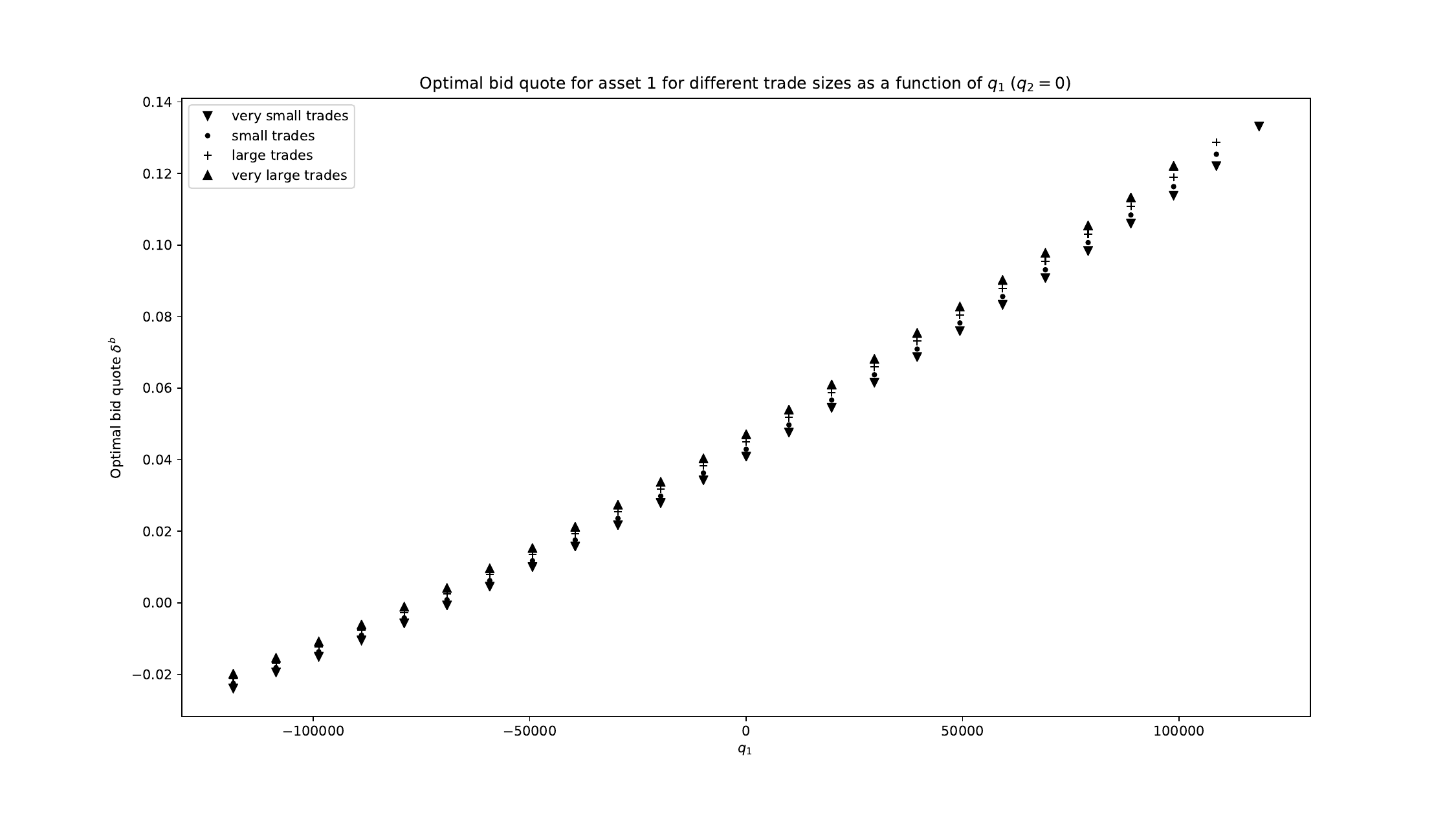}\\
\caption{Optimal bid quote for asset 1 for different trade sizes as a function of $q_1$ ($q_2 = 0$).}\label{deltas_asset_1_q_1_different_sizes}
\end{figure}

\begin{figure}[!h]\centering
\includegraphics[width=0.97\textwidth]{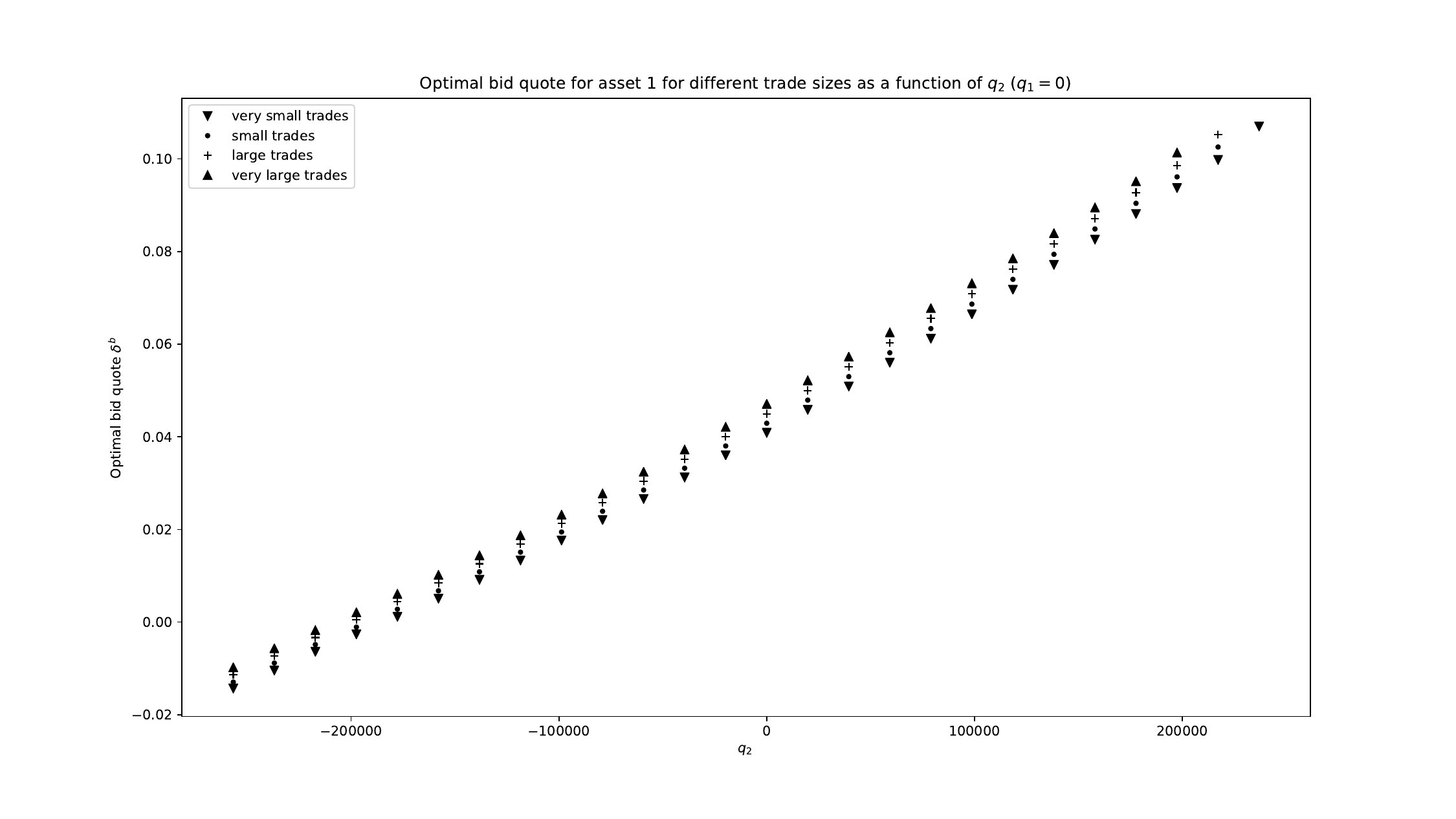}\\
\caption{Optimal bid quote for asset 1 for different trade sizes as a function of $q_2$ ($q_1 = 0$).}\label{deltas_asset_1_q_2_different_sizes}
\end{figure}

\begin{figure}[!h]\centering
\includegraphics[width=0.97\textwidth]{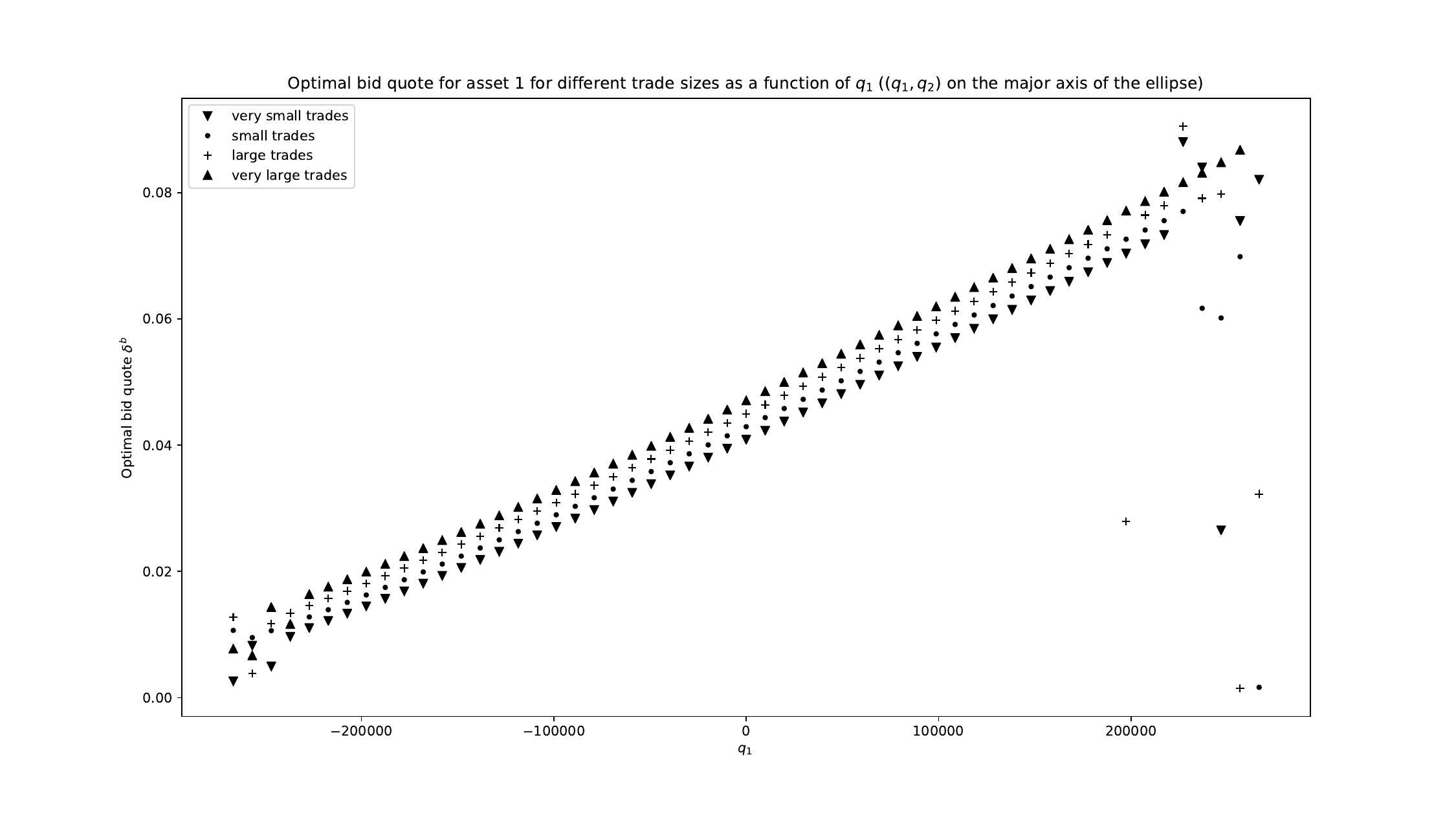}\\
\caption{Optimal bid quote for asset 1 for different trade sizes as a function of $q_1$ ($(q_1,q_2)$ on the major axis of the ellipse).}\label{deltas_asset_1_major_axis_different_sizes}
\end{figure}

\begin{figure}[!h]\centering
\includegraphics[width=0.97\textwidth]{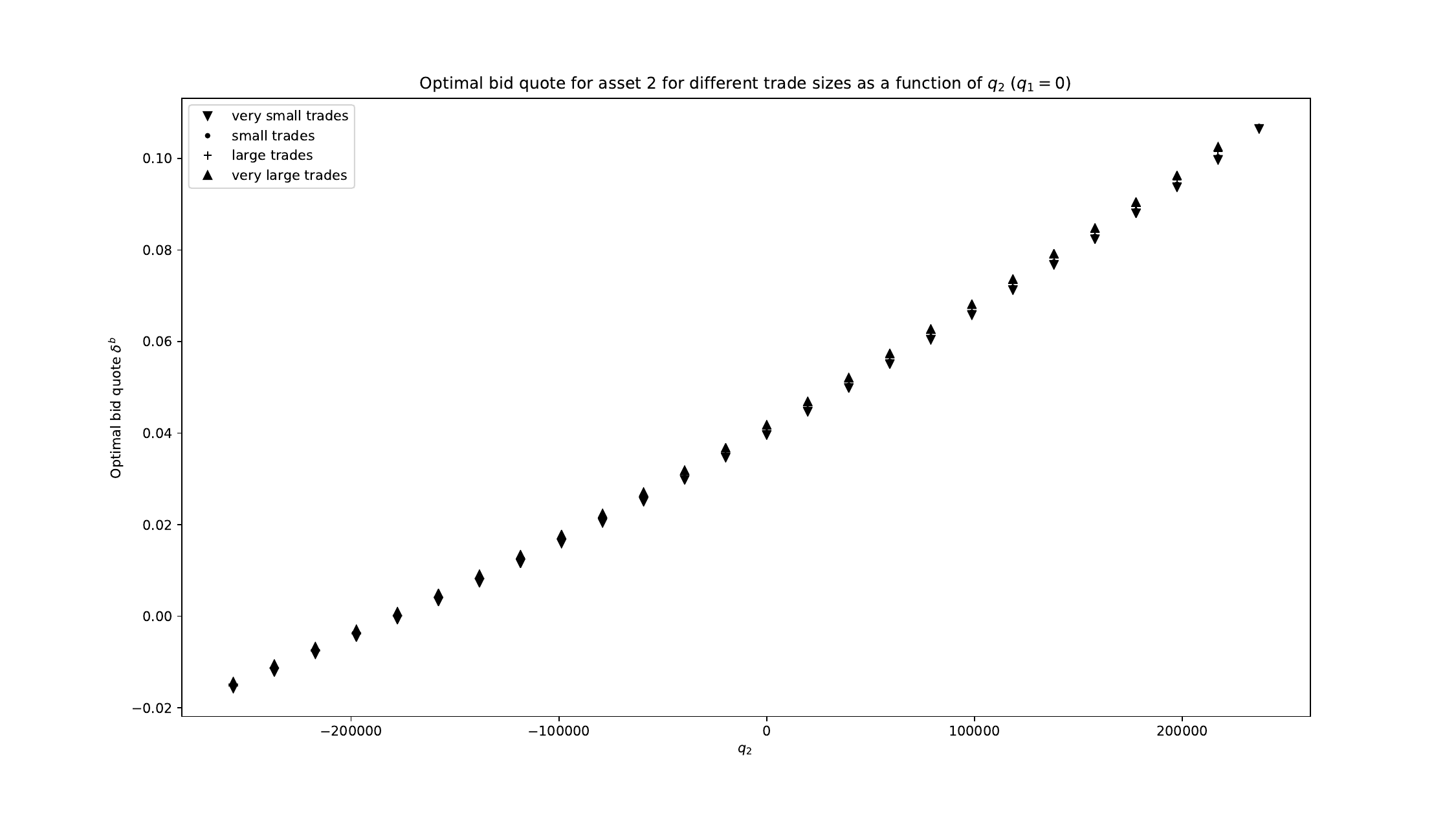}\\
\caption{Optimal bid quote for asset 2 for different trade sizes as a function of $q_2$ ($q_1 = 0$).}\label{deltas_asset_2_q_2_different_sizes}
\end{figure}

\begin{figure}[!h]\centering
\includegraphics[width=0.97\textwidth]{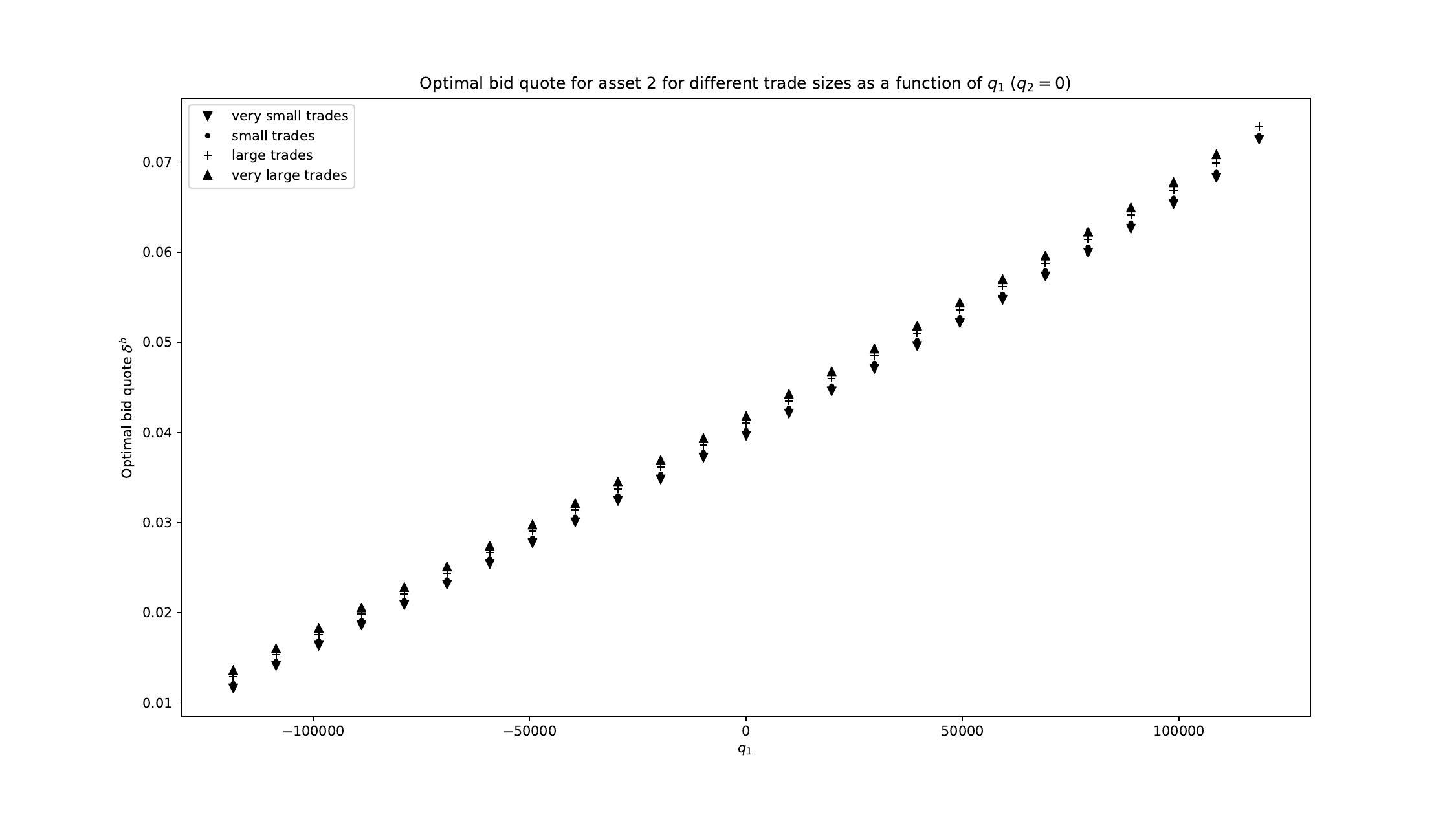}\\
\caption{Optimal bid quote for asset 2 for different trade sizes as a function of $q_1$ ($q_2 = 0$).}\label{deltas_asset_2_q_1_different_sizes}
\end{figure}

\begin{figure}[!h]\centering
\includegraphics[width=\textwidth]{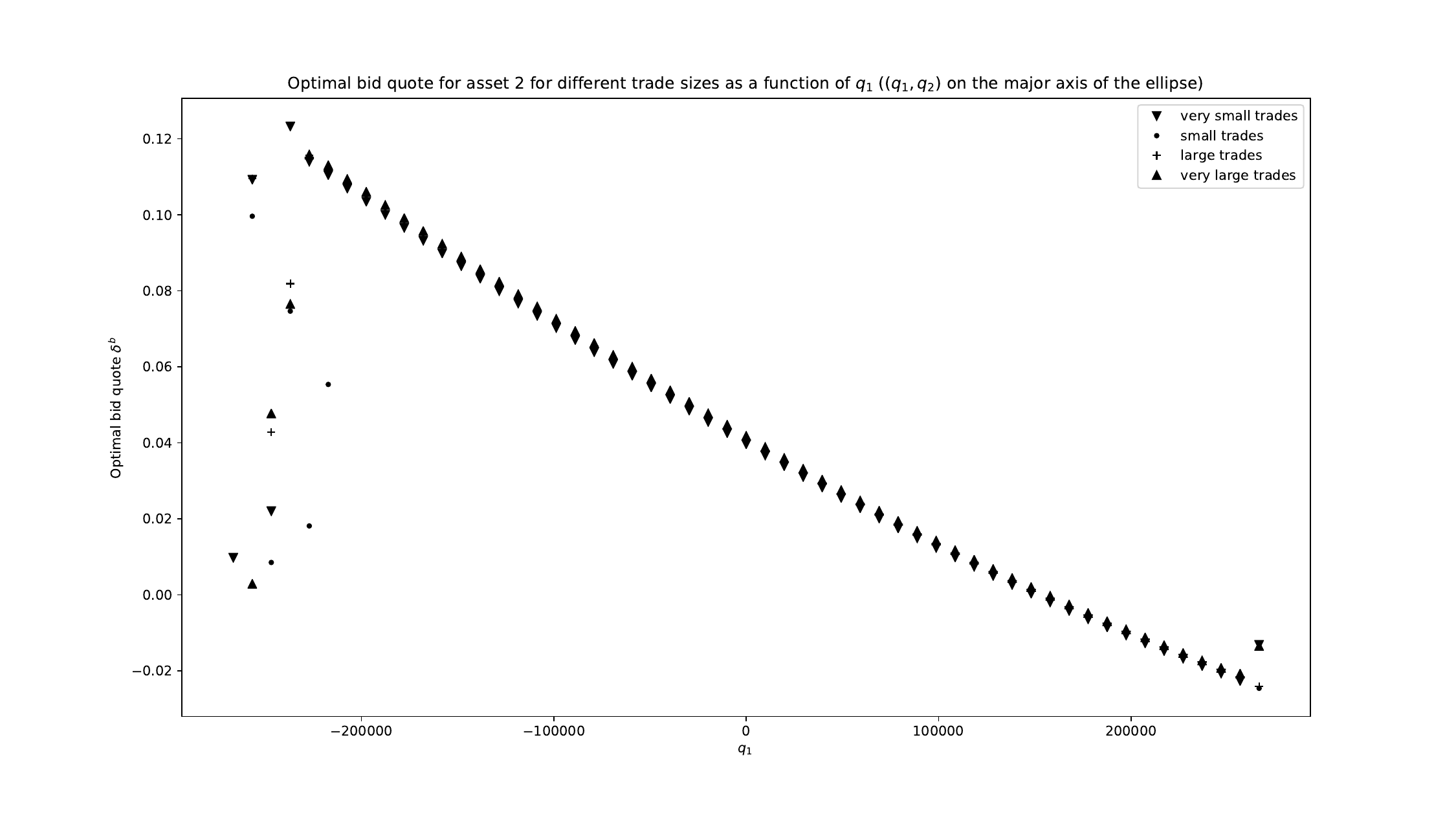}\\
\caption{Optimal bid quote for asset 2 for different trade sizes as a function of $q_1$ ($(q_1,q_2)$ on the major axis of the ellipse).}\label{deltas_asset_2_major_axis_different_sizes}
\end{figure}
In order to check that the value of $B$ defining the risk limits does not have a significant impact on our numerical approximation, we carried out a Monte-Carlo simulation with 2000 trajectories starting from zero inventory, using the optimal quotes. The distribution of inventory is plotted in Figure \ref{q_opt_dist}.\footnote{The shades of gray are in logarithmic scale.} We clearly see that the ellipse of authorized inventory is wide enough to have little influence on the outcome.\\ \\
\iffalse
Our 2000 simulations also enable us to illustrate the distribution of the PnL at time $T$: $\text{PnL}_T = X_{T} + \sum_{i=1}^{2} q^{i}_{T}S^{i}_{T}$. The histogram is plotted in Figure \ref{pnl_distrib_2d2f}.\\
\fi
The statistics associated with our simulations are documented in Table 1: the average PnL at time $T$, the standard deviation of that PnL at date $T$, the part of that standard deviation not related to market risk -- i.e. only related to the randomness of RFQs --,\footnote{Using the law of total variance, it is easy to see that $$\mathbb{V}(\text{PnL}_T) = \mathbb{E}\left[\int_0^T q_t'\Sigma q_t dt\right] + \mathbb{V}\left(\int_0^T \sum_{i=1}^{2} \int_{\mathbb{R}_{+}^{*}} \left(\delta^{i,b}(t,z)zJ^{i,b}(dt,dz) + \delta^{i,a}(t,z)zJ^{i,a}(dt,dz) \right)\right).$$ This formula enables to distinguish the part of the variance of the PnL at time $T$ coming from market risk (the first term) from that coming from the randomness of RFQs (the second term).} and the estimated value of the objective function, i.e. the empirical mean of $\text{PnL}_T - \frac \gamma 2 \int_0^T q_t'\Sigma q_t dt$.\\
\vspace{0.5cm}

\begin{table}[h!]
\centering
\begin{tabular}{|c|c|c|c|}
  \hline
  Mean PnL & Stdev PnL & Stdev coming from RFQs & Objective function \\
  \hline
  72081 & 80432 & 5959 & 69293\\
  \hline
\end{tabular}
\label{pnl_distrib_2d2f_table}
\caption{Statistics associated with our 2000 simulations starting from zero inventory (with optimal quotes).}
\end{table}

\vspace{1cm}

\begin{figure}[!h]\centering
\includegraphics[width=0.98\textwidth]{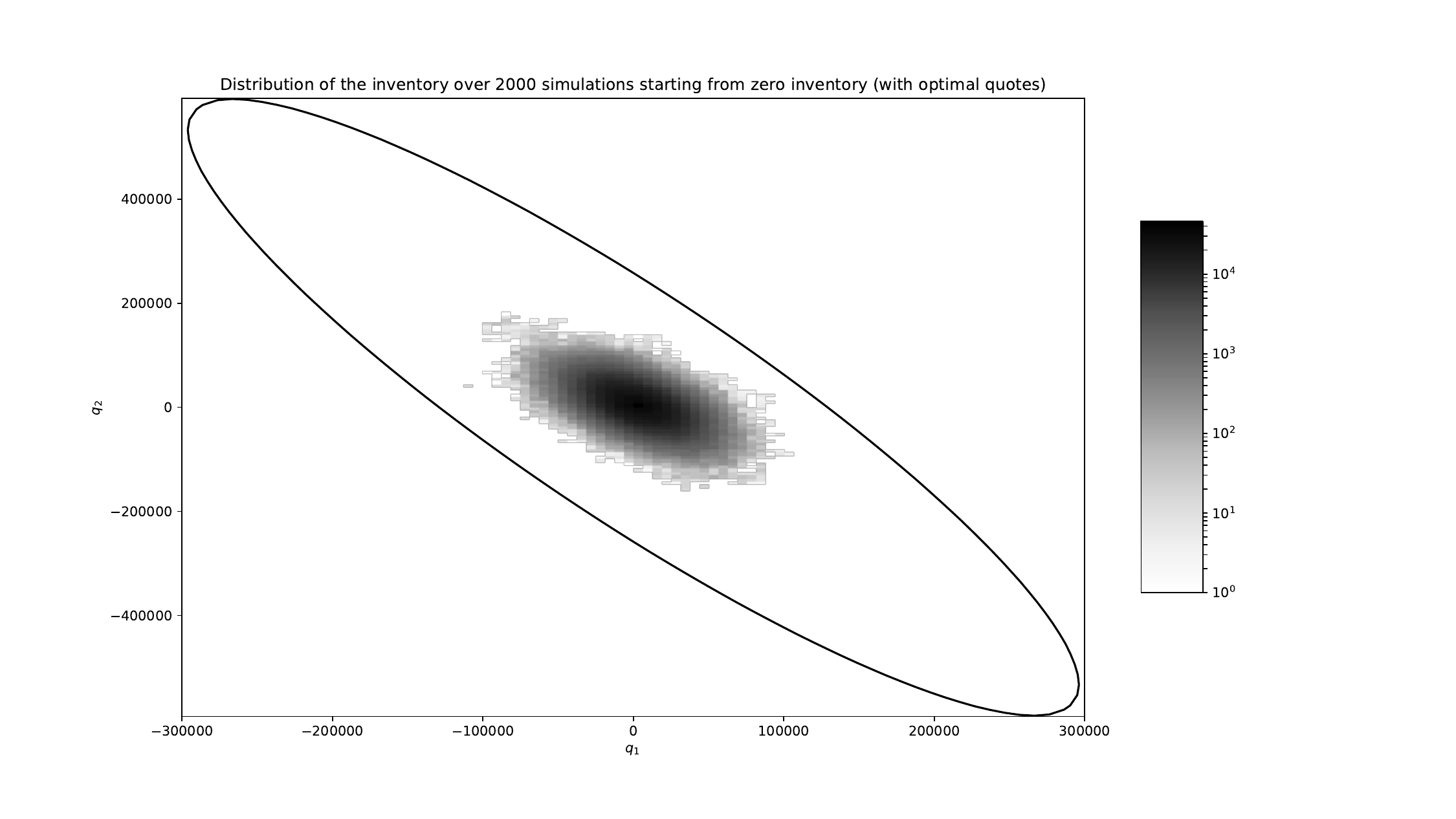}\\
\caption{Distribution of the inventory over 2000 simulations starting from zero inventory (with optimal quotes).}\label{q_opt_dist}
\end{figure}
\iffalse
\begin{figure}[!h]\centering
\includegraphics[width=0.98\textwidth]{pnl_distrib_2d2f.pdf}\\
\caption{Histogram of the PnL over 2000 simulations starting from zero inventory (with optimal quotes).}\label{pnl_distrib_2d2f}
\end{figure}
\fi
These figures have to be compared with those associated with a basic naive strategy. A basic strategy consists, for each asset and side, in always quoting the same ``myopic'' quote that maximizes the expected instantaneous PnL. In other words, these myopic quotes are defined, for all $i \in\{1, \ldots, d\}$, by
$$\delta^{i,b}_{\text{myopic}} = \text{argmax\ } \delta \Lambda^{i,b}(\delta) \quad \text{and} \quad \delta^{i,a}_{\text{myopic}} = \text{argmax\ } \delta \Lambda^{i,a}(\delta).$$ In our case, the myopic quotes are all equal to $0.03854\ \textrm{\euro}$.\\

We carried out 2000 simulations using these quotes with the same source of randomness as above. The distribution of inventory is plotted in Figure \ref{q_myopic_dist}, and the statistics associated with our simulations are documented in Table 2.\\

\begin{table}[h!]
\centering
\begin{tabular}{|c|c|c|c|}
  \hline
  Mean PnL & Stdev PnL & Stdev coming from RFQs & Objective function \\
  \hline
  73410 & 265906 & 6211 & 43953\\
  \hline
\end{tabular}
\label{pnl_distrib_2dmyopic_table}
\caption{Statistics associated with our 2000 simulations starting from zero inventory (with myopic quotes).}
\end{table}

With these figures, we clearly see that the main source of risk is market risk and not the risk associated with the randomness of RFQs. This justifies our choice of objective function that only penalizes the part of the variance coming from market risk.\\

We also clearly see that the use of the optimal quotes drastically reduces the variance of the PnL and results in a high value of the objective function. More precisely, although the use of optimal quotes reduces the average PnL from around $73410$ to around $72081$, it enables a reduction by a factor $3$ of the standard deviation of the PnL from around $265906$ to around $80432$, hence a major increase of the objective function from around $43953$ to around $69293$ (a figure consistent with the maximum of the value function plotted in Figure \ref{theta_f_3d} which is approximately 69174).\\

\begin{figure}[!h]\centering
\includegraphics[width=0.98\textwidth]{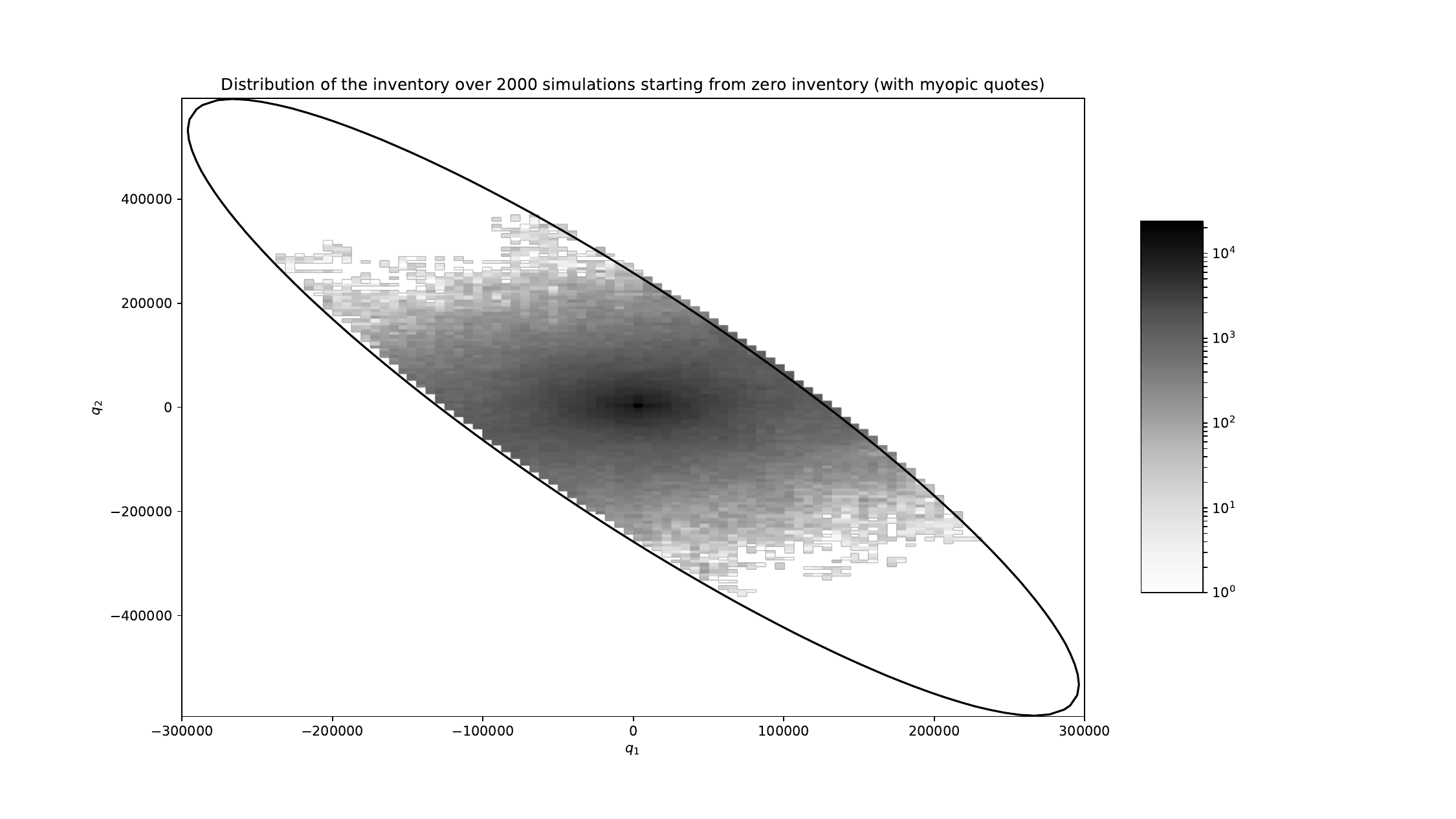}\\
\caption{Distribution of the inventory over 2000 simulations starting from zero inventory (with myopic quotes).}\label{q_myopic_dist}
\end{figure}
\iffalse
\begin{figure}[!h]\centering
\includegraphics[width=\textwidth]{pnl_distrib_2dmyopic.pdf}\\
\caption{Histogram of the PnL over 2000 simulations starting from zero inventory (with myopic quotes).}\label{pnl_distrib_2dmyopic}
\end{figure}
\fi

\subsubsection{Results with the one-factor model and comparison}

Let us now compare the results with two factors to the results with one factor, i.e. when the smallest eigenvalue of $\Sigma$ is replaced by $0$.\\

As above, we start with an approximation of the solution $\tilde{\theta}$ to \eqref{eqn:HJBfac} with one factor. We used a monotone explicit Euler scheme on a grid of size $141$ for the unique factor and the same discretization (with 4 sizes) as in the previous paragraphs for the RFQ size distribution. The set of authorized inventory $\{q \in \mathbb R^2 | q'\Sigma q  \le B\}$ is of course replaced by the set $\{q \in \mathbb R^2 | (\beta'q)'V (\beta'q) \le B\}$ which corresponds, in terms of the unique factor, to the interval $\left[-\sqrt {\frac BV}, \sqrt {\frac BV}\right]$.\\

 The value function (at time $t=0$) as a function of the inventory, obtained through linear interpolation is plotted in Figure \ref{theta_3d_1f}. The difference between the one-factor value function and the two-factor one is plotted in Figure \ref{diff_theta_3d_1f}.\\

We see that the value function in the one-factor case is above that of the two-factor case. This comes from the fact that not all the risk is taken into account in the one-factor case. We also see that the difference between the two value functions is very large at the two extremes of the major axis of the ellipse. This comes from the fact that the market maker using the one-factor model believes that positions close to the major axis of the ellipse are associated with low risk whereas this is less and less the case as the inventory in each asset increases in absolute value.\\

As in the two-factor case, we deduce from the value function the optimal bid and ask quotes of the market maker (at time $t=0$) as a function of inventory and request size. The optimal bid quotes for asset 1 and asset 2 (in the case of the smallest RFQ size) are plotted in Figures \ref{delta_b_3d_asset_1_size_0_1f} and \ref{delta_b_3d_asset_2_size_0_1f}. The differences between the optimal bid quotes in the one-factor case and two-factor case (in the case of the smallest RFQ size) are plotted in Figure \ref{diff_delta_b_3d_asset_1_size_0_1f} for asset 1 and Figure \ref{diff_delta_b_3d_asset_2_size_0_1f} for asset 2.\\

We clearly see that the larger (in absolute value) the inventory in each asset, the larger the difference in optimal quotes between the exact model and the one-factor approximation. This is in line with expectation.\\
\vspace{0.5cm}
\iffalse
\begin{figure}[!h]\centering
\includegraphics[width=0.98\textwidth]{theta_f_1f.pdf}\\
\caption{Value function in the one-factor case for different values of the factor.}\label{theta_f_1f}
\end{figure}
\fi

\begin{figure}[!h]\centering
\includegraphics[width=0.98\textwidth]{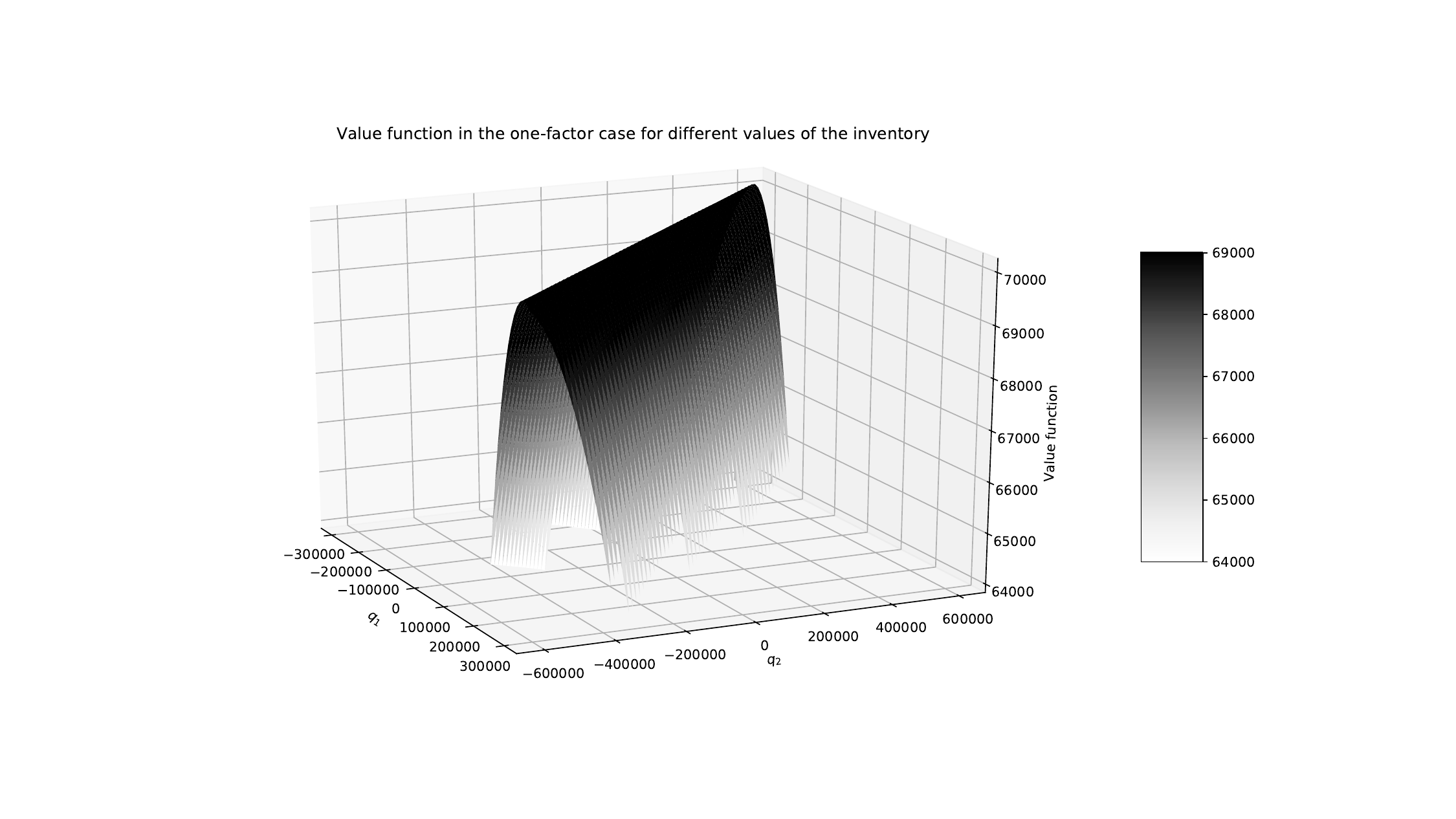}\\
\caption{Value function in the one-factor case for different values of the inventory.}\label{theta_3d_1f}
\end{figure}

\begin{figure}[!h]\centering
\includegraphics[width=0.98\textwidth]{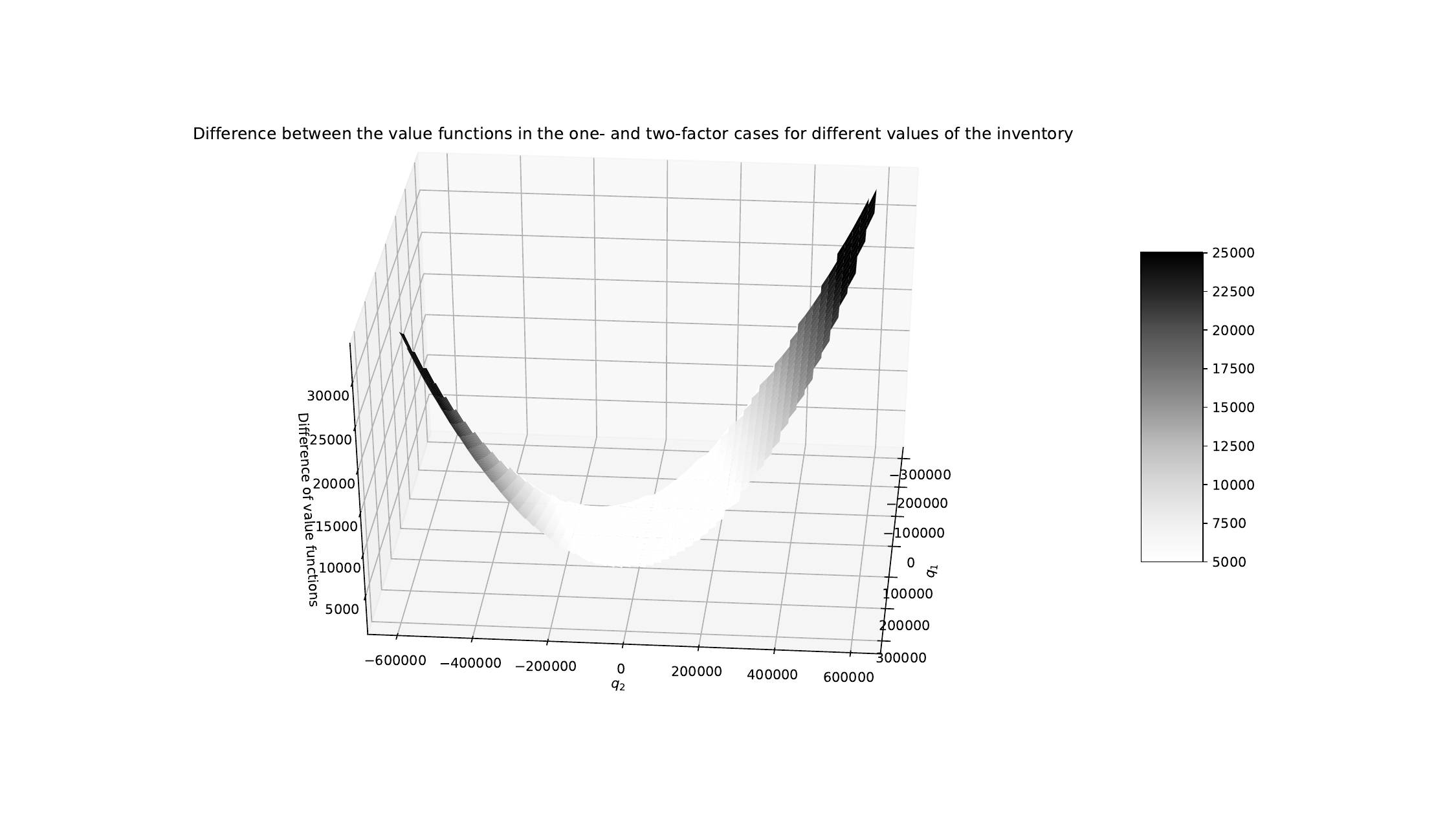}\\
\caption{Difference between the value functions in the one- and two-factor cases for different values of the inventory.}\label{diff_theta_3d_1f}
\end{figure}

\begin{figure}[!h]\centering
\includegraphics[width=0.98\textwidth]{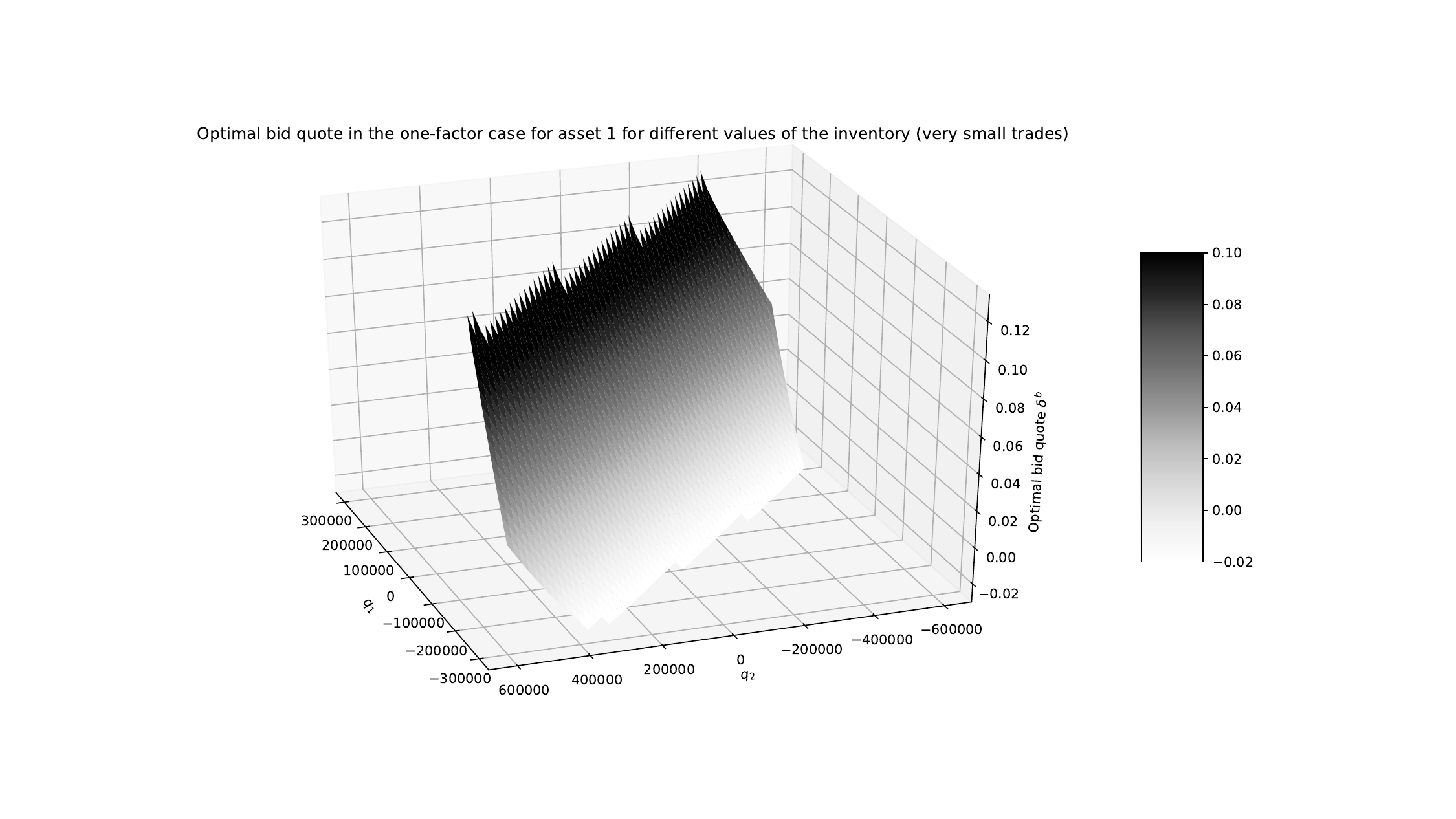}\\
\caption{Optimal bid quote in the one-factor case for asset 1 for different values of the inventory (very small trades).}\label{delta_b_3d_asset_1_size_0_1f}
\end{figure}

\begin{figure}[!h]\centering
\includegraphics[width=0.97\textwidth]{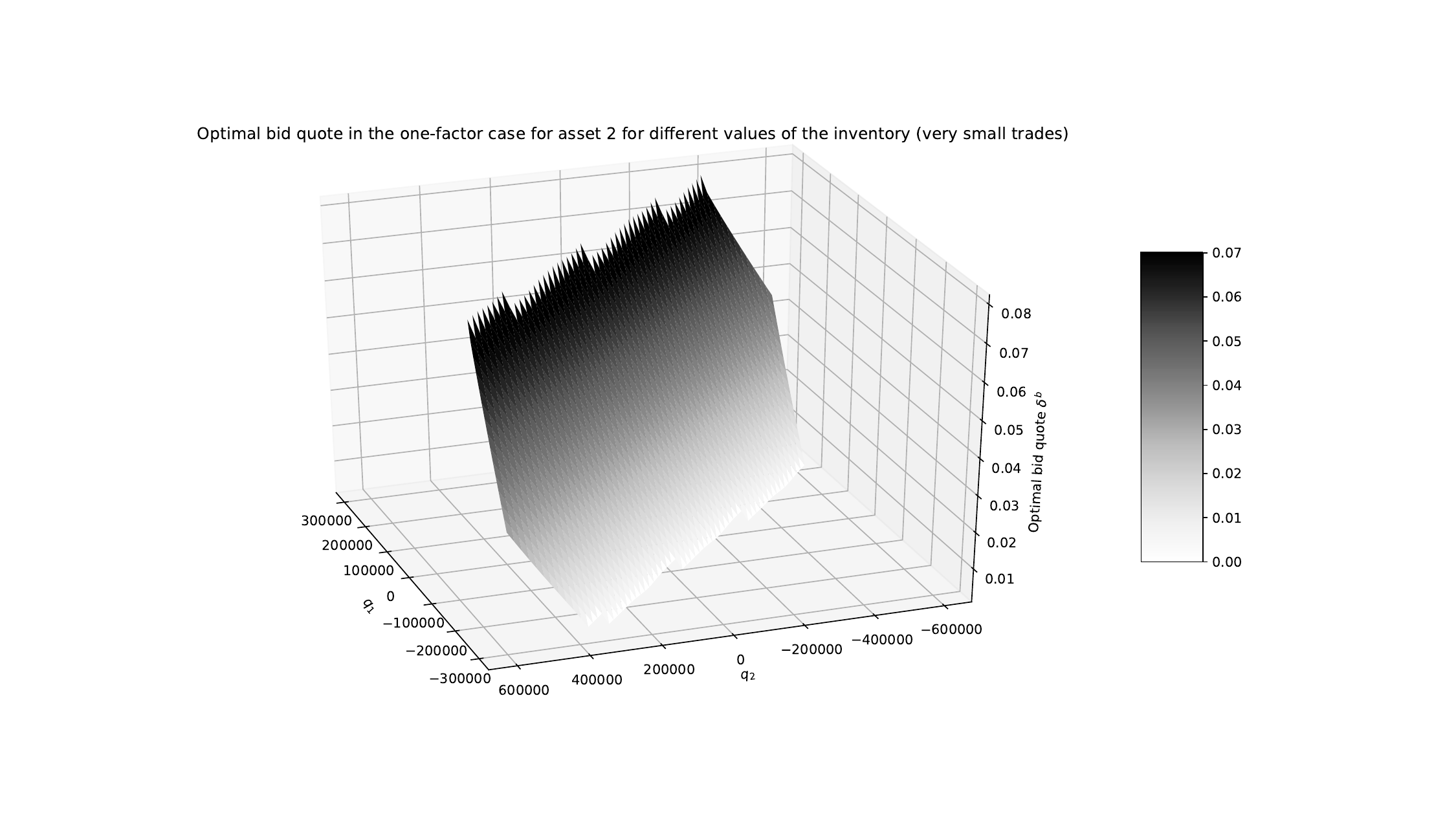}\\
\caption{Optimal bid quote in the one-factor case for asset 2 for different values of the inventory (very small trades).}\label{delta_b_3d_asset_2_size_0_1f}
\end{figure}

\begin{figure}[!h]\centering
\includegraphics[width=0.95\textwidth]{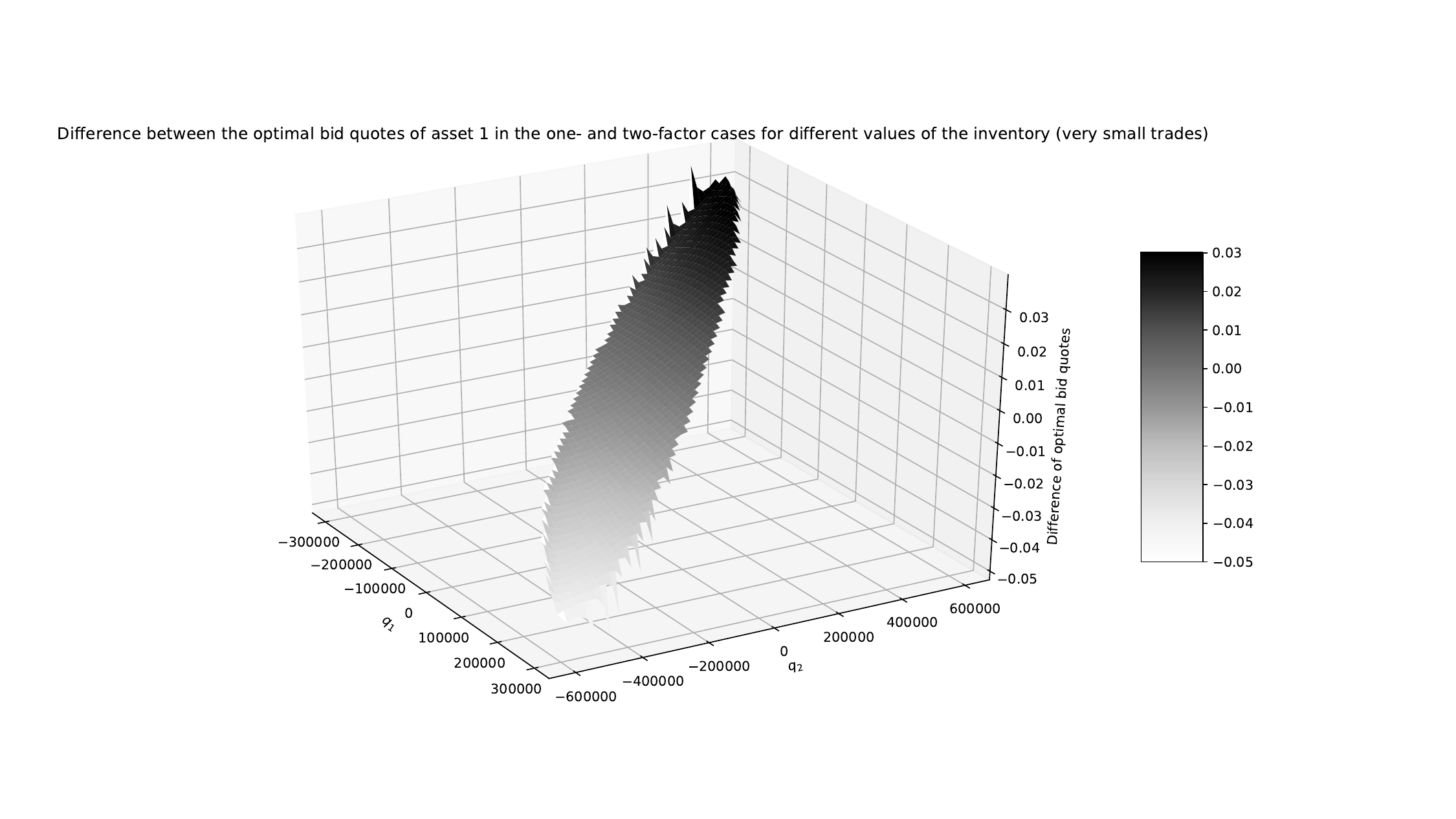}\\
\caption{Difference between the optimal bid quotes of asset 1 in the one- and two-factor cases for different values of the inventory (very small trades).}\label{diff_delta_b_3d_asset_1_size_0_1f}
\end{figure}

\begin{figure}[!h]\centering
\includegraphics[width=0.95\textwidth]{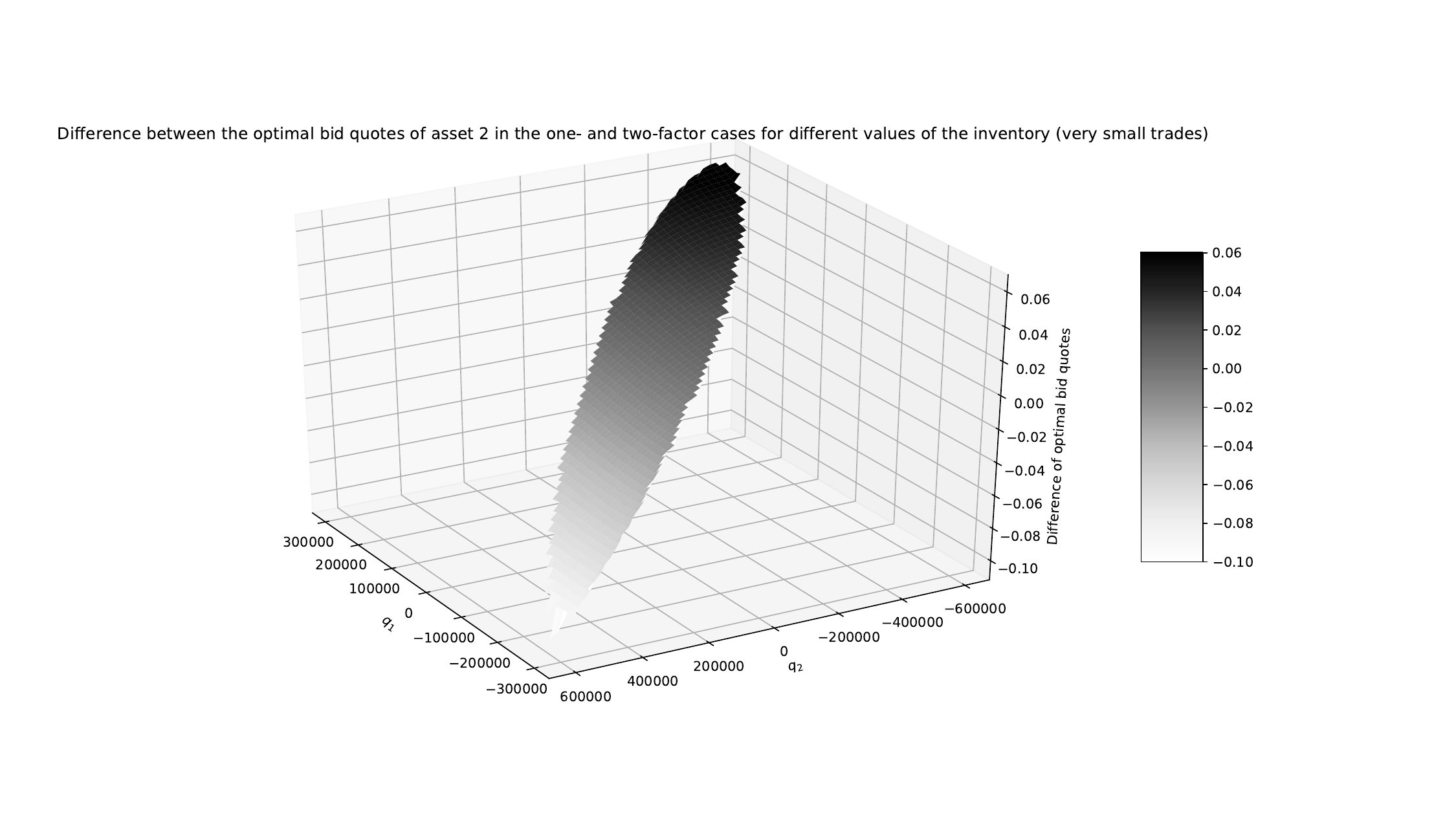}\\
\caption{Difference between the optimal bid quotes of asset 2 in the one- and two-factor cases for different values of the inventory (very small trades).}\label{diff_delta_b_3d_asset_2_size_0_1f}
\end{figure}

As in the two-factor case, we checked the convergence of the quotes towards their asymptotic values. We see in Figure \ref{conv_deltas_1f} that the optimal quotes reach their stationary values in less than $3$ days, far before the $12$ days.\\

\begin{figure}[!h]\centering
\includegraphics[width=0.8\textwidth]{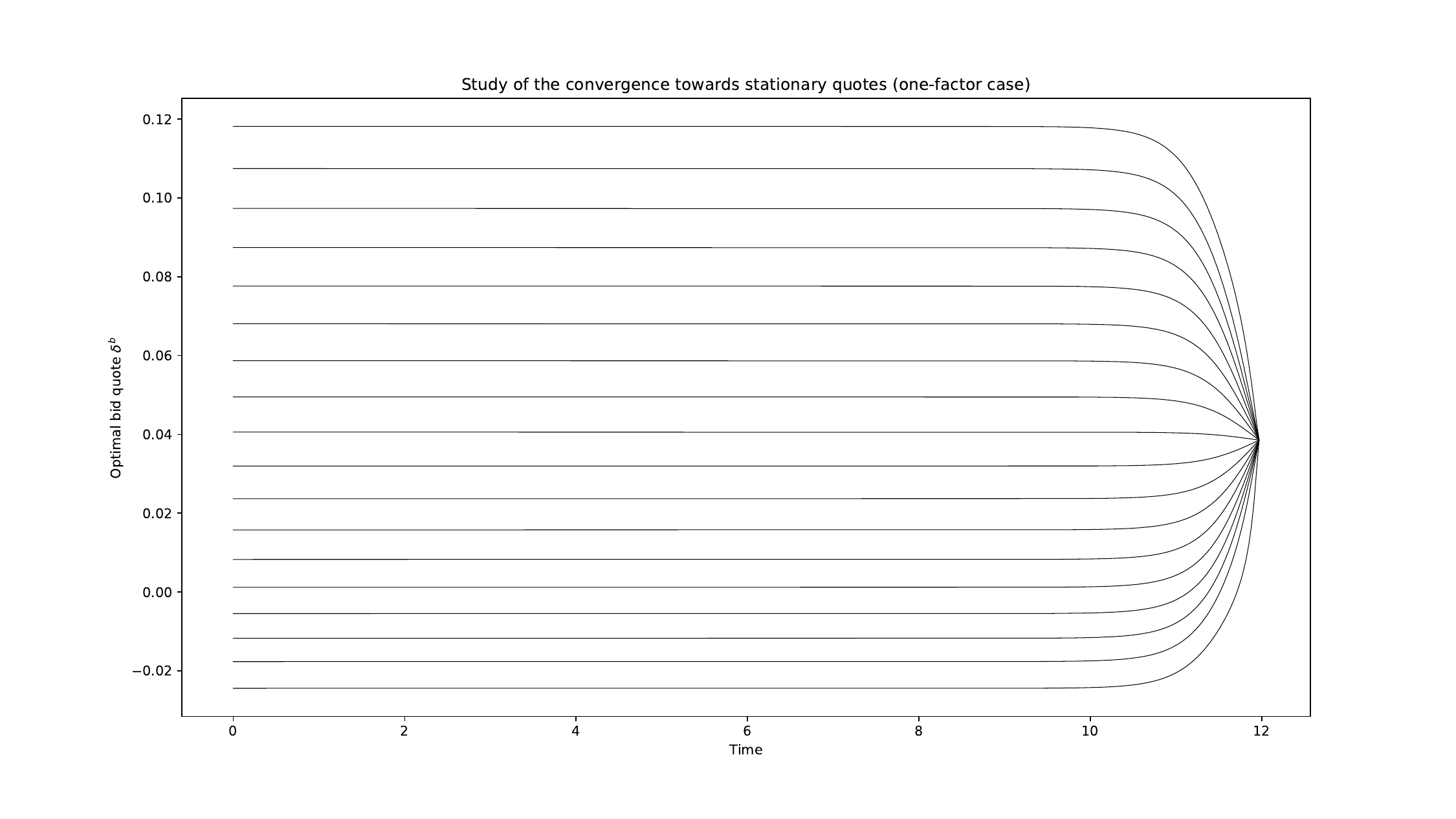}\\
\caption{Optimal bid quote of asset 1 as a function of time for various values of the factor.}\label{conv_deltas_1f}
\end{figure}

To better compare the quotes and see the impact of RFQ size, we plot in Figures \ref{diff_deltas_asset_1_q_1_different_sizes_1f}, \ref{diff_deltas_asset_1_q_2_different_sizes_1f}, and \ref{diff_deltas_asset_1_major_axis_different_sizes_1f} the optimal bid quotes of asset 1 when $q_2=0$ for different values of $q_1$, when $q_1=0$ for different values of $q_2$, and when $(q_1, q_2)$  spans the major axis of the ellipse of authorized inventory.\\

\begin{figure}[!h]\centering
\includegraphics[width=0.9\textwidth]{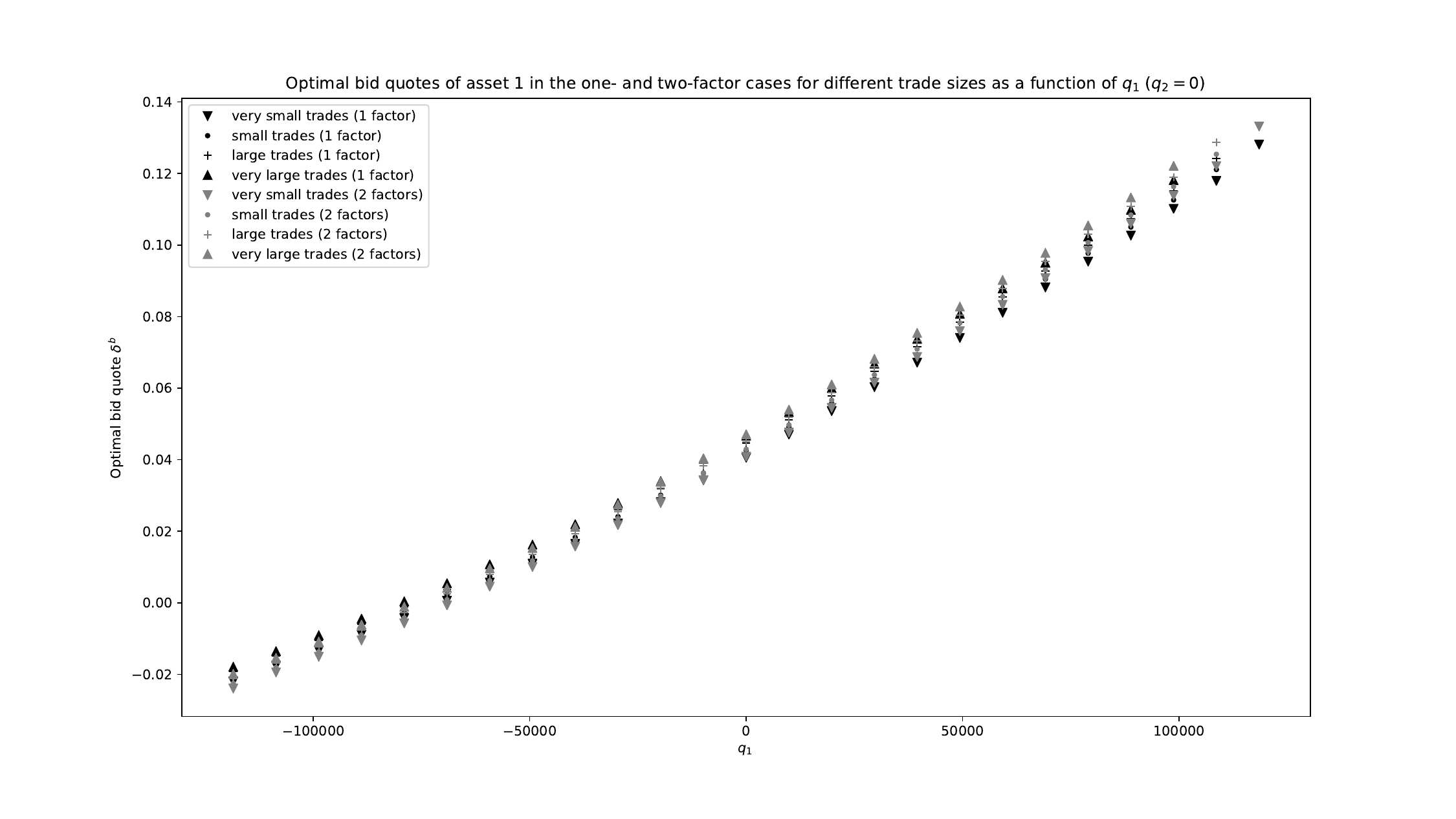}\\
\caption{Optimal bid quote for asset 1 for different trade sizes as a function of $q_1$ ($q_2 = 0$).}\label{diff_deltas_asset_1_q_1_different_sizes_1f}
\end{figure}

\begin{figure}[!h]\centering
\includegraphics[width=0.9\textwidth]{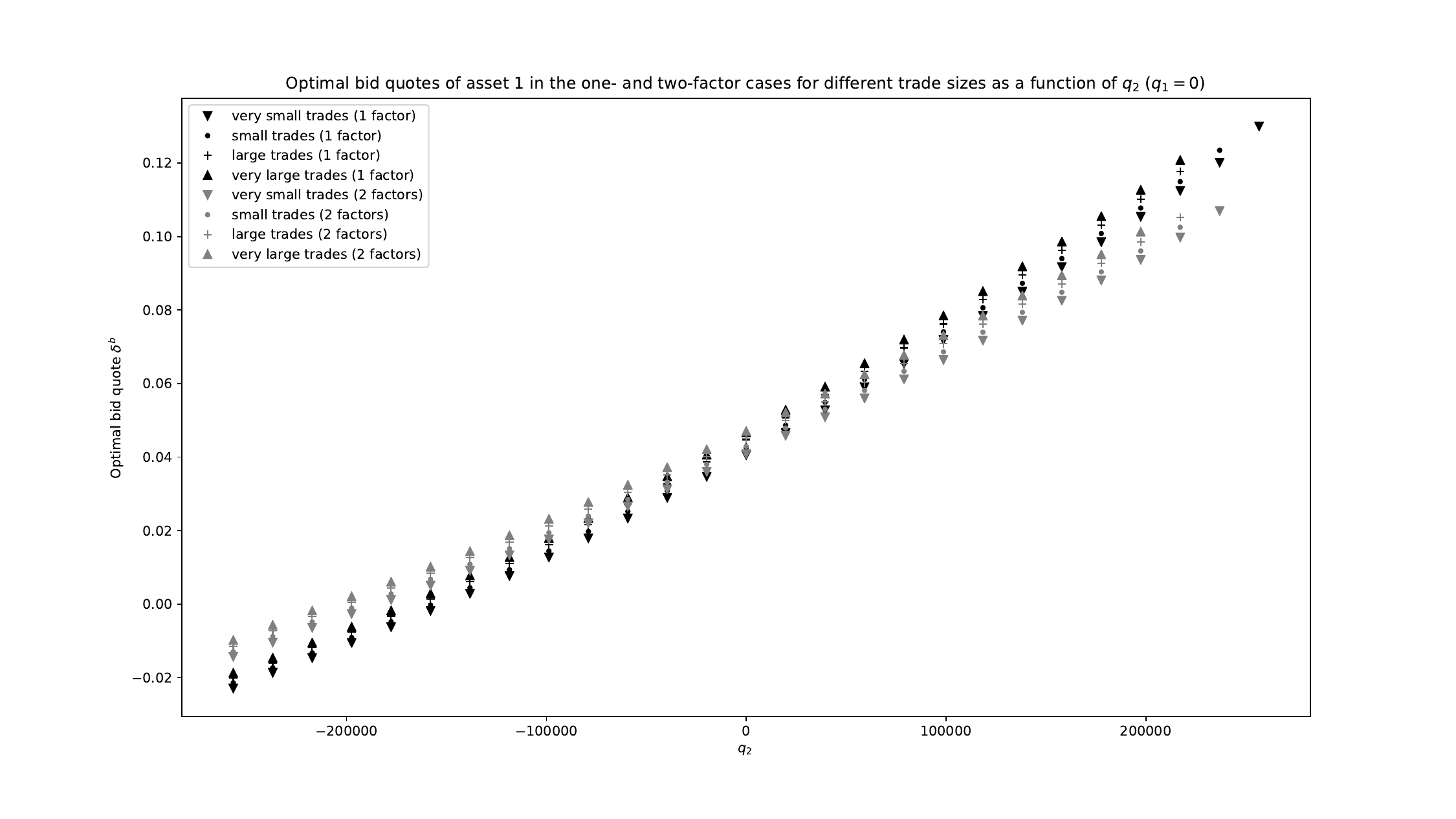}\\
\caption{Optimal bid quote for asset 1 for different trade sizes as a function of $q_2$ ($q_1 = 0$).}\label{diff_deltas_asset_1_q_2_different_sizes_1f}
\end{figure}

\begin{figure}[!h]\centering
\includegraphics[width=0.9\textwidth]{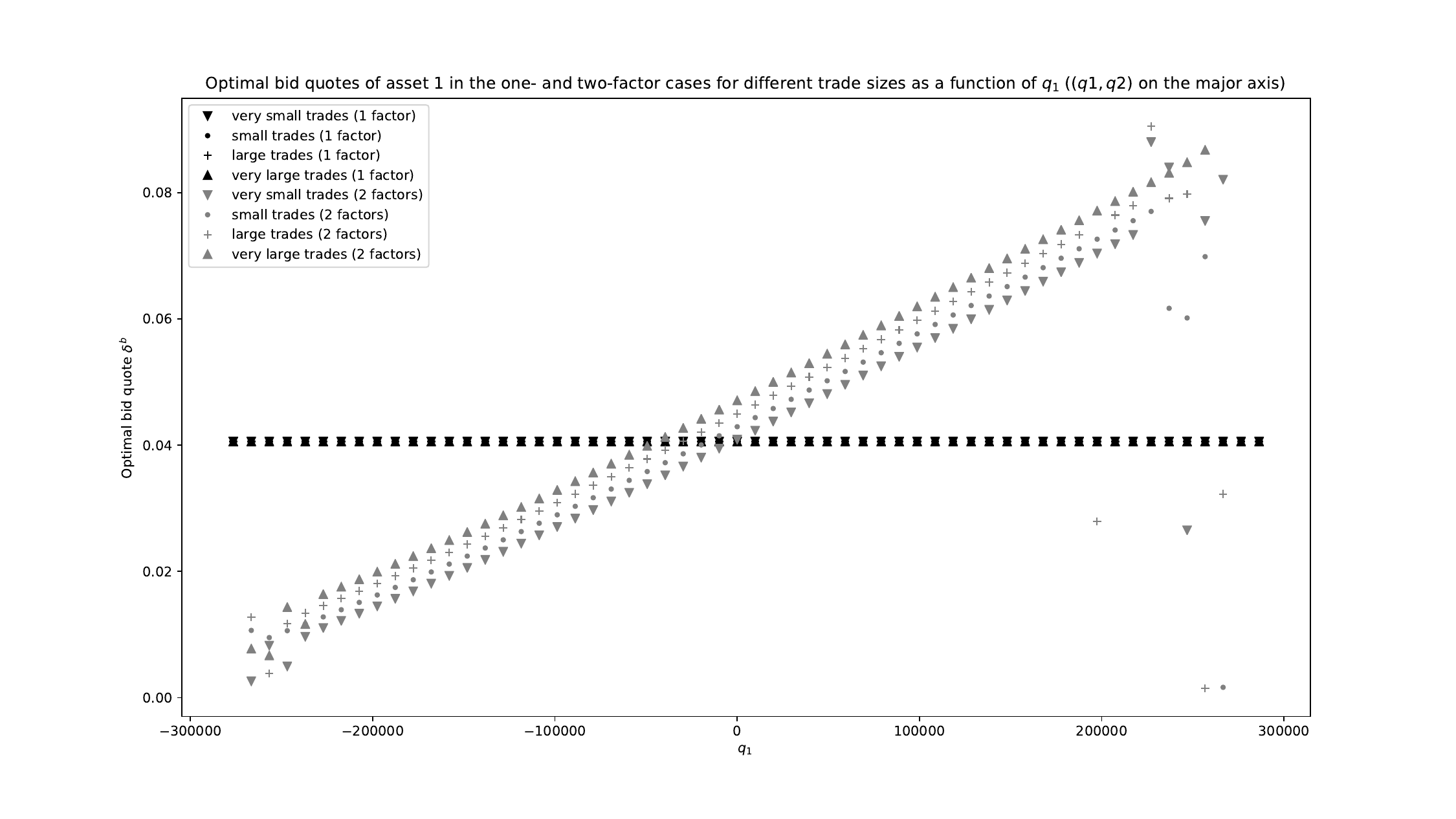}\\
\caption{Optimal bid quote for asset 1 for different trade sizes as a function of $q_1$ ($(q_1,q_2)$ on the major axis of the ellipse).}\label{diff_deltas_asset_1_major_axis_different_sizes_1f}
\end{figure}

Likewise, we plot in Figures \ref{diff_deltas_asset_2_q_2_different_sizes_1f}, \ref{diff_deltas_asset_2_q_1_different_sizes_1f}, and \ref{diff_deltas_asset_2_major_axis_different_sizes_1f} the optimal bid quotes of asset 2 when $q_1=0$ for different values of $q_2$, when $q_2=0$ for different values of $q_1$, and when $(q_1, q_2)$  spans the major axis of the ellipse of authorized inventory.\\

\begin{figure}[!h]\centering
\includegraphics[width=0.97\textwidth]{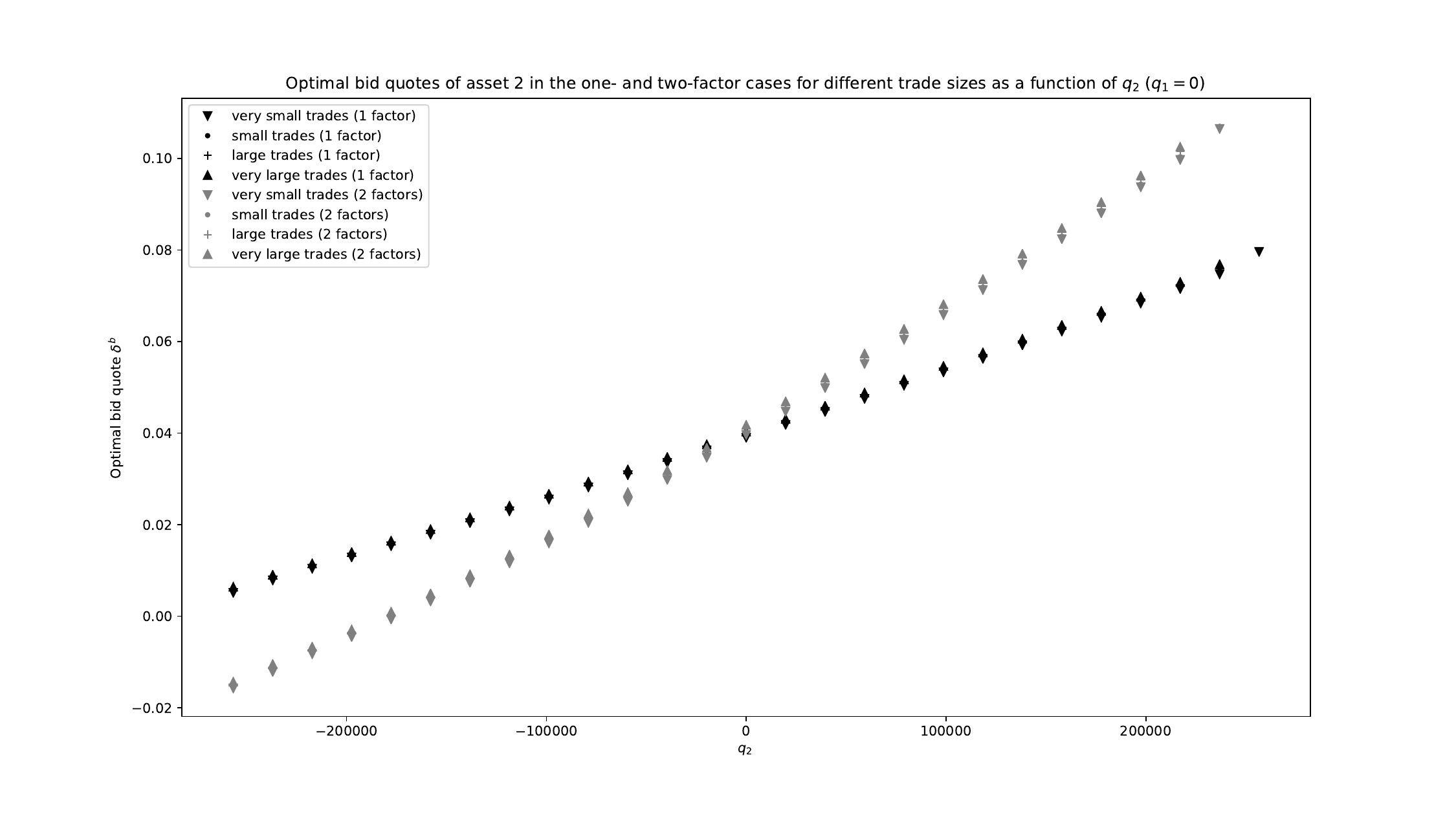}\\
\caption{Optimal bid quote for asset 2 for different trade sizes as a function of $q_2$ ($q_1 = 0$).}\label{diff_deltas_asset_2_q_2_different_sizes_1f}
\end{figure}

\begin{figure}[!h]\centering
\includegraphics[width=0.97\textwidth]{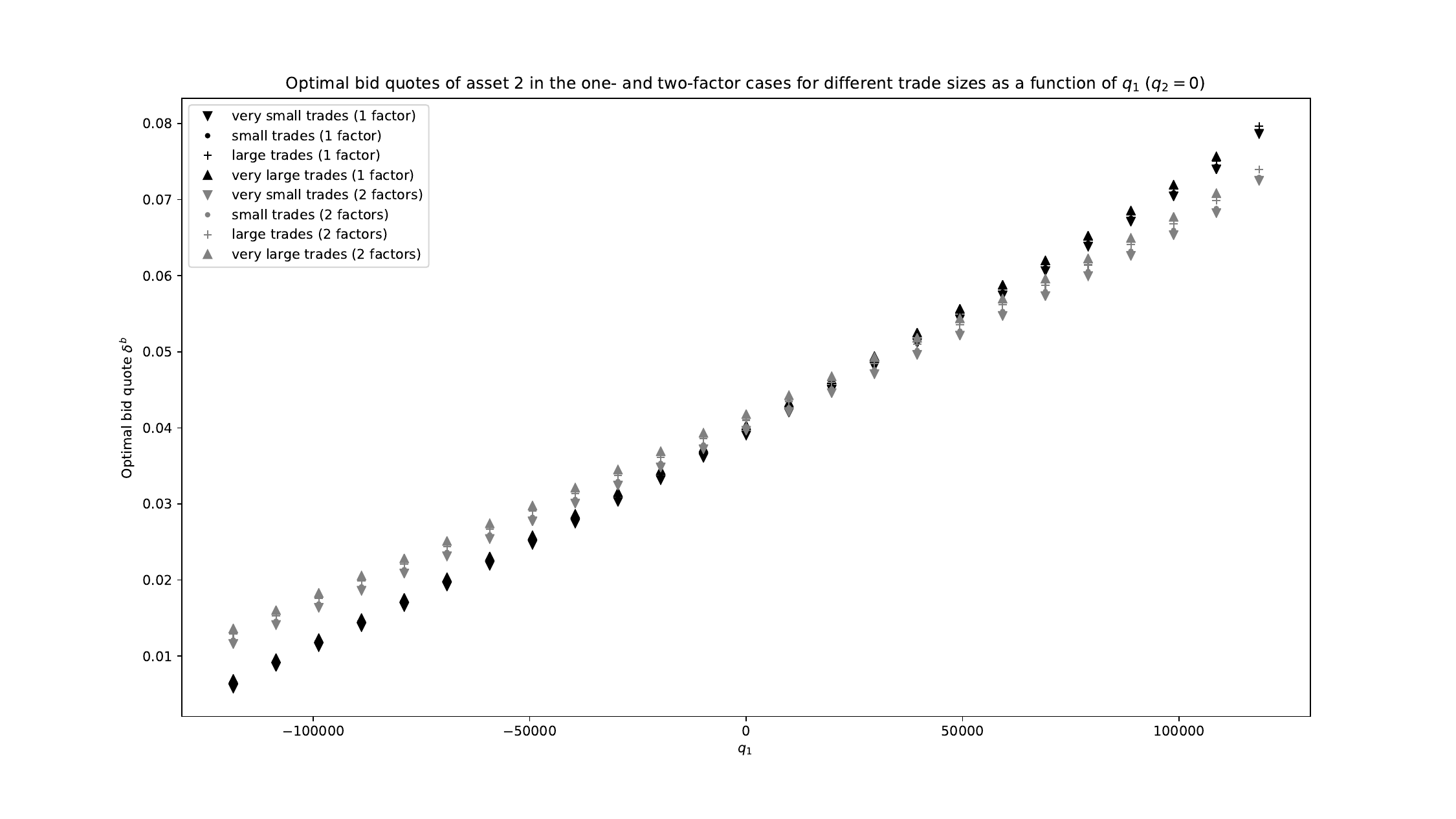}\\
\caption{Optimal bid quote for asset 2 for different trade sizes as a function of $q_1$ ($q_2 = 0$).}\label{diff_deltas_asset_2_q_1_different_sizes_1f}
\end{figure}

\begin{figure}[!h]\centering
\includegraphics[width=\textwidth]{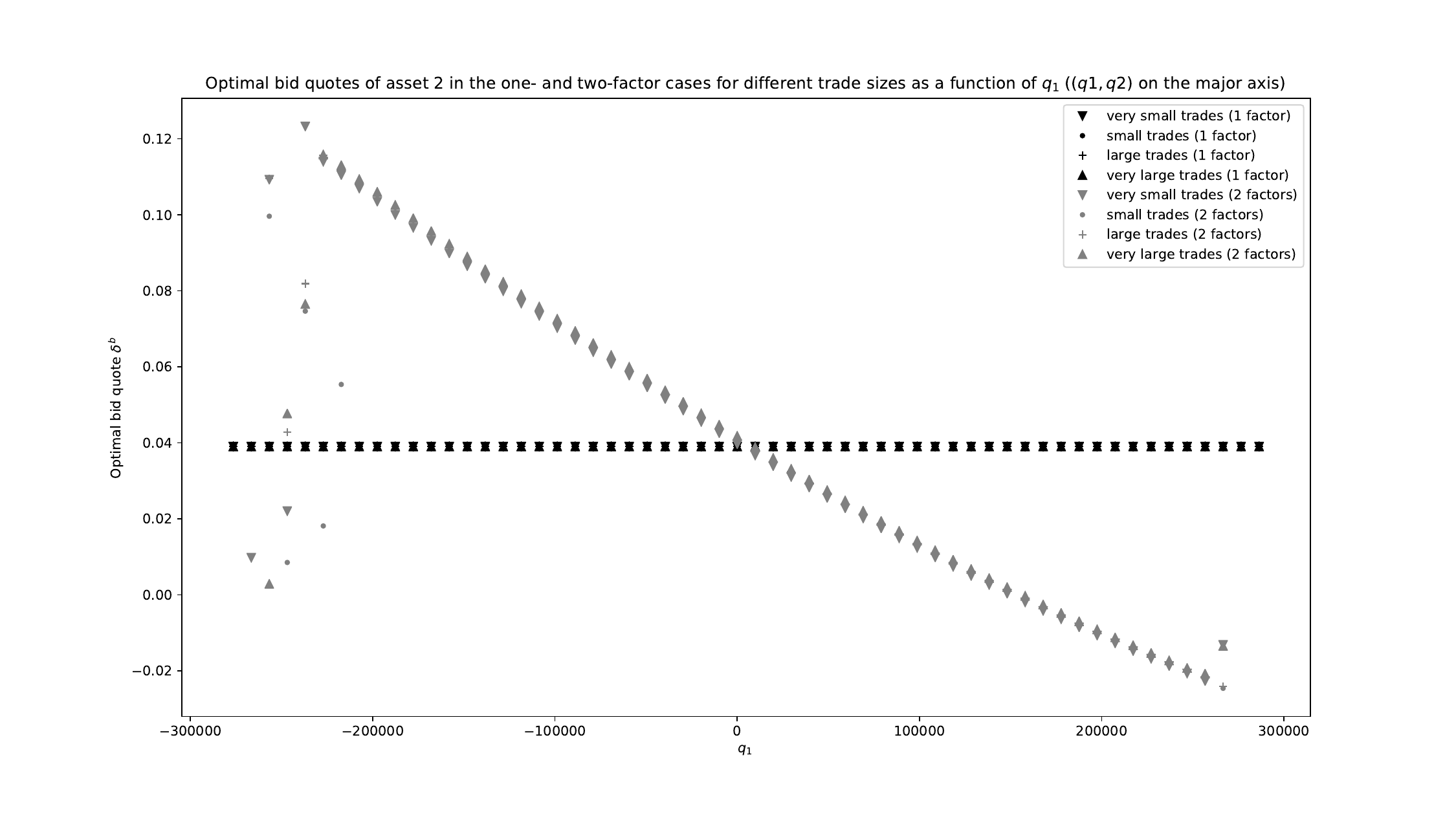}\\
\caption{Optimal bid quote for asset 2 for different trade sizes as a function of $q_1$ ($(q_1,q_2)$ on the major axis of the ellipse).}\label{diff_deltas_asset_2_major_axis_different_sizes_1f}
\end{figure}

These plots confirm that the optimal quotes in the one-asset model are good approximations of the true optimal ones whenever the inventory in each asset is not too large, all the more for inventories that are not close to the major axis of the ellipse of authorized inventory. Moreover in our example, the one-factor model seems to return quotes closer to the true optimal ones for asset 1 than for asset 2: this is due to the fact that here, the factor obtained through PCA explains better the risk of asset 1 than that of asset 2 (the residual variance of the latter is four times that of the former). \\

Comparing quotes is important but what really matters is to compare the distribution of the PnL at time $T$ when using the quotes obtained within the one-factor model with the distribution of the PnL at time $T$ when using the optimal quotes (of the two-factor model). For that purpose, we carried out a Monte-Carlo with 2000 simulations using the same source of randomness as in Section 4.1.2. The distribution of inventory when using the optimal quotes in the one-factor model is plotted in Figure \ref{q_opt_1f_dist}. The statistics associated with our simulations are documented in Table 3.\\

\begin{table}[h!]
\centering
\begin{tabular}{|c|c|c|c|}
  \hline
  Mean PnL & Stdev PnL & Stdev coming from RFQs & Objective function \\
  \hline
  72523 & 96746 & 6033 & 68567\\
  \hline
\end{tabular}
\caption{Statistics associated with our 2000 simulations starting from zero inventory (with the optimal quotes of the one-factor case).}
\end{table}

We clearly see that the performance of the one-factor approximation is very good. The value of the objective function is indeed around $68567$ when using the quotes obtained in the one-factor model, a bit smaller than the value of approximately $69293$ obtained with the same source of randomness when using the optimal quotes. In fact, the average PnL is higher with the one-factor optimal quotes but, since the distribution of inventory is denser in areas that are falsely believed to be risk-free or wrongfully associated with low risk, the standard deviation of the PnL is in fact higher (around $96746$ versus around $80432$), resulting in a lower value of the objective function.\\

\begin{figure}[!h]\centering
\includegraphics[width=0.99\textwidth]{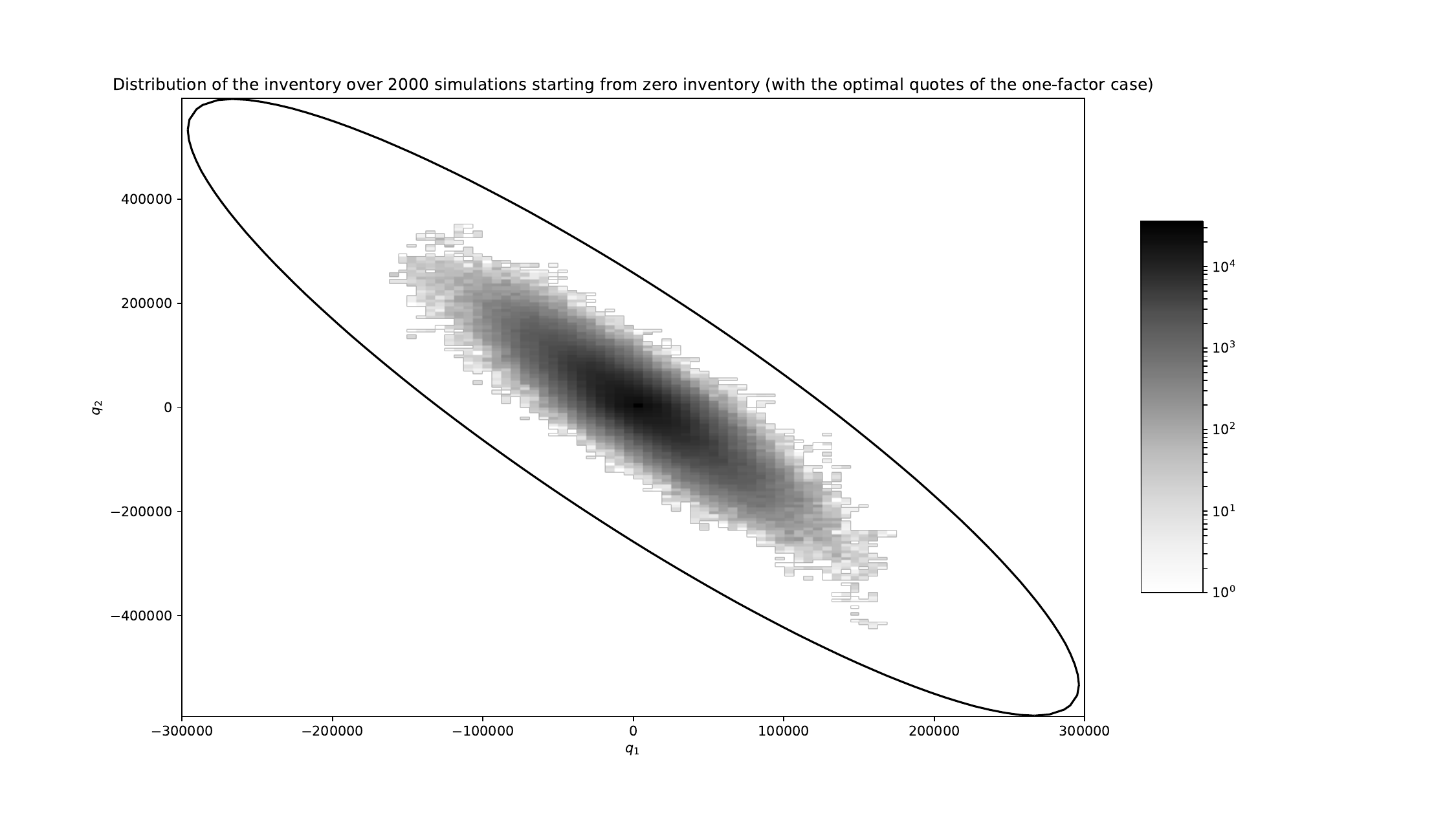}\\
\caption{Distribution of the inventory over 2000 simulations starting from zero inventory (with the optimal quotes of the one-factor case).}\label{q_opt_1f_dist}
\end{figure}
\iffalse
\begin{figure}[!h]\centering
\includegraphics[width=0.98\textwidth]{pnl_distrib_2d1f.pdf}\\
\caption{Histogram of the PnL over 2000 simulations starting from zero inventory (with the optimal quotes of the one-factor case).}\label{pnl_distrib_2d1f}
\end{figure}
\fi
\subsubsection{Taking the residual risk into account with our Monte-Carlo method}
\label{MCnum}
We have seen above that the use of the optimal quotes of the one-factor model provides very good results in terms of the value of the objective function. Nevertheless, the distribution of the inventory plotted in Figure \ref{q_opt_1f_dist} differs from the distribution associated with the true optimal quotes plotted in Figure \ref{q_opt_dist} because the major axis of the ellipse of authorized inventory is associated with zero risk in the one-factor model. In this section, we illustrate the Monte-Carlo method proposed in Section \ref{MCr} in order to account for the residual risk in the approximation of the value function and the optimal quotes.\\

It is noteworthy that the Monte-Carlo method of Section \ref{MCr}, unlike the grid method of Section \ref{grid}, does not allow to compute the optimal quotes for all assets, sides, sizes, and values of the inventory (we ignore time by focusing on $t=0$) at once. Instead, it requires a different Monte-Carlo simulation for each desired quote. In particular, should it be used by practitioners, its use should be online. In other words, the computations should only be carried out upon receiving an RFQ or slightly beforehand if one wants to prepare the quotes (given the current inventory) for the most probable RFQs.\\

A related point is that, even for illustration and even with two assets, it is too time-consuming to compute the Monte-Carlo adjustment for all assets, sides, sizes, and possible inventories. As a consequence, it is too time-consuming to carry out simulations of the PnL with the quotes amended by the Monte-Carlo method of Section \ref{MCr}. Instead of a complete analysis, we focus on a sectional analysis by looking at the cases $q^2 = 0$, $q^1=0$, and $(q^1, q^2)$ on the major axis of the ellipse of authorized inventory.\\

In Figures \ref{value_funcs_q_1_q_2_0}, \ref{value_funcs_q_2_q_1_0}, and \ref{value_funcs_major_axis}, we compare the value function obtained in the two-factor case, i.e. the true value function,  to the value function of the one-factor case and to its adjustment through the Monte-Carlo technique of Section \ref{MCr} -- we use 50 simulations for each point (with the same source of randomness for all points). Figure \ref{value_funcs_q_1_q_2_0} deals with the comparison of the values on the section $\{q^2 = 0\}$, Figure \ref{value_funcs_q_2_q_1_0} deals with the comparison of the values on the section $\{q^1 = 0\}$, and Figure \ref{value_funcs_major_axis} deals with the comparison of the values on the major axis of the ellipse of authorized inventory.\\

\begin{figure}[!h]\centering
\includegraphics[width=0.95\textwidth]{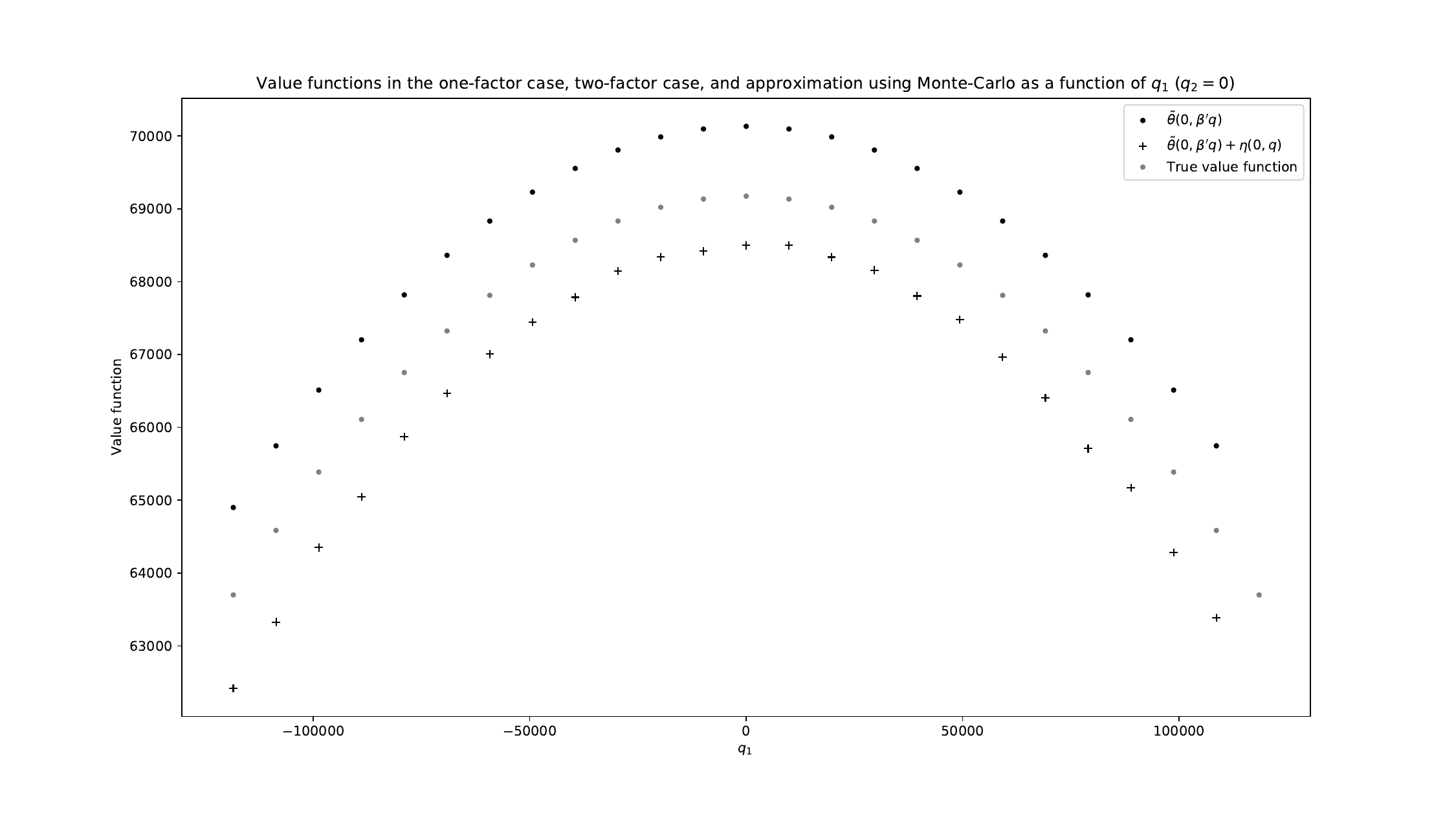}\\
\caption{Value functions in the one-factor case, two-factor case, and approximation using Monte-Carlo as a function of $q_1$ ($q_2 = 0$).}\label{value_funcs_q_1_q_2_0}
\end{figure}

\begin{figure}[!h]\centering
\includegraphics[width=0.95\textwidth]{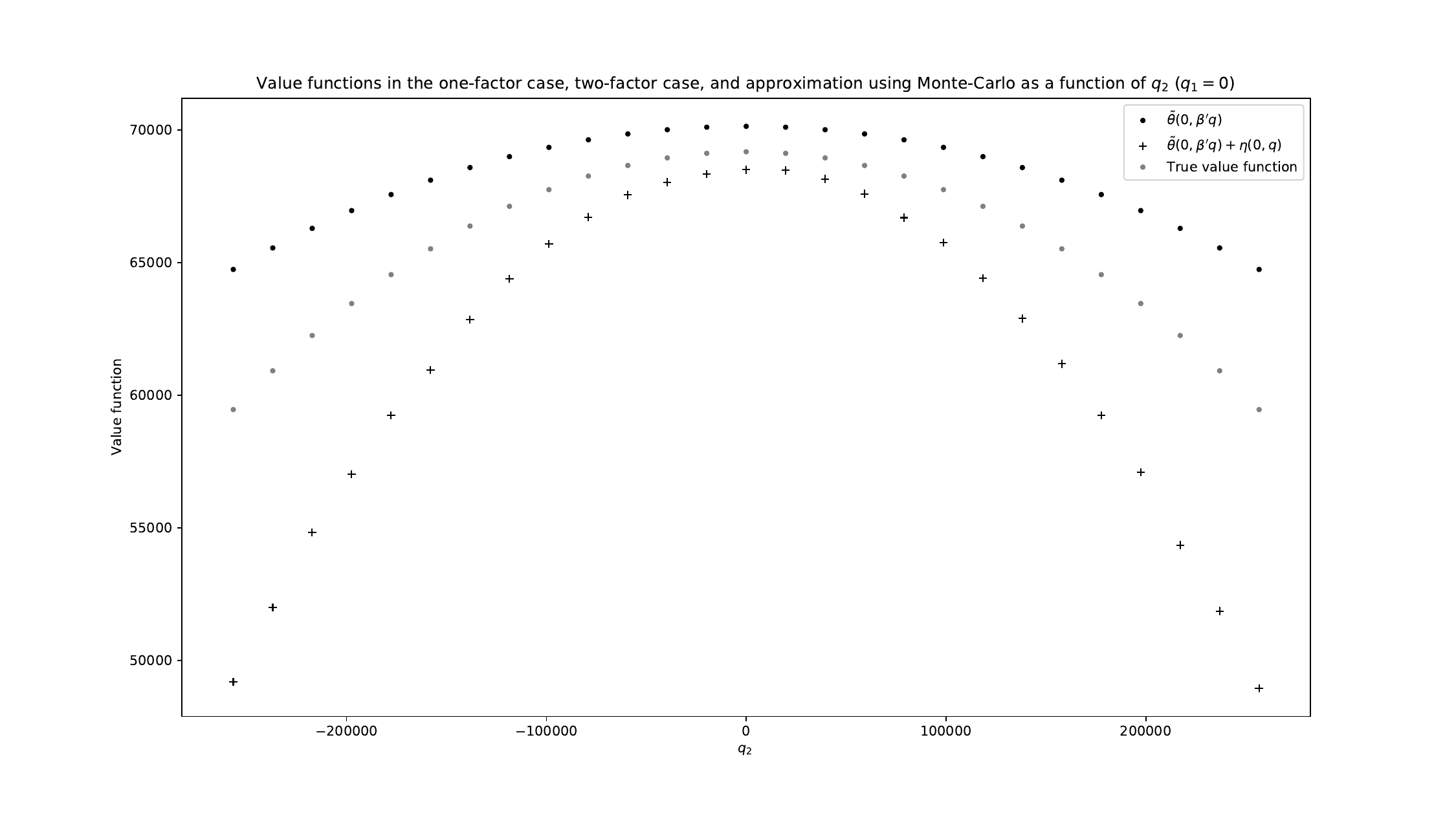}\\
\caption{Value functions in the one-factor case, two-factor case, and approximation using Monte-Carlo as a function of $q_2$ ($q_1 = 0$).}\label{value_funcs_q_2_q_1_0}
\end{figure}

\begin{figure}[!h]\centering
\includegraphics[width=0.95\textwidth]{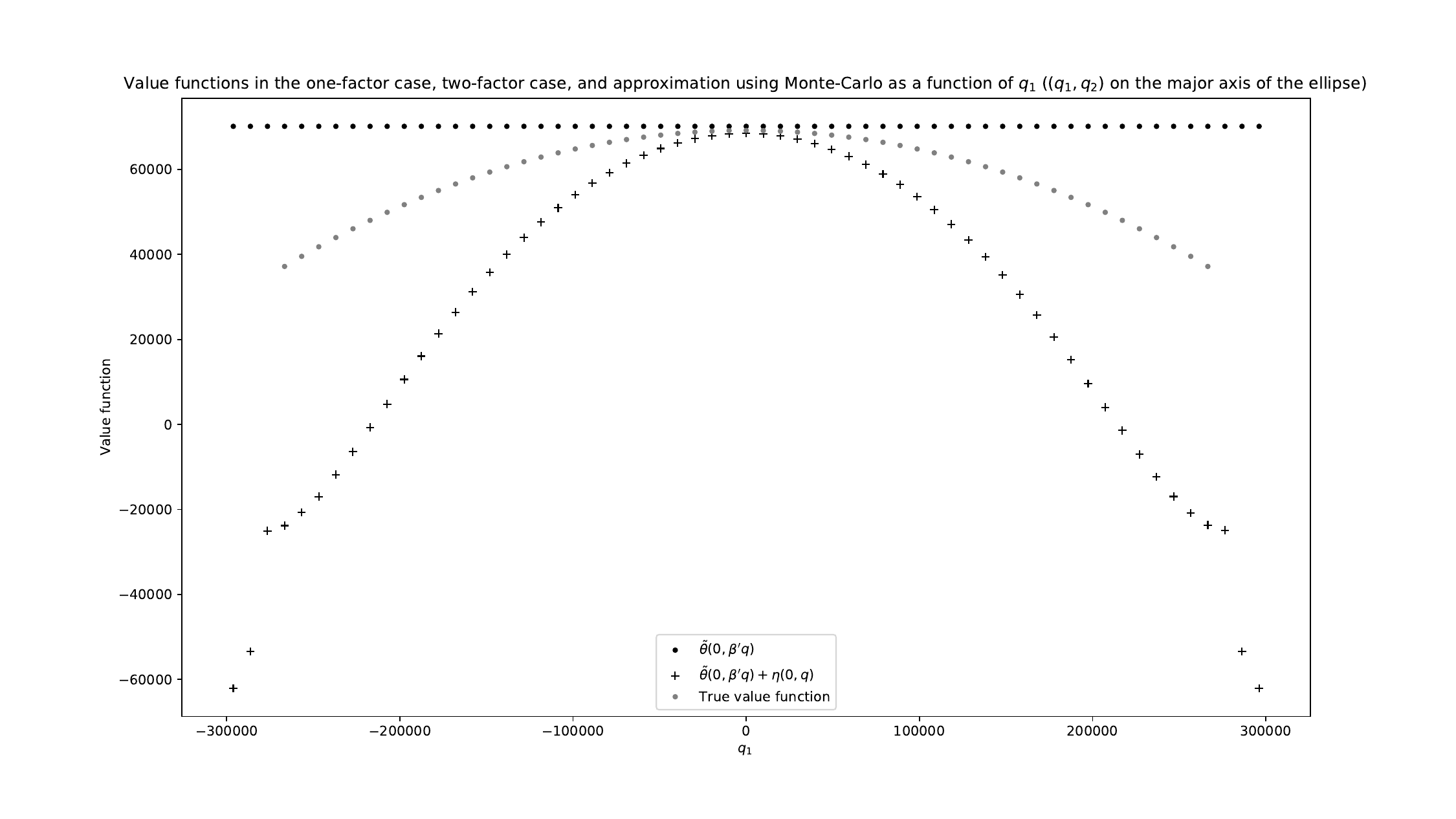}\\
\caption{Value functions in the one-factor case, two-factor case, and approximation using Monte-Carlo as a function of $q_1$ ($(q_1,q_2)$ on the major axis of the ellipse.}\label{value_funcs_major_axis}
\end{figure}

\begin{figure}[!h]\centering
\includegraphics[width=0.95\textwidth]{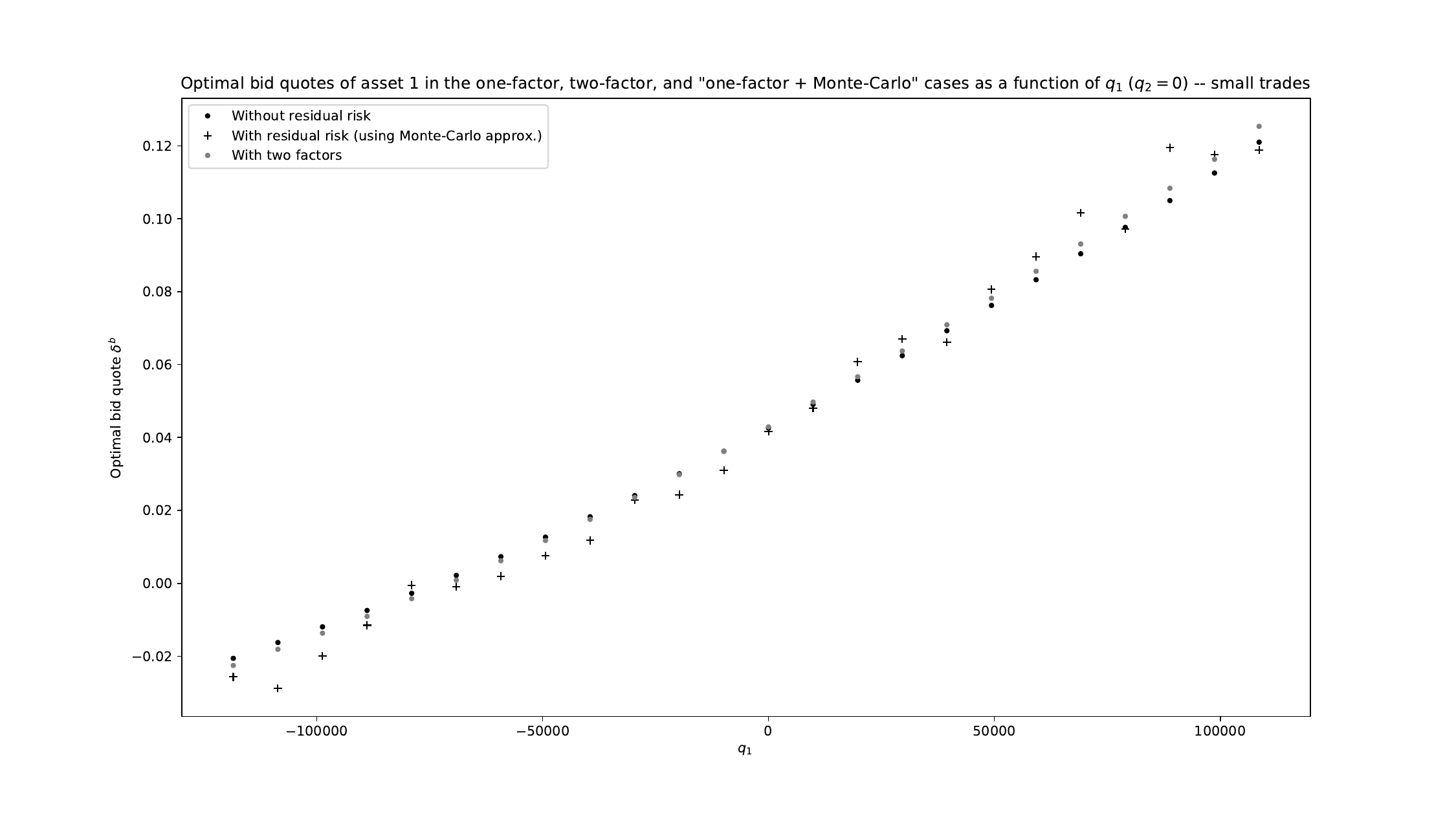}\\
\caption{Optimal bid quotes of asset 1 in the one-factor, two-factor, and "one-factor + Monte-Carlo" cases as a function of $q_1$ ($q_2 = 0$) -- small trades.}\label{Bid_1_with_residual_risk_2d}
\end{figure}

\begin{figure}[!h]\centering
\includegraphics[width=0.95\textwidth]{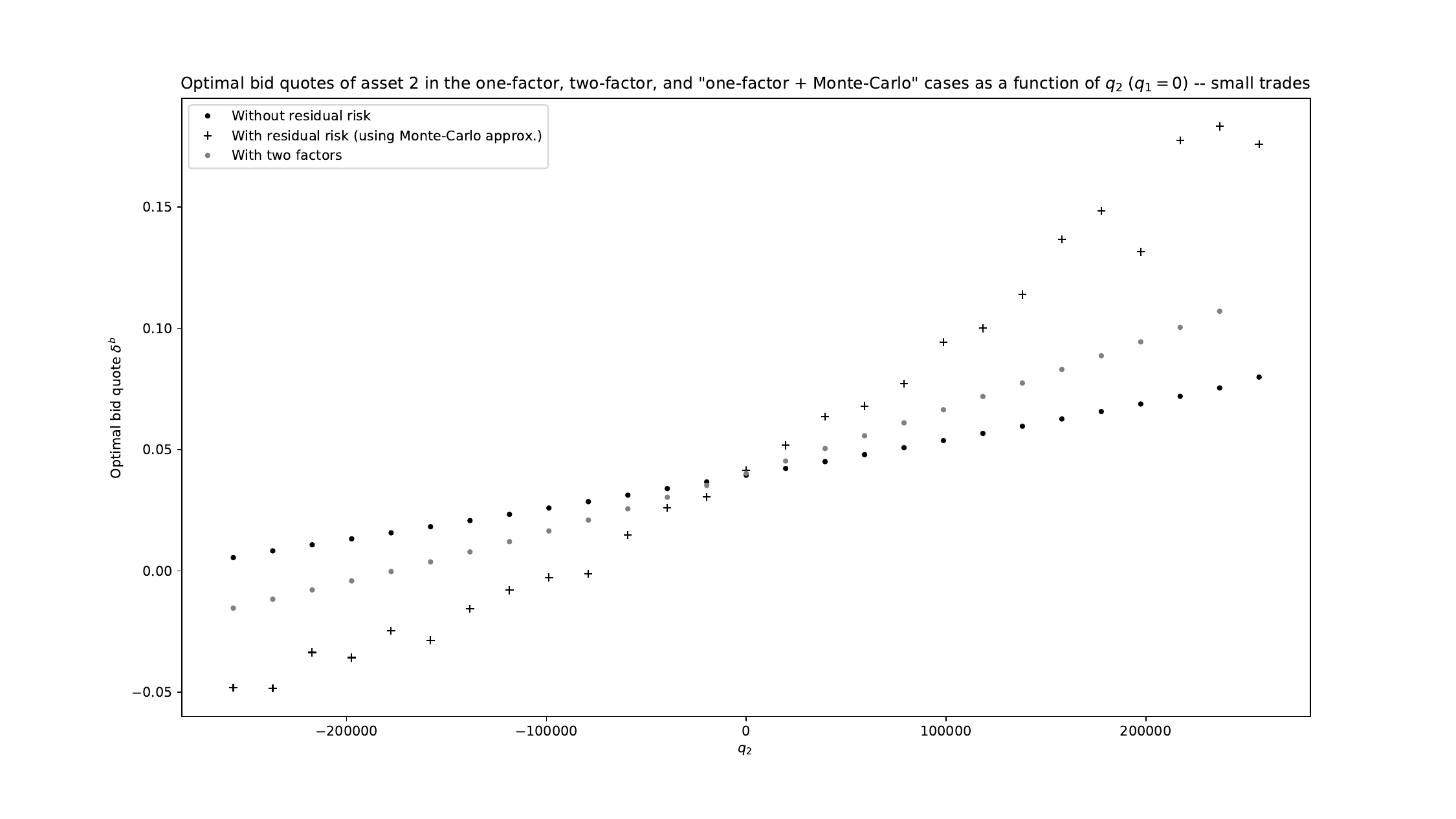}\\
\caption{Optimal bid quotes of asset 2 in the one-factor, two-factor, and "one-factor + Monte-Carlo" cases as a function of $q_2$ ($q_1 = 0$) -- small trades.}\label{Bid_2_with_residual_risk_2d}
\end{figure}

\begin{figure}[!h]\centering
\includegraphics[width=0.95\textwidth]{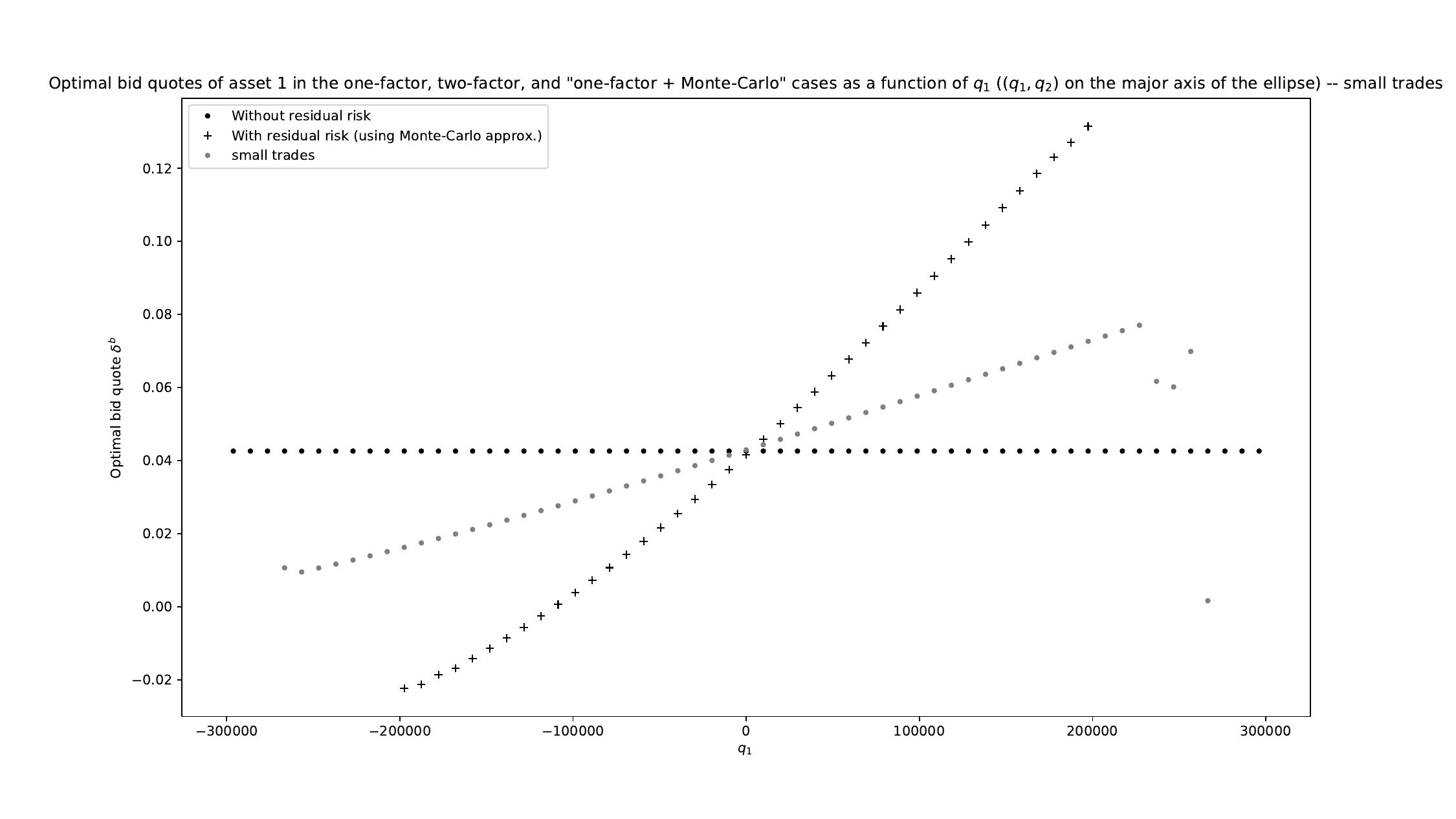}\\
\caption{Optimal bid quotes of asset 1 in the one-factor, two-factor, and "one-factor + Monte-Carlo" cases as a function of $q_1$ ($(q_1,q_2)$ on the major axis of the ellipse) -- small trades.}\label{Bid_1_with_residual_risk_major_axis_2d}
\end{figure}

We clearly see that, unsurprisingly, the Monte-Carlo adjustment goes in the right direction. However, the Monte-Carlo method leads to (i) an overestimation of the gap between the value function of the one-factor case and the true value function and (ii) an overestimation of the degree of concavity of the value function (this is particularly the case for the third section).\\

It is noteworthy that the quality of the approximation is the best around $0$. This point is interesting if one wants to estimate the degree of suboptimality of a quoting strategy in a scenario starting with zero inventory. The poor approximation of the concavity is however a limitation since quotes are based on finite differences of the value function. This is well illustrated by Figures \ref{Bid_1_with_residual_risk_2d}, \ref{Bid_2_with_residual_risk_2d}, and \ref{Bid_1_with_residual_risk_major_axis_2d}. Nevertheless, even though the Monte-Carlo adjustments of quotes are too large, especially on the major axis of the ellipse (see Figure \ref{Bid_1_with_residual_risk_major_axis_2d}), this drawback of the Monte-Carlo method of Section \ref{MCr} should be qualified as the quotes obtained with the Monte-Carlo technique naturally lead -- because of the overestimated slope of the quotes -- to trajectories of the inventory more concentrated around $0$ and therefore to a very rare use of the quotes that are too different from the optimal ones.\\

Before we go on with an example including 30 assets, let us conclude on the two-asset case. The method we propose to tackle the curse of dimensionality is based on the projection of market risk on a low-dimensional space of factors. It works very well in the two-asset case as the quotes permit to reach a value of the objective function close to the optimal one. The Monte-Carlo adjustment we suggested in Section \ref{MCr} allows to approximate the true value function at the point of zero inventory, which is quite useful when one does not have access to the true value function, as is the case in high dimension. However, it overestimates the changes one must make to the quotes computed with the low-dimensional approximation.\\

\subsection{Dealing with 30 assets}

We now consider the more challenging case of a market maker in charge of 30 assets (here bonds) with the following characteristics:\\
\begin{itemize}
  \item Asset prices: $S^i_0 = 100\ \textrm{\euro}, \quad \forall i \in \{1, \ldots, 30\}$.
  \item Volatility of assets: $\sigma^i = 1.2\  \textrm{\euro} \cdot \textrm{day}^{-\frac{1}{2}},  \forall i \in \{1, \ldots, 15\}$, and $\sigma^i = 0.6\  \textrm{\euro} \cdot \textrm{day}^{-\frac{1}{2}},  \forall i \in \{16, \ldots, 30\}$.
    \item Correlation matrix: $\begin{bmatrix}
R_{11} & R_{12}\\
R_{21} & R_{22} \\
\end{bmatrix},$ where
    \begin{equation*}
R_{11} = R_{22} = \begin{bmatrix}
1.0 & 0.9 & \hdots & \hdots & 0.9 \\
0.9 & \ddots & \ddots & \ddots & 0.9 \\
\vdots & \ddots & \ddots & \ddots & \vdots \\
0.9 & \ddots & \ddots & \ddots & 0.9 \\
0.9 & \hdots & \hdots & 0.9 & 1.0 \\
\end{bmatrix}
\quad \textrm{and} \quad
R_{12} = R_{21} = \begin{bmatrix}
0.2 & 0.2 & \hdots & 0.2 & 0.2 \\
0.2 & \ddots & \ddots & \ddots & 0.2 \\
\vdots & \ddots & \ddots & \ddots & \vdots \\
0.2 & \ddots & \ddots & \ddots & 0.2 \\
0.2 & 0.2 & \hdots & 0.2 & 0.2 \\
\end{bmatrix}.
\end{equation*}
  \item Intensity functions: $$\Lambda^{i,b}(\delta) = \Lambda^{i,a}(\delta) = \lambda_{RFQ} \frac{1}{1+e^{\alpha_\Lambda + \beta_\Lambda\delta}},\quad  \forall i \in \{1, \ldots, 30\},$$ with $\lambda_{RFQ} = 10\ \textrm{day}^{-1}$, $\alpha_\Lambda=0.7$, and $\beta_\Lambda=30\ \textrm{\euro}^{-1}$. This corresponds to 10 RFQs per day for each asset, a probability of $\frac 1 {1+e^{0.7}} \simeq 33\%$ to trade when the answered quote is the reference price and a probability of $\frac 1 {1+e^{-0.2}} \simeq 55\%$ to trade when the answered quote is the reference price improved by 3 cents.
  \item Request sizes are distributed according to a Gamma distribution $\Gamma(\alpha,\beta)$ with $\alpha = 4$ and $\beta=4\cdot 10^{-4}$. This corresponds to an average request size of $10000$ assets (i.e. approximately $1000000\ \textrm{\euro}$) and a standard deviation equal to half the average.\\
\end{itemize}

The variance-covariance matrix $\Sigma$ has two eigenvalues equal to $19.895060$ and  $4.584941$, and 28 eigenvalues below~$0.15$. The first eigenspace is spanned by a vector with all coordinates of the same sign. The associated factor -- the first factor -- represents an index of the 30 assets. The second eigenspace is spanned by a vector with the first 15 coordinates of the same sign and the next 15 of the same, but opposite, sign. The associated factor -- the second factor -- allows to separate the two groups of assets.\\

We can therefore legitimately approximate our 30-asset problem by a two-factor one and solve the corresponding Hamilton-Jacobi equation \eqref{eqn:HJBfac} to approximate the optimal quotes.\\

Regarding the objective function, we consider the following:\\
\begin{itemize}
  \item Time horizon given by $T = 2\ \textrm{days}$. This horizon ensures convergence towards stationary quotes at time $t=0$ -- see Figure \ref{conv_deltas_30_d}.
  \item $\psi: q \in \mathbb{R}^2 \mapsto \frac \gamma 2 q'\Sigma q$ with $\gamma = 8\cdot10^{-7}\ \textrm{\euro}^{-1}$.
  \item $\ell_d = 0$.\\
\end{itemize}

We approximate the solution $\tilde{\theta}$ to \eqref{eqn:HJBfac} with two factors using a monotone explicit Euler scheme with linear interpolation on a grid of size $71\times71$ for the factors and a discretization of the RFQ size distribution with the same 4 sizes as in the above two-asset example.\footnote{We considered risk limits similar to those of the above two-asset example. Here, no trade that would result in an inventory $q \in \mathbb{R}^{30}$ such that $q'\Sigma q > B$ was admitted, where $B = 5\cdot10^{10}$.}\\

The value function (at time $t=0$) is plotted in Figure \ref{theta_f_3d_30_d}. The associated optimal bid quotes for asset 1 and asset~16 (for the smallest RFQ size) are plotted in Figures \ref{delta_b_3d_asset_1_size_0_30_d} and \ref{delta_b_3d_asset_15_size_0_30_d}. We see on these graphs that the optimal quotes depend monotonously on the two factors.\footnote{Exceptions to this monotonicity property are related to boundary effects.}\\

As discussed above, we chose $T=2$ days to ensure convergence of the quotes to their stationary values. This is illustrated in Figure \ref{conv_deltas_30_d}.\\

\begin{figure}[!h]\centering
\includegraphics[width=\textwidth]{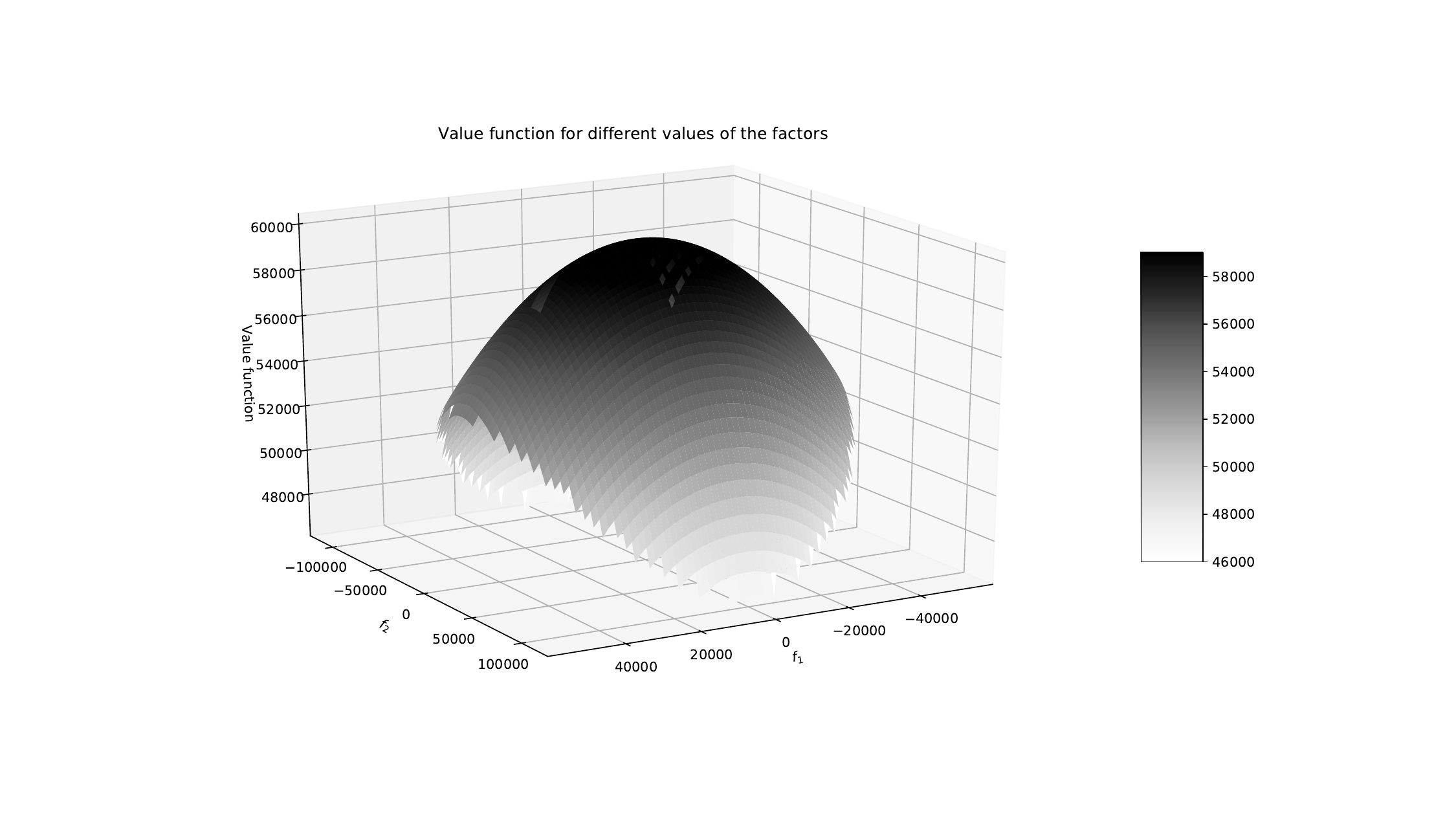}\\
\caption{Value function for different values of the factors.}\label{theta_f_3d_30_d}
\end{figure}
\vspace{1cm}

\begin{figure}[!h]\centering
\includegraphics[width=0.98\textwidth]{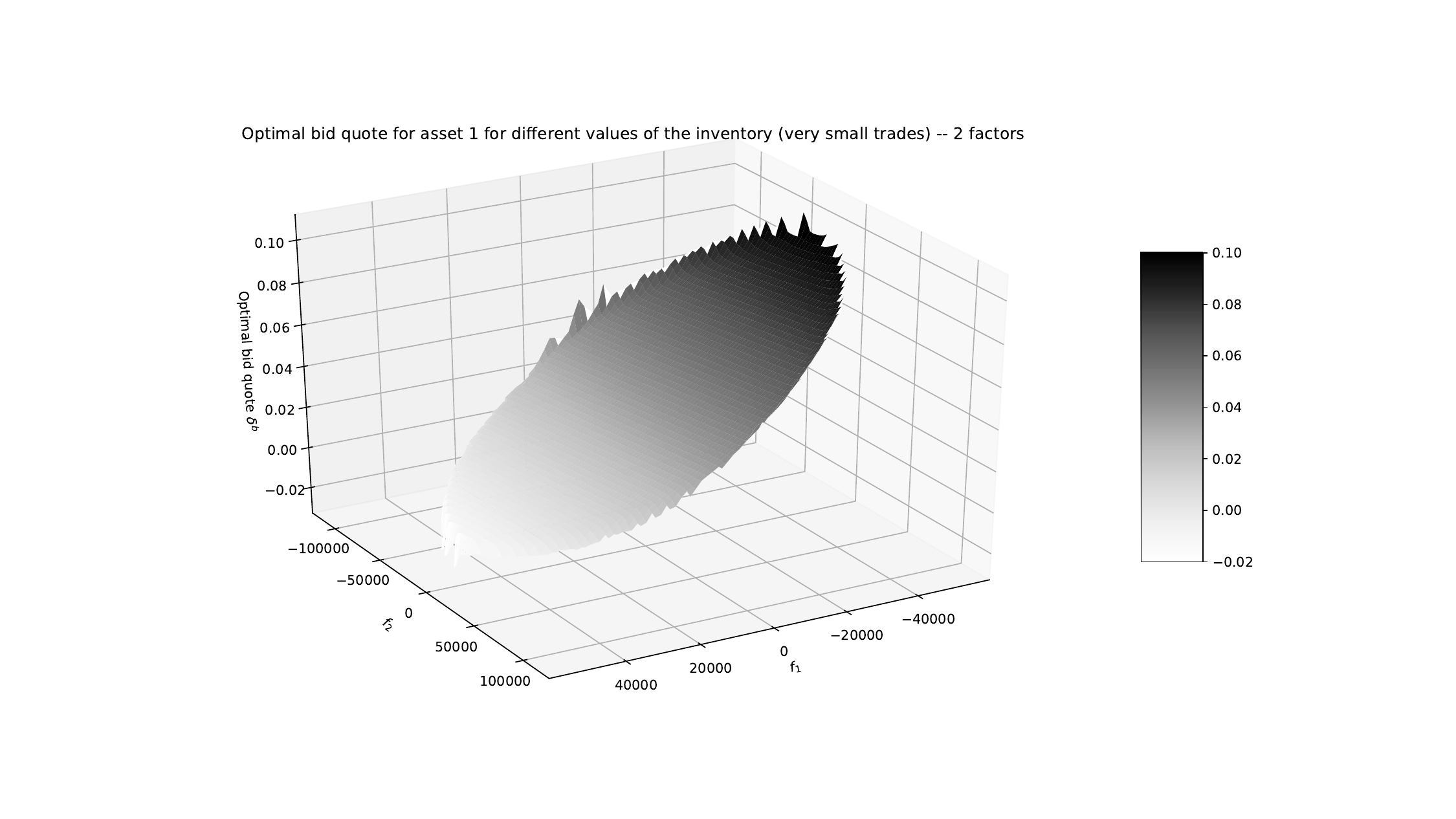}\\
\caption{Optimal bid quote for asset 1 for different values of the inventory (very small trades) -- 2 factors.}\label{delta_b_3d_asset_1_size_0_30_d}
\end{figure}
\vspace{5mm}

\begin{figure}[!h]\centering
\includegraphics[width=0.98\textwidth]{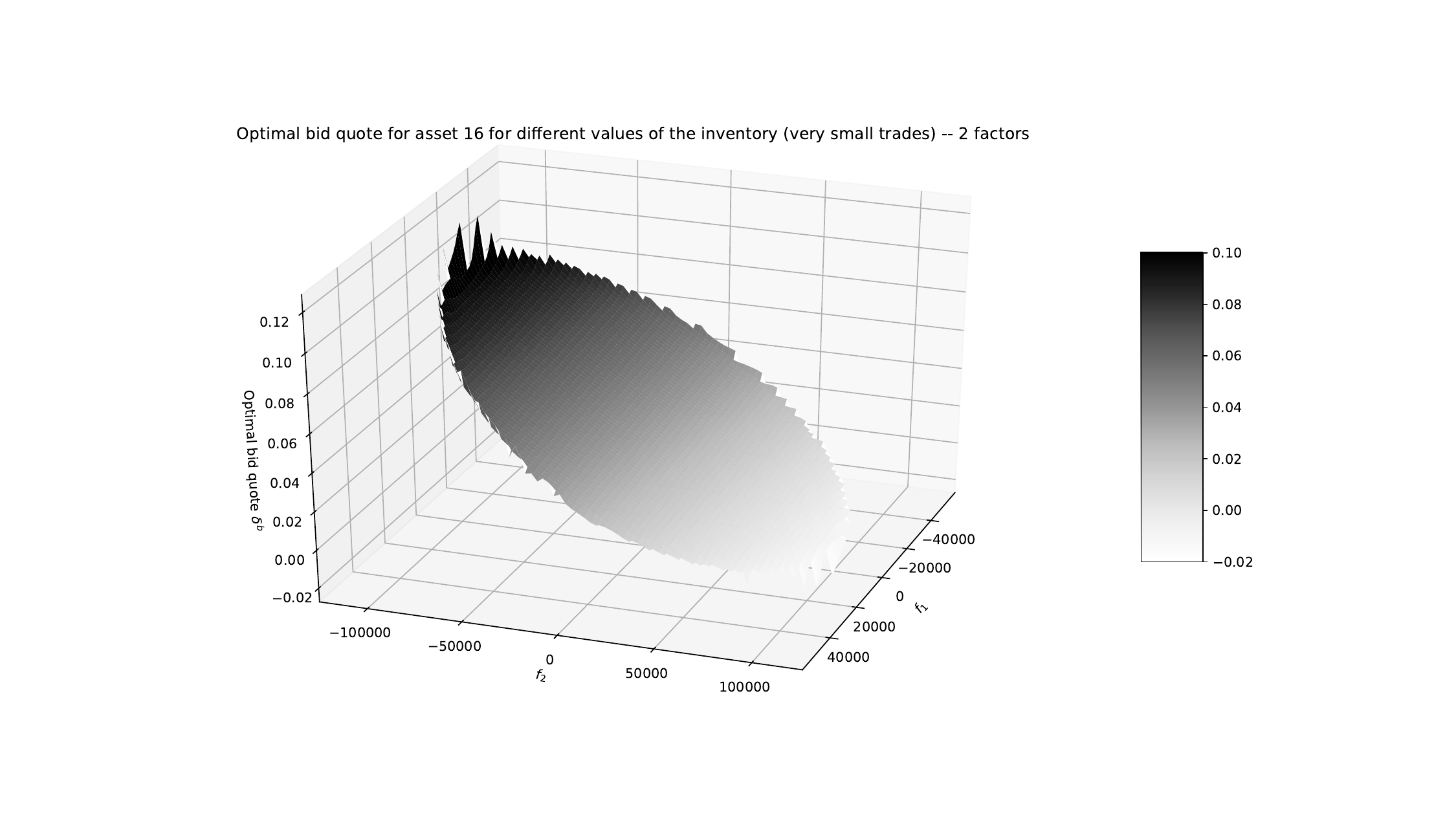}\\
\caption{Optimal bid quote for asset 16 for different values of the inventory (very small trades) -- 2 factors.}\label{delta_b_3d_asset_15_size_0_30_d}
\end{figure}

\begin{figure}[!h]\centering
\includegraphics[width=0.9\textwidth]{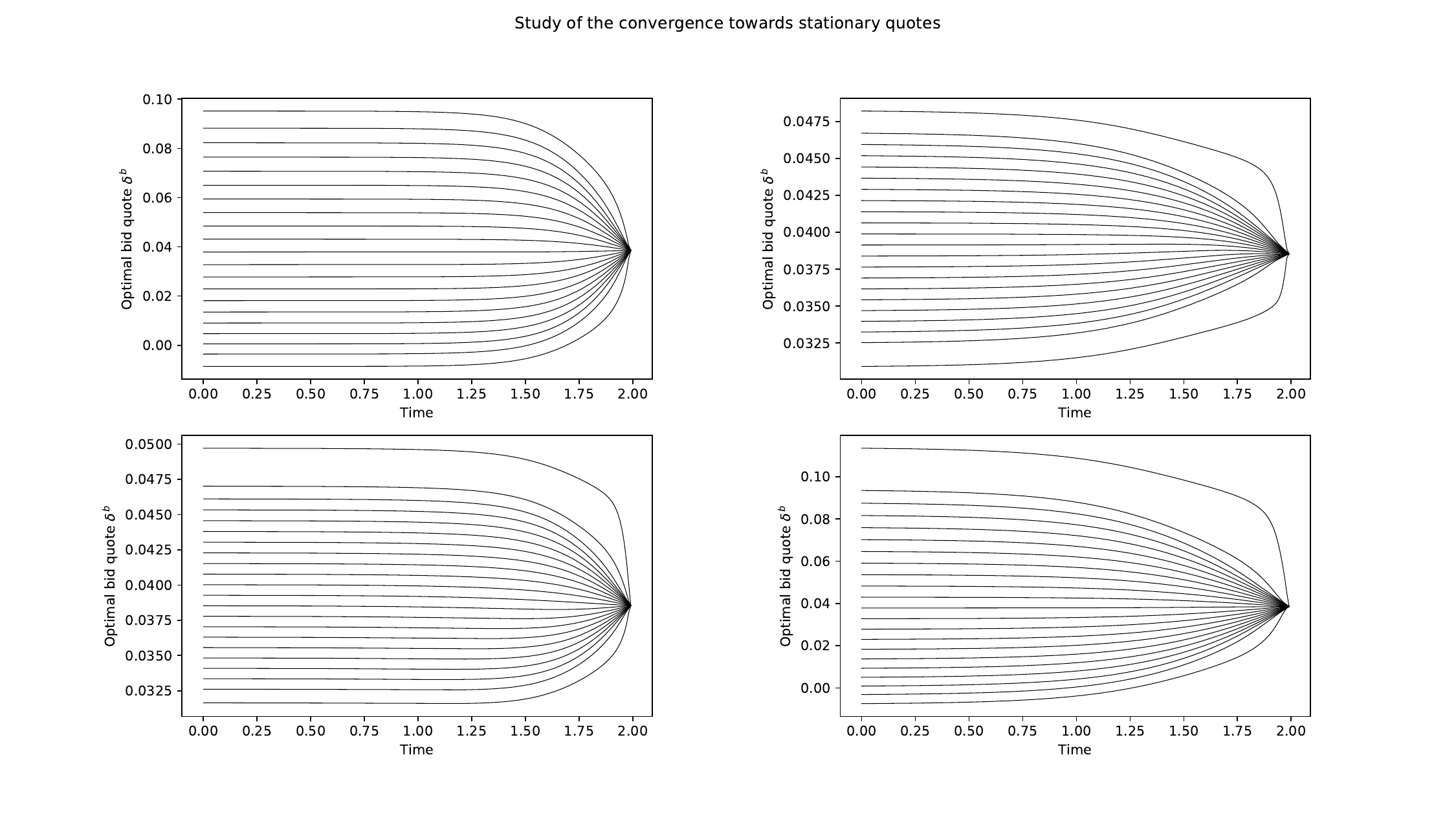}\\
\caption{Optimal bid quotes as a function of time for various values of the two factors. Top left: Asset 1 when $f^2 =0$. Top right: Asset 1 when $f^1 =0$. Bottom left: Asset 16 when $f^2 =0$. Bottom right: Asset 16 when $f^1 =0$.}\label{conv_deltas_30_d}
\end{figure}

To see the role of RFQ size on the optimal quotes, we plot in Figure \ref{deltas_asset_0_f_1_different_sizes_30_d} the four functions $f^1 \mapsto \bar{\delta}^{1,b}(0,f^1,0,z^k), k \in \{1,\ldots,4\}$ and in Figure \ref{deltas_asset_0_f_2_different_sizes_30_d} the four functions $f^2 \mapsto \bar{\delta}^{1,b}(0,0,f^2,z^k), k \in \{1,\ldots,4\}.$\\

\begin{figure}[!h]\centering
\includegraphics[width=0.92\textwidth]{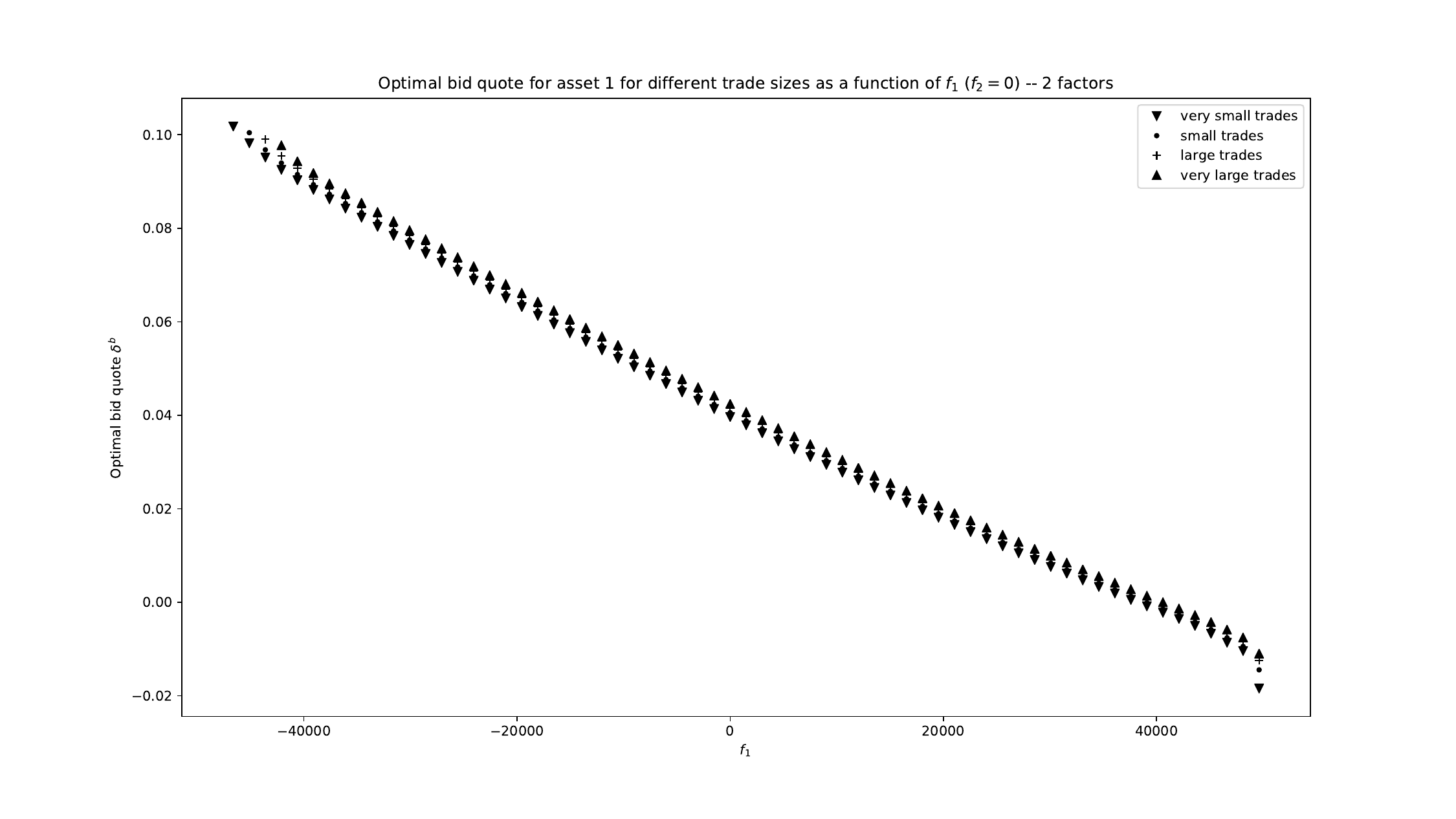}\\
\caption{Optimal bid quote for asset 1 for different trade sizes as a function of $f_1$ ($f_2 = 0$) -- 2 factors.}\label{deltas_asset_0_f_1_different_sizes_30_d}
\end{figure}

\begin{figure}[!h]\centering
\includegraphics[width=0.92\textwidth]{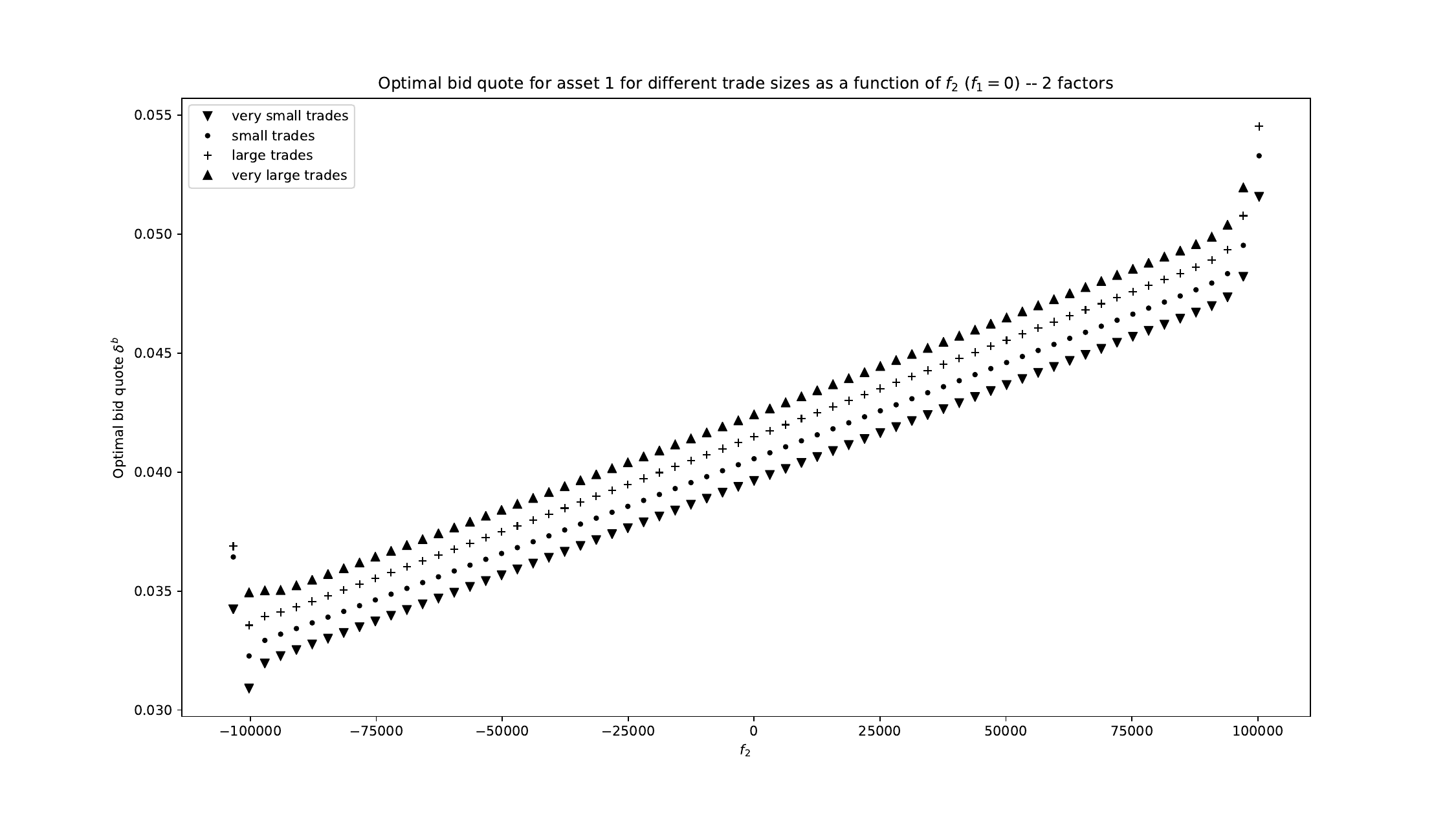}\\
\caption{Optimal bid quote for asset 1 for different trade sizes as a function of $f_2$ ($f_1 = 0$) -- 2 factors.}\label{deltas_asset_0_f_2_different_sizes_30_d}
\end{figure}

Likewise, we plot in Figure \ref{deltas_asset_15_f_1_different_sizes_30_d} the four functions $f^1 \mapsto \bar{\delta}^{16,b}(0,f^1,0,z^k), k \in \{1,\ldots,4\}$ and in Figure \ref{deltas_asset_15_f_2_different_sizes_30_d} the four functions $f^2 \mapsto \bar{\delta}^{16,b}(0,0,f^2,z^k), k \in \{1,\ldots,4\}.$\\

\begin{figure}[!h]\centering
\includegraphics[width=0.92\textwidth]{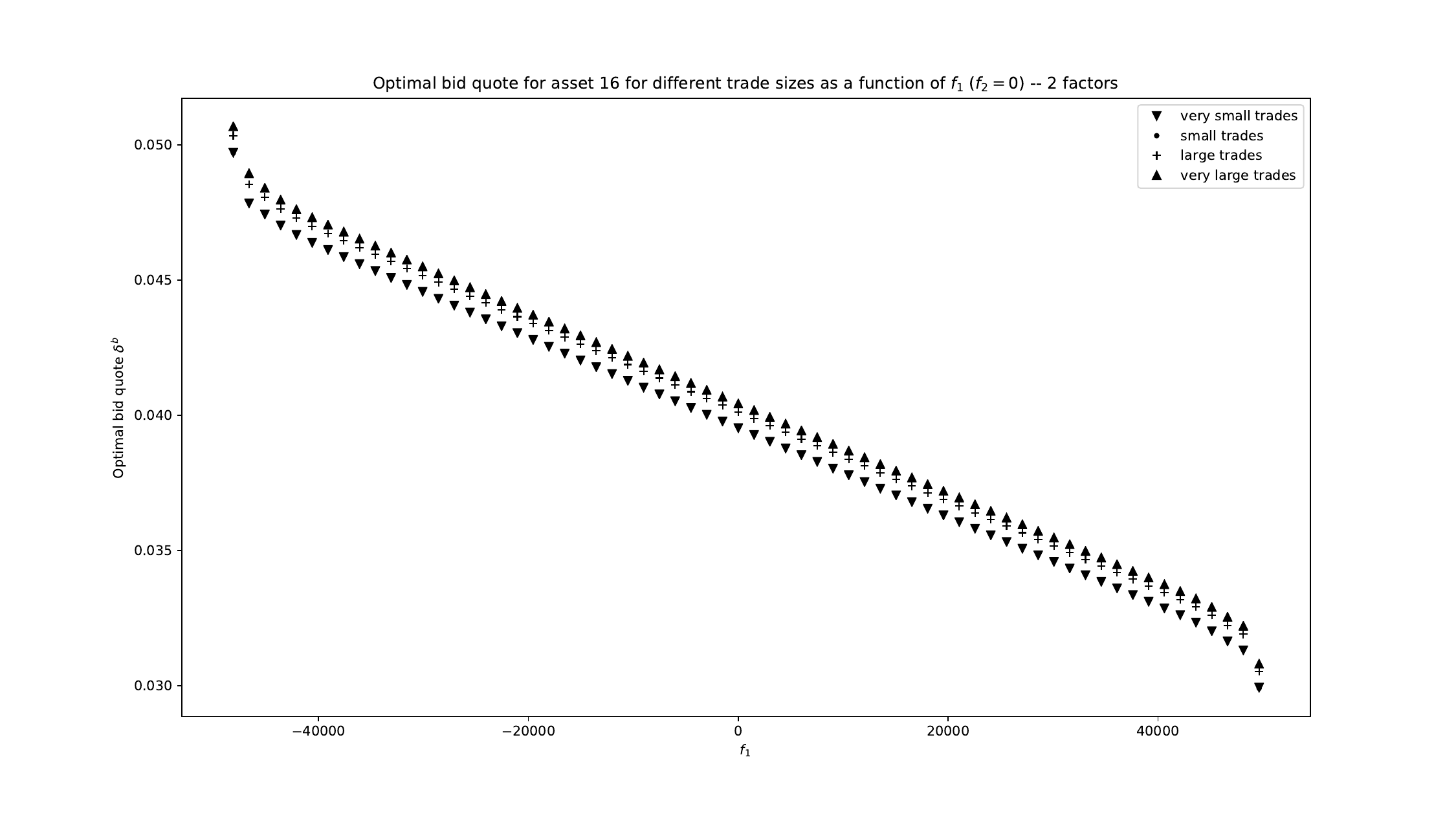}\\
\caption{Optimal bid quote for asset 16 for different trade sizes as a function of $f_1$ ($f_2 = 0$) -- 2 factors.}\label{deltas_asset_15_f_1_different_sizes_30_d}
\end{figure}

\begin{figure}[!h]\centering
\includegraphics[width=0.92\textwidth]{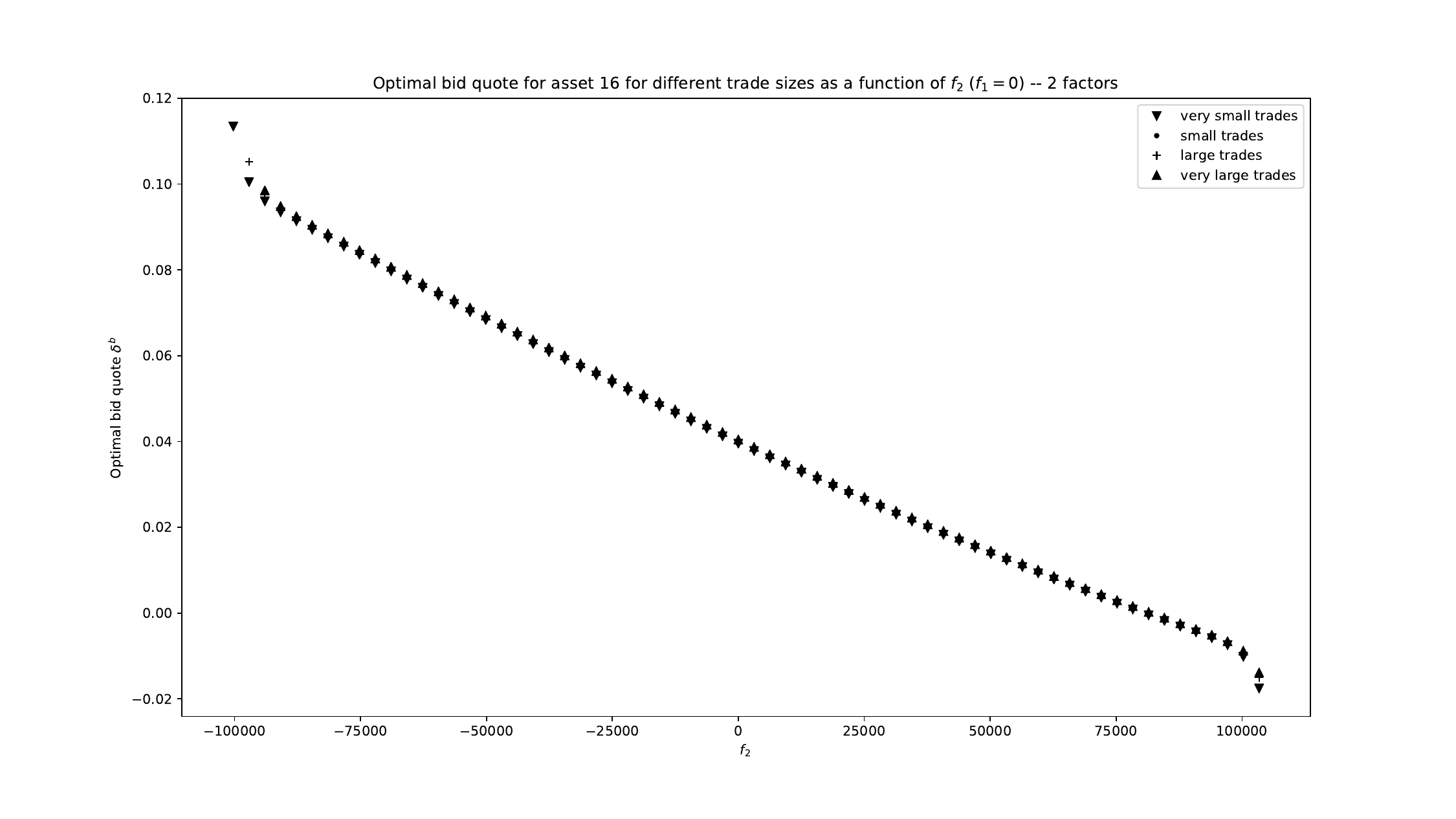}\\
\caption{Optimal bid quote for asset 16 for different trade sizes as a function of $f_2$ ($f_1 = 0$) -- 2 factors.}\label{deltas_asset_15_f_2_different_sizes_30_d}
\end{figure}

We see, especially in Figures \ref{deltas_asset_0_f_2_different_sizes_30_d} and \ref{deltas_asset_15_f_1_different_sizes_30_d}, that the size of the RFQ significantly impacts the quotes that should be answered (as computed with our two-factor approximation).\\

Unlike what we did in the two-asset case, it is impossible in our 30-asset case to know how far from real optimality are the optimal quotes computed with the two-factor approximation. Nevertheless, we can use Monte-Carlo simulations to estimate the value of the objective function associated with the optimal quotes computed with the two-factor approximation in a scenario starting from zero inventory, and compare such an estimation to an approximation of the value function at $(t,q)=(0,0)$ computed through the Monte-Carlo approximation of Section \ref{MCr}.\\

We carried out 2000 trajectories starting from zero inventory, using the optimal quotes computed with the two-factor approximation. These 2000 simulations enable to illustrate the distribution of the PnL at time $T$. The statistics associated with our simulations are documented in Table 4.\\

\begin{table}[h!]
\centering
\begin{tabular}{|c|c|c|c|}
  \hline
  Mean PnL & Stdev PnL & Stdev coming from RFQs & Objective function \\
  \hline
   61471 & 64911 & 5338  & 59765 \\
  \hline
\end{tabular}
\caption{Statistics associated with our 2000 simulations starting from zero inventory (with the two-factor optimal quotes).}
\end{table}

\iffalse
\begin{figure}[!h]\centering
\includegraphics[width=0.98\textwidth]{pnl_distrib_30d2f.pdf}\\
\caption{Histogram of the PnL over 2000 simulations starting from zero inventory (with the two-factor optimal quotes).}\label{pnl_distrib_30d2f}
\end{figure}
\fi
The value of the objective function documented in Table 4 has to be compared to an approximation of the true value function at $(t,q) = (0,0)$. Given the value $\tilde \theta(0, 0) = 60156$ obtained with our numerical scheme and given an estimation of $\eta(0,0)$ equal to $-643$ obtained using the Feynman-Kac representation of Section \ref{MCr} -- with 500 trajectories --, we obtain an approximation of the true value function at $(t,q) = (0,0)$ equal to $59513$. From the very small value of $\frac{\eta(0,0)}{\tilde \theta(0, 0)}$, we deduce that our two-factor approximation is quite satisfactory. The near-optimality of the quotes obtained with our two-factor approximation is confirmed by the value $59765$ obtained with our 2000 trajectories (see Table 4) which is even slightly above $59513$.\\

\section*{Conclusion}

In this paper, we generalized existing market making models to introduce trade size variability. This extension led to an integro-differential equation of the Hamilton-Jacobi type that can be solved using ODE techniques in an infinite-dimensional space. Then, we introduced a numerical method for approximating the optimal bid and ask quotes of a market maker over a large set of assets using a dimensionality reduction technique based on a factor decomposition of the risk. To exemplify our findings, and show that they contribute to beating the curse of dimensionality, we considered two cases of market making with respectively 2 and 30 assets. Our method scales linearly in the number of assets and exponentially in the number of factors, and can therefore be used on large markets driven by a few number of factors.\\

\section*{Appendix: On the construction of the processes $J^{i,b}$ and $J^{i,a}$}

Let us consider a new filtered probability space $\big(\Omega,\mathcal{F},(\mathcal{F}_t)_{t\in \mathbb{R}_{+}},\tilde{\mathbb{P}}\big)$. For the sake of simplicity, assume that $d=1$ and let us omit the superscript $i$ (the generalization is straightforward). Let us introduce $N^b$ and $N^a$ two independent compound Poisson processes of intensity 1 whose increments follow respectively the distributions $\mu^b(dz)$ and $\mu^a(dz)$ with support on $\mathbb{R}_+^*$. We denote by $J^b(dt,dz)$ and $J^a(dt,dz)$ the associated random measures.\\

For each $\delta \in \mathcal{A}$, we introduce the probability measure $\tilde{\mathbb{P}}^\delta$ given by the Radon-Nikodym derivative
\begin{equation}
\frac{d\tilde{\mathbb{P}}^\delta}{d\tilde{\mathbb{P}}} \Big|_{\mathcal{F}_t} = L_t^{\delta},
\end{equation}
where $\left(L_t^{\delta} \right)_{t\geq 0}$ is the unique solution of the stochastic differential equation
\begin{equation}
    dL_t^{\delta} = L_{t-}^{\delta} \left( \int_{\mathbb{R}_+^*} \left(\Lambda^b(\delta^b(t,z))-1 \right)\tilde{J}^b(dt,dz) + \int_{\mathbb{R}_+^*} \left(\Lambda^a(\delta^a(t,z))-1 \right)\tilde{J}^a(dt,dz) \right),\nonumber
\end{equation}
with $L_0^{\delta} = 1$, where $\tilde{J}^b(dt,dz)$ and $\tilde{J}^a(dt,dz)$ stand for the compensated measures.\\

We then know from Girsanov theorem (in its Brémaud-Jacod version) that under $\tilde{\mathbb{P}}^\delta$, $J^b(dt,dz)$ and $J^a(dt,dz)$ have respective intensity kernels
$$\lambda^{\delta,b}_t(dz)  = \Lambda^b(\delta^b(t,z)) \mu^b(dz) \quad \text{and} \quad \lambda^{\delta,a}_t(dz)  = \Lambda^a(\delta^a(t,z)) \mu^a(dz)$$ as in the body of the paper.\\

\section*{Data Availability Statement}

Data sharing is not applicable to this article as no new data were created or analyzed in this study.

\end{document}